\documentclass{article}
\usepackage[utf8]{inputenc}

\usepackage[
backend=biber,
style=alphabetic,
sorting=ynt
]{biblatex}

\usepackage{amsmath}
\usepackage{amsfonts}
\usepackage{amssymb}
\usepackage{amsthm}
\usepackage{thmtools}
\usepackage{mathtools}
\usepackage{mathrsfs}
\usepackage{color}
\usepackage{tikz-cd}
\usepackage{tikz}
\usepackage{comment}
\usepackage{hyperref}
\usepackage{cleveref}
\usepackage{todonotes}
\usepackage[top=3.33cm, left=3.33cm, right=3.33cm, bottom=3.33cm]{geometry}
\usepackage{enumitem}
\usepackage{subcaption}

\theoremstyle{plain}
\newtheorem{thm}{Theorem}[section]
\newtheorem{lem}[thm]{Lemma}
\newtheorem{prop}[thm]{Proposition}

\newtheorem{cor}[thm]{Corollary}

\theoremstyle{definition}
\newtheorem{defn}[thm]{Definition}

\newtheorem{exmp}[thm]{Example}

\newtheorem{construction}[thm]{Construction}
\theoremstyle{remark}
\newtheorem{rem}[thm]{Remark}

\newcommand{\N}{\mathbb{N}}

\DeclarePairedDelimiterX{\norm}[1]{\lVert}{\rVert}{#1}

\newcommand{\enter}{\vspace{\baselineskip}}
\newcommand{\henter}{\vspace{0.5\baselineskip}}
\newcommand{\QFT}{\mathbf{Q}}
\newcommand{\Id}{\operatorname{Id}}
\newcommand{\RQ}{\mathcal{R}_\QFT}
\newcommand{\LQ}{\mathcal{L}_\QFT}
\newcommand{\LQR}{\LQ^\texttt{R}}
\newcommand{\GQ}{\mathcal{G}}
\newcommand{\HQ}{\mathcal{H}}
\newcommand{\AQ}{\mathcal{A}}
\newcommand{\DQ}[1]{\mathcal{D} \left ( #1 \right )}

\newcommand{\combgreen}{\mathfrak{X}^r}
\newcommand{\one}{\mathbb{I}}
\newcommand{\coone}{\hat{\mathbb{I}}}
\newcommand{\CG}{G^\HQ_\AQ}
\newcommand{\pol}{\mathfrak{P}^{(r)}_l}
\newcommand{\cut}{\mathfrak{C}_{\mathfrak{e}, \mathfrak{v}}}
\newcommand{\eval}[1]{#1 \Bigg \vert}
\newcommand{\feyndiag}[2]{\vcenter{\hbox{\includegraphics[width=#2]{#1}}}}
\newcommand{\propagator}{\feyndiag{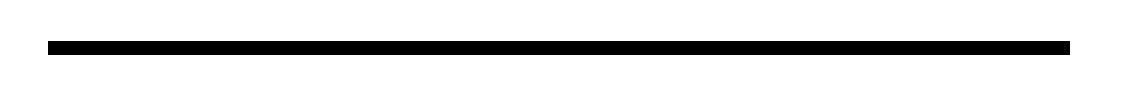}{1.5cm}} 
\newcommand{\propagatorresidue}{\feyndiag{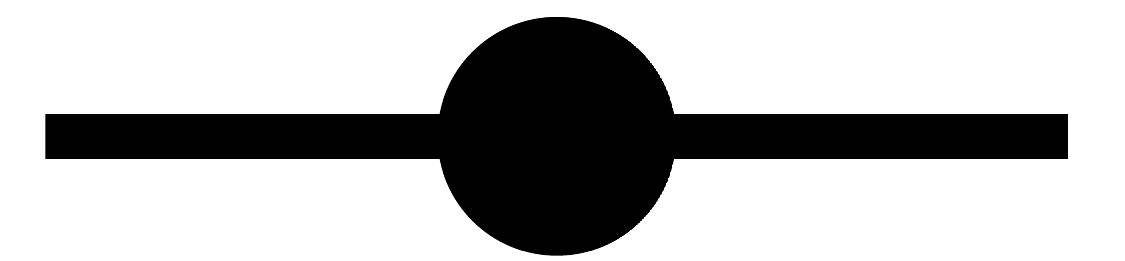}{0.67cm}} 
\newcommand{\vertex}{\feyndiag{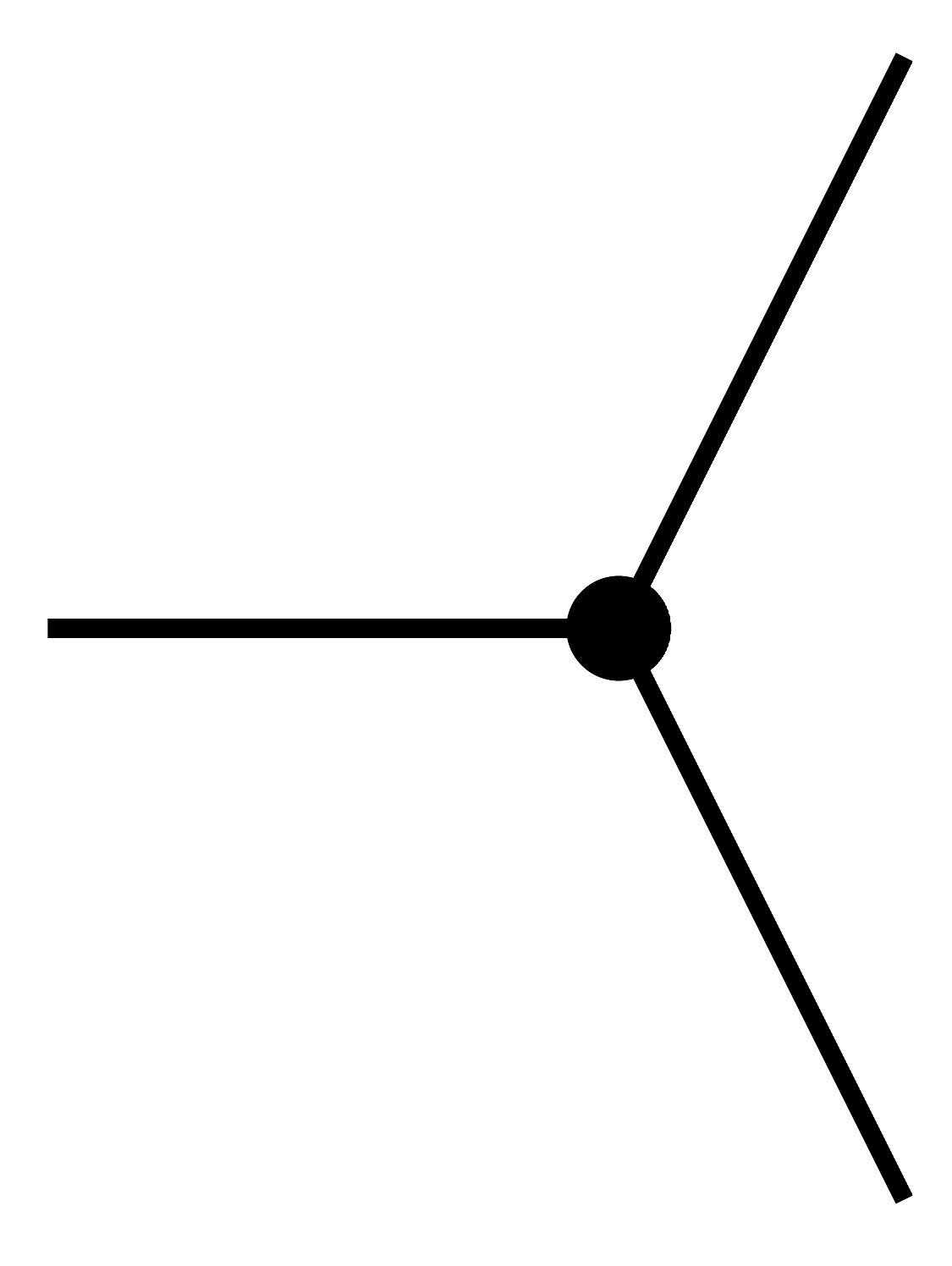}{0.9cm}} 
\newcommand{\vertexresidue}{\feyndiag{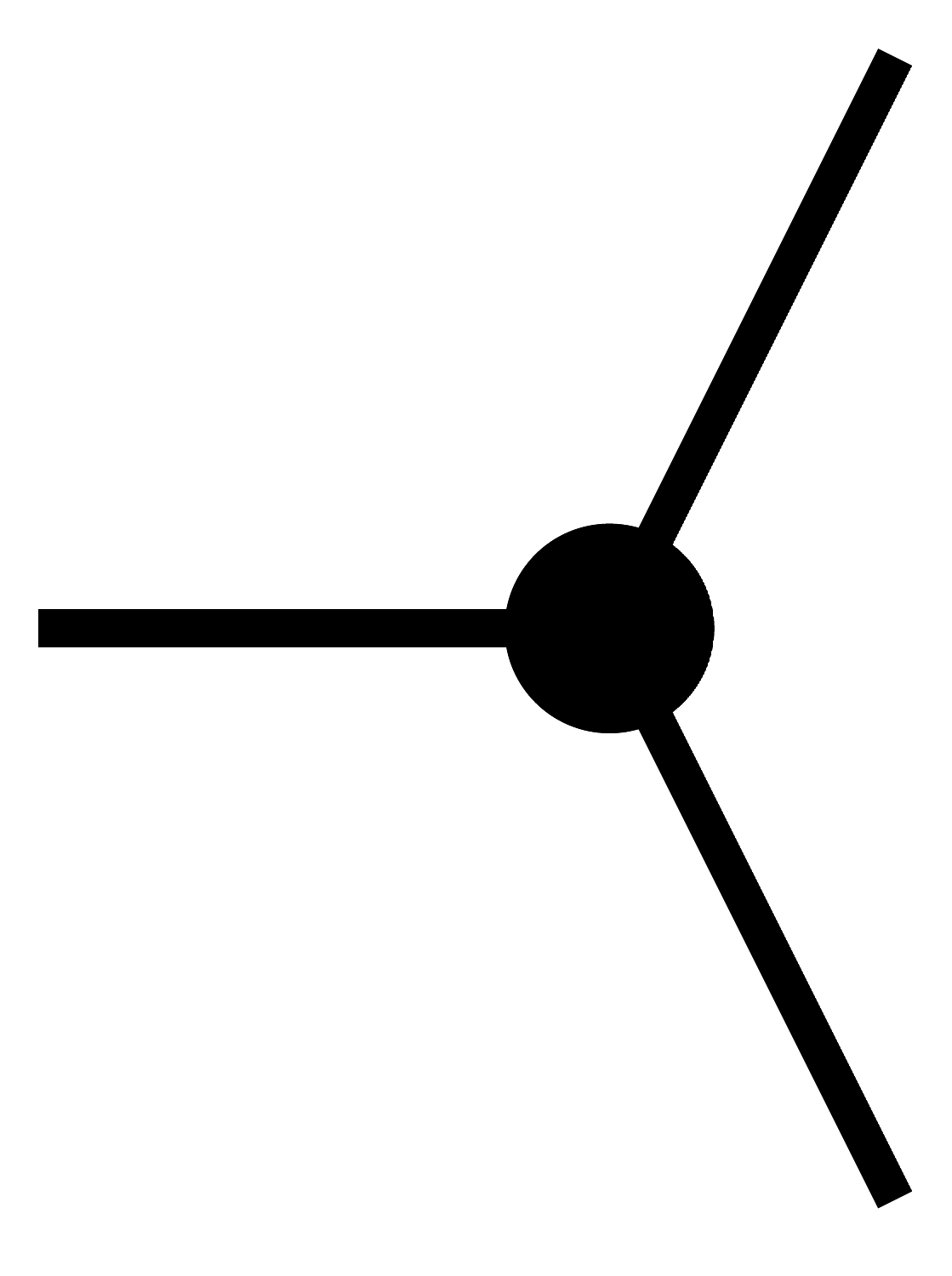}{0.67cm}} 
\newcommand{\bubblegraph}{\feyndiag{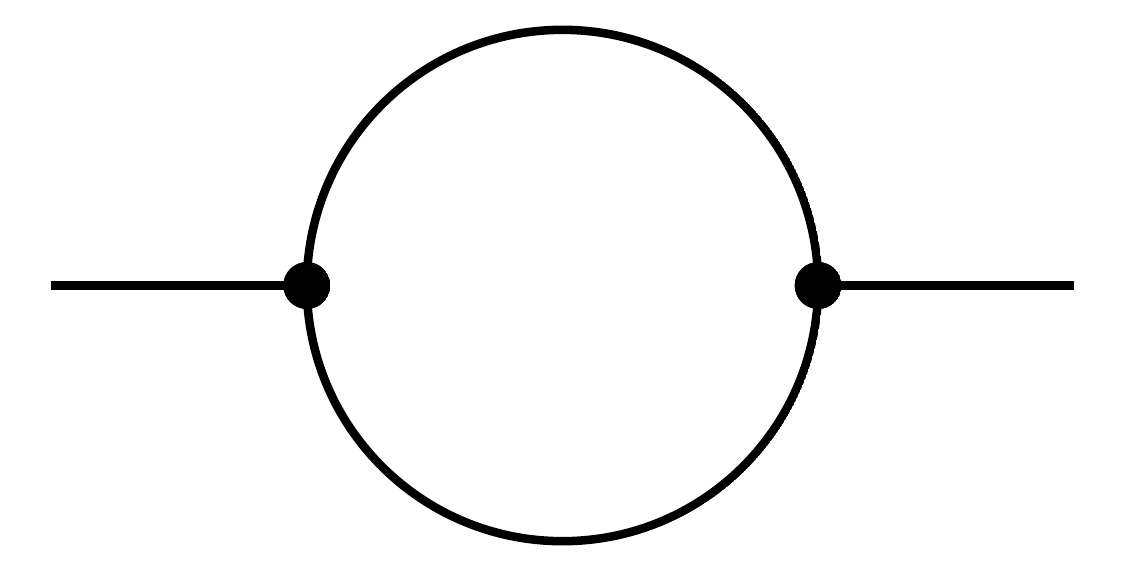}{3cm}} 
\newcommand{\bubblevertexcorrectiongraph}{\feyndiag{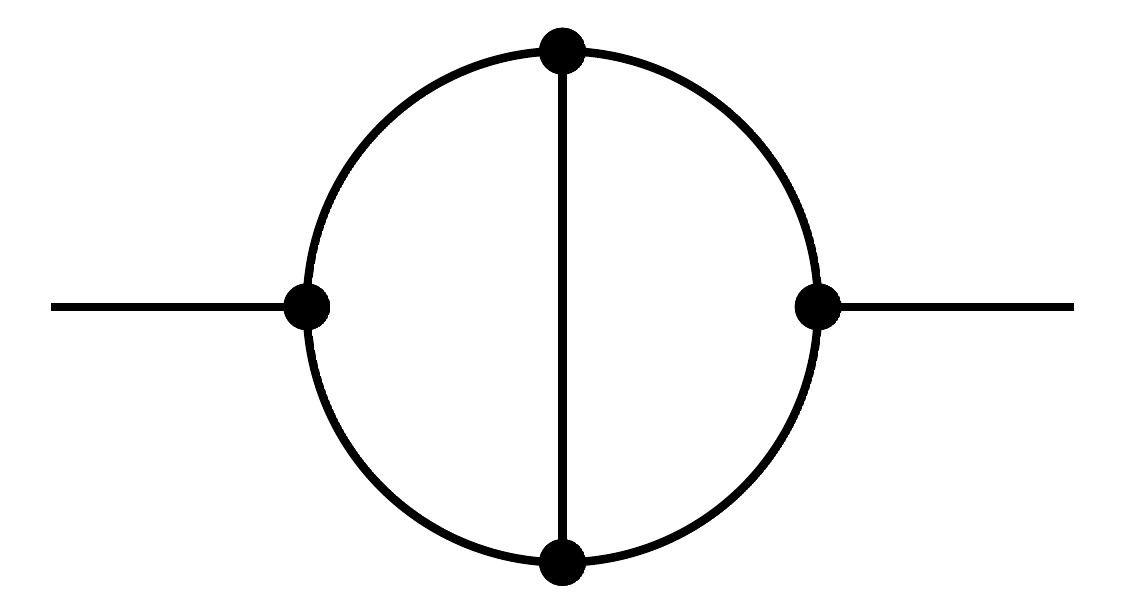}{3cm}} 
\newcommand{\bubblevertexcorrectiongraphdeletedleft}{\feyndiag{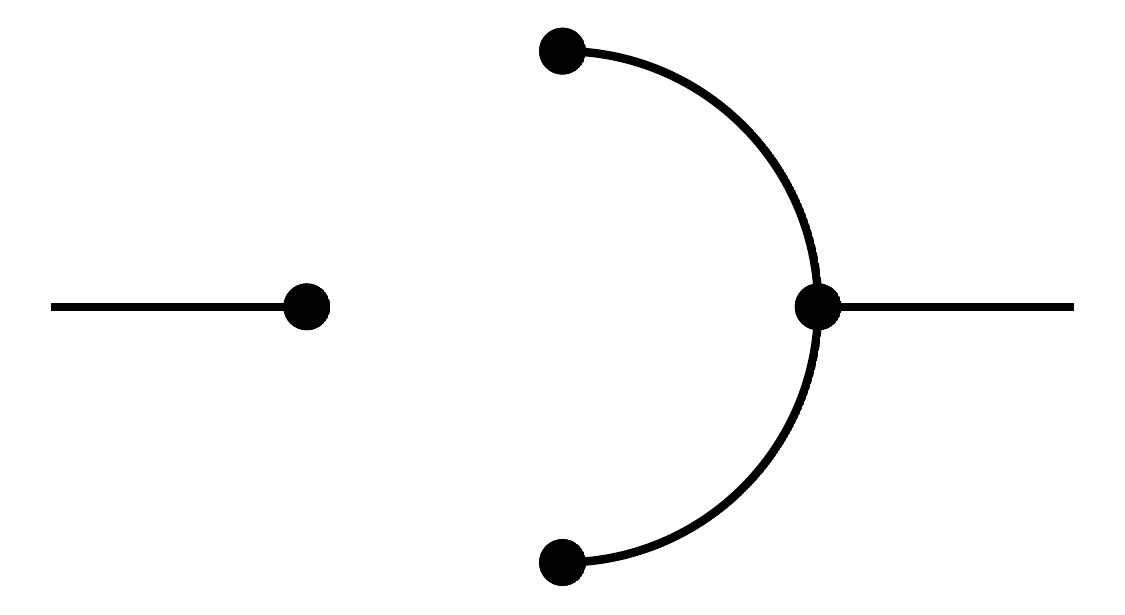}{3cm}} 
\newcommand{\bubblevertexcorrectiongraphdeletedright}{\feyndiag{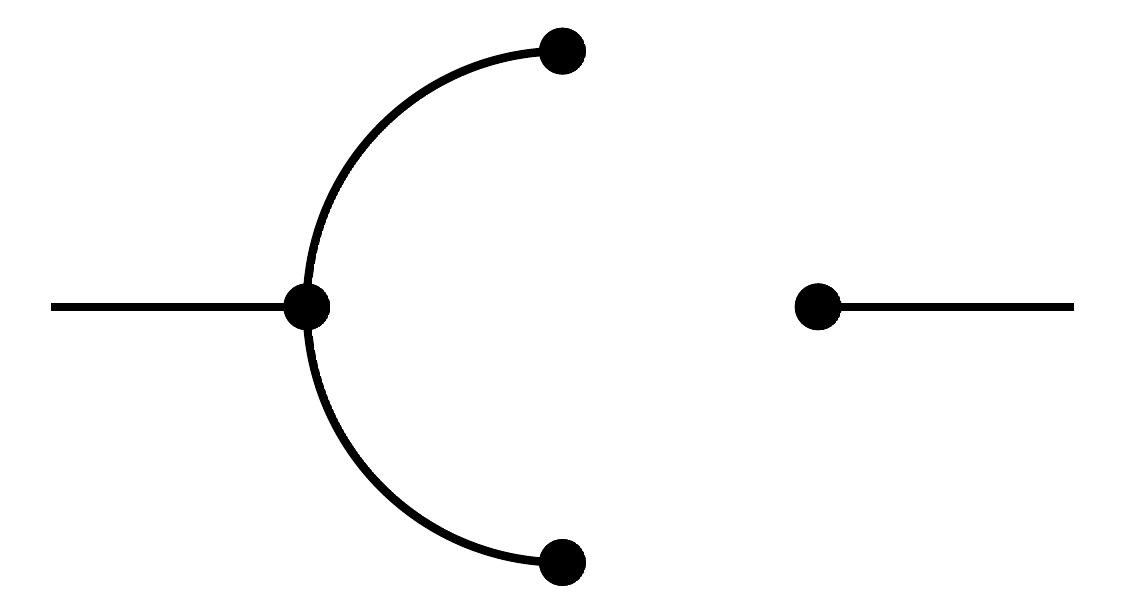}{3cm}} 
\newcommand{\bubbleupperedgecorrectiongraph}{\raisebox{0.1cm}{$\feyndiag{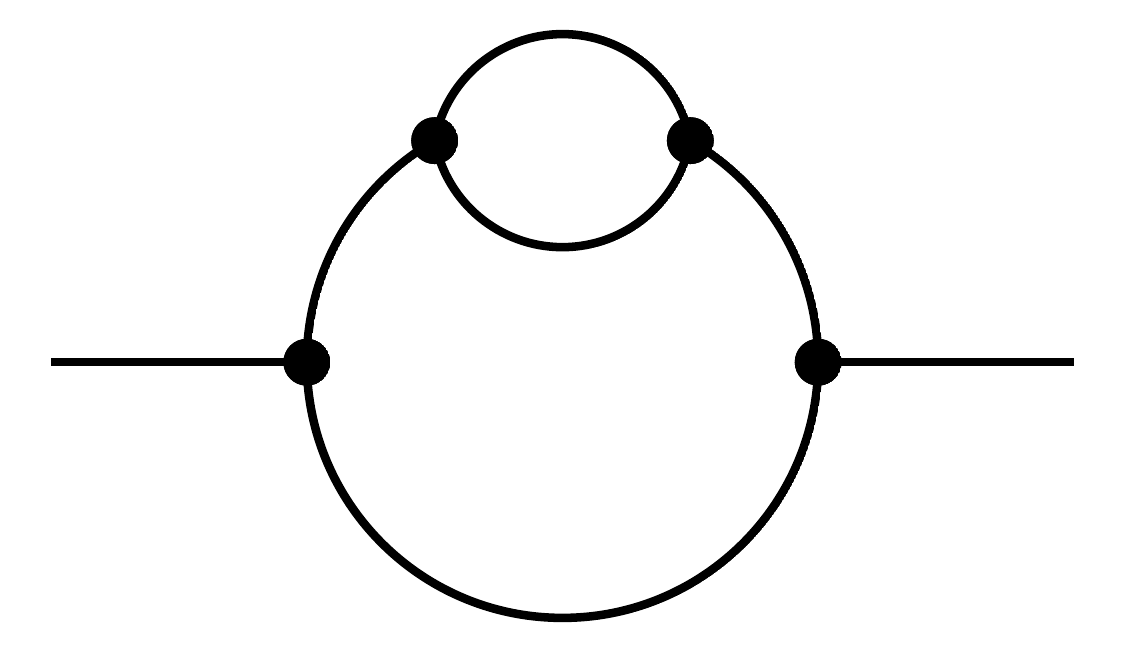}{3cm}$}} 
\newcommand{\bubbleupperedgecorrectiondeletedgraph}{\raisebox{0.1cm}{$\feyndiag{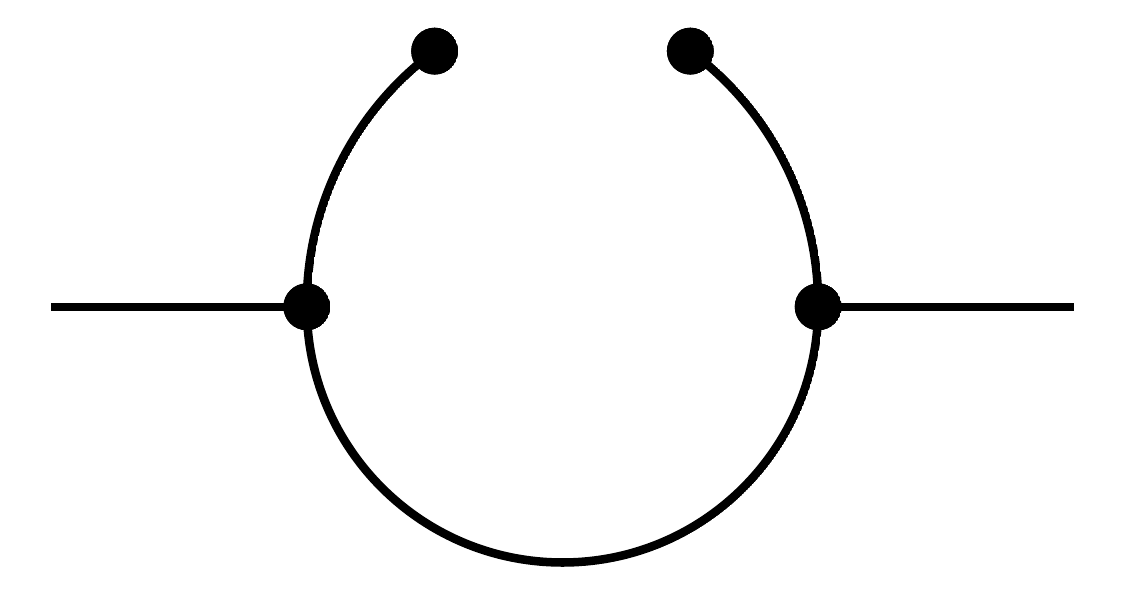}{3cm}$}} 
\newcommand{\bubbleloweredgecorrectiongraph}{\raisebox{-0.1cm}{$\feyndiag{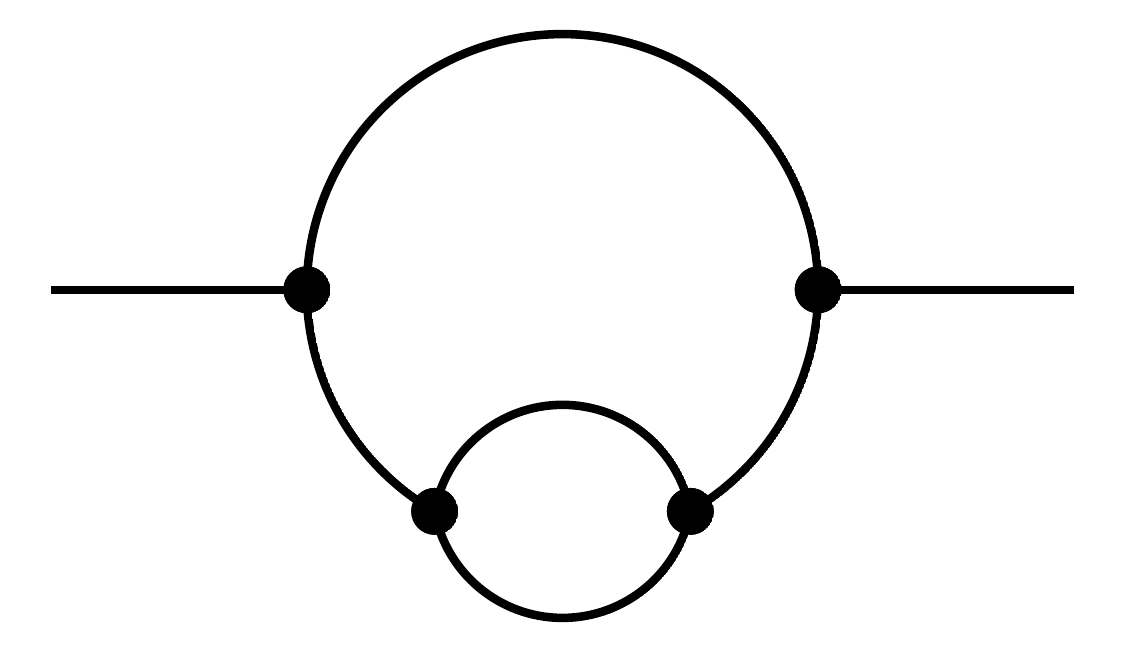}{3cm}$}} 
\newcommand{\bubbleupperedgecorrectiongraphsmall}{\raisebox{0.067cm}{$\feyndiag{bubble-upper-edge-correction}{2cm}$}} 
\newcommand{\bubbleloweredgecorrectiongraphsmall}{\raisebox{-0.067cm}{$\feyndiag{bubble-lower-edge-correction}{2cm}$}} 
\newcommand{\trianglegraph}{\feyndiag{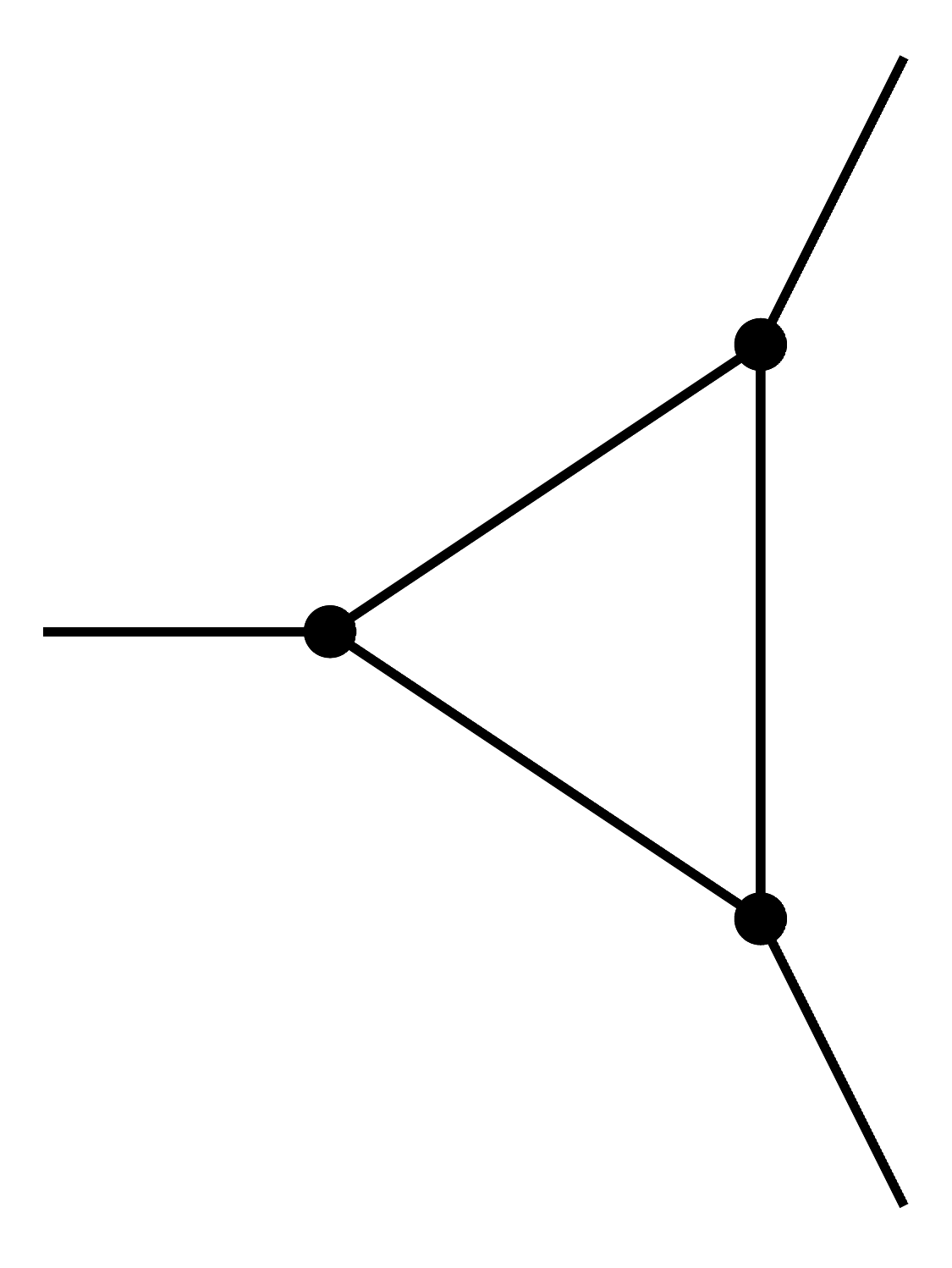}{2.25cm}} 
\newcommand{\trianglegraphmirrored}{\feyndiag{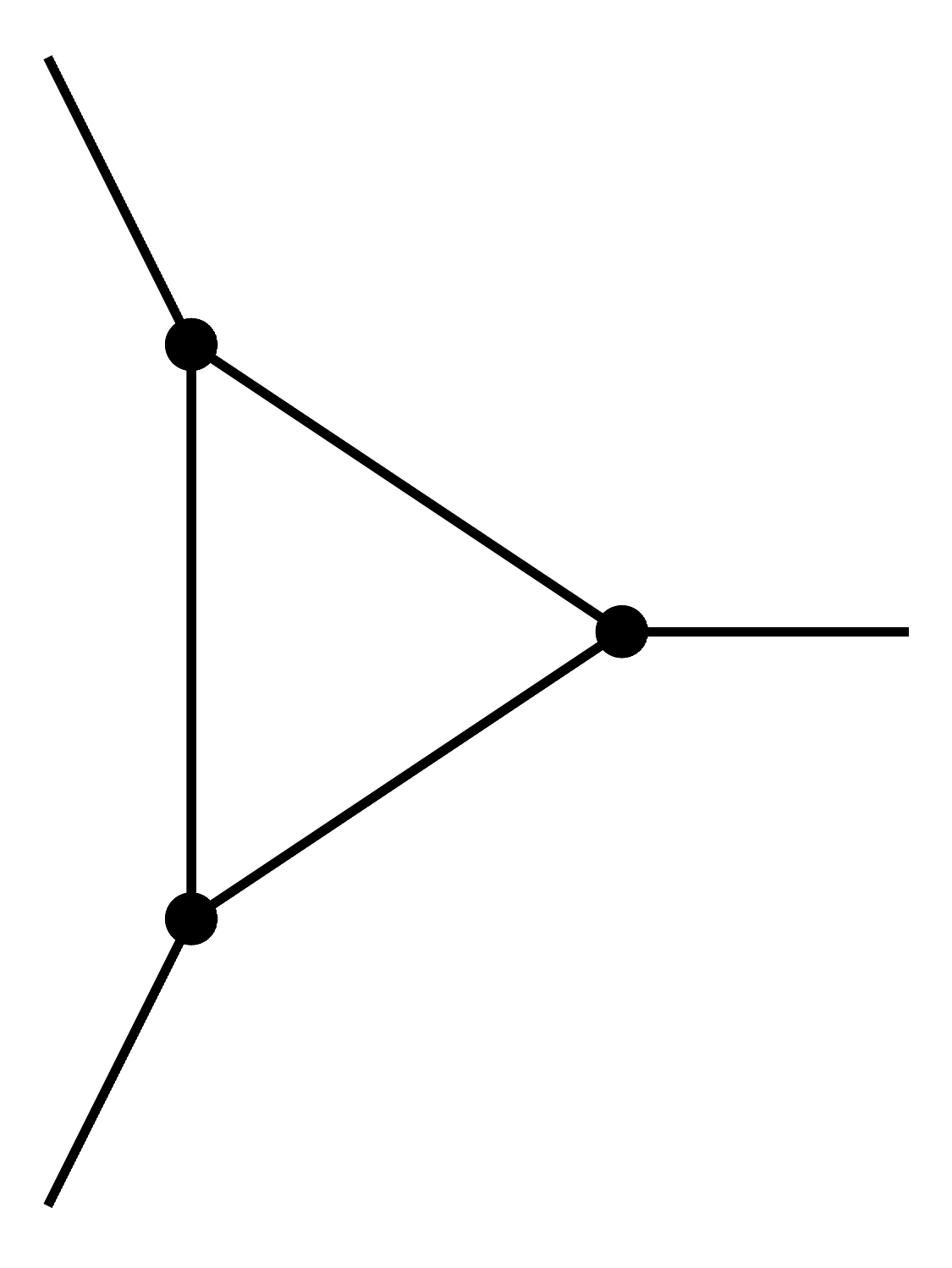}{2.25cm}} 
\newcommand{\triangleleftvertexcorrectiongraph}{\feyndiag{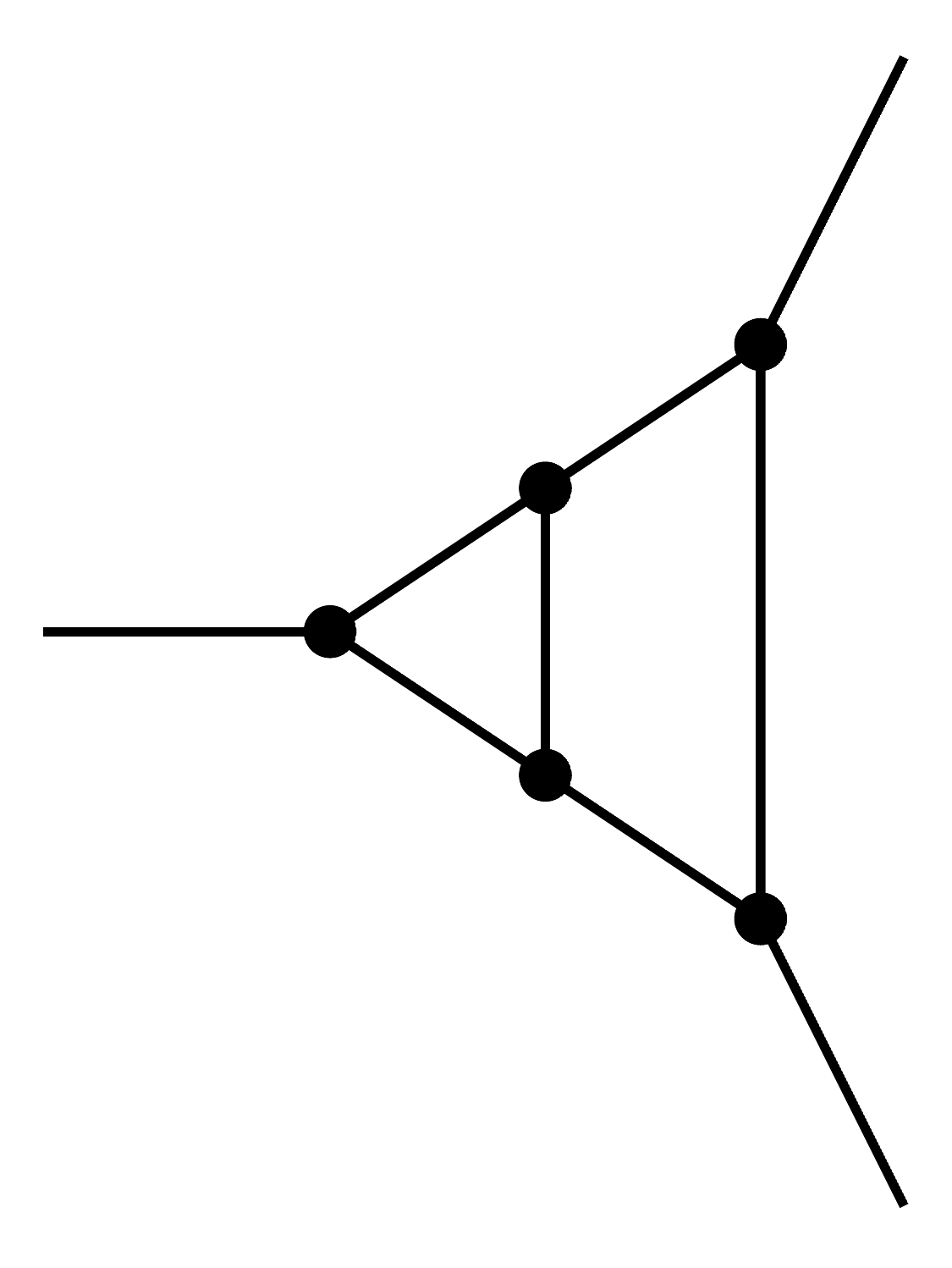}{2.25cm}} 
\newcommand{\triangleleftvertexcorrectiondeletedgraph}{\feyndiag{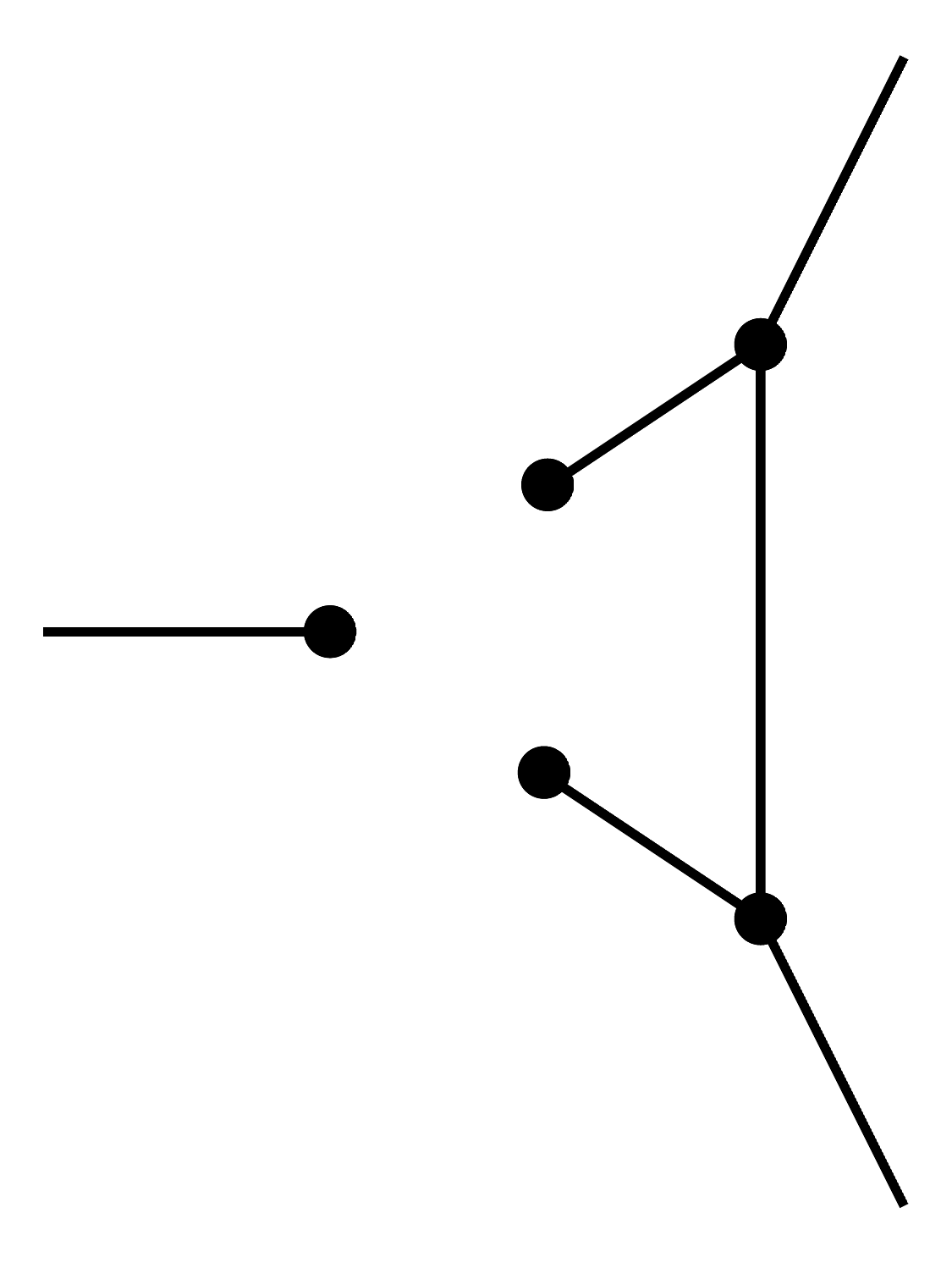}{2.25cm}} 
\newcommand{\triangleuppervertexcorrectiongraph}{\feyndiag{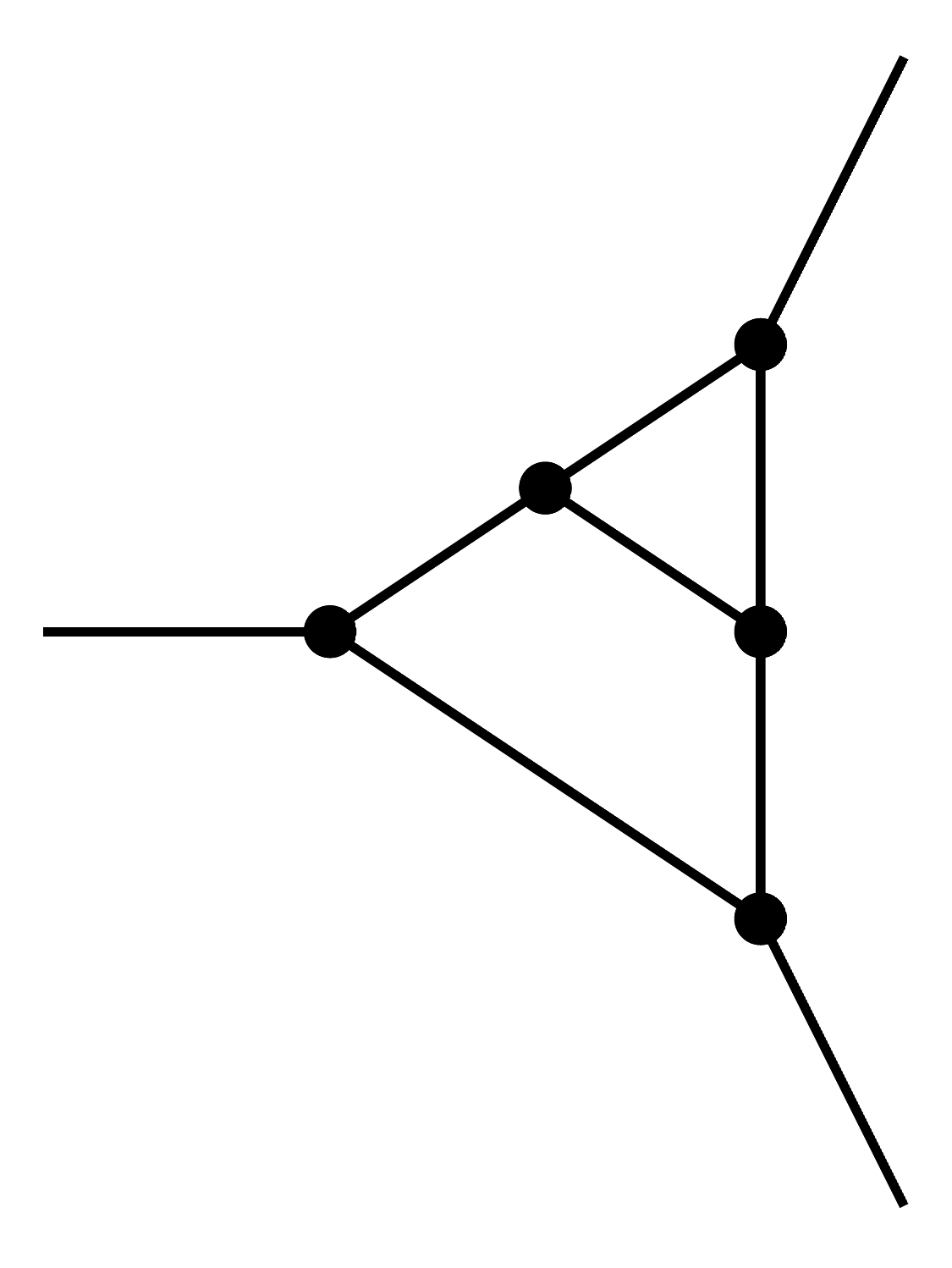}{2.25cm}} 
\newcommand{\trianglelowervertexcorrectiongraph}{\feyndiag{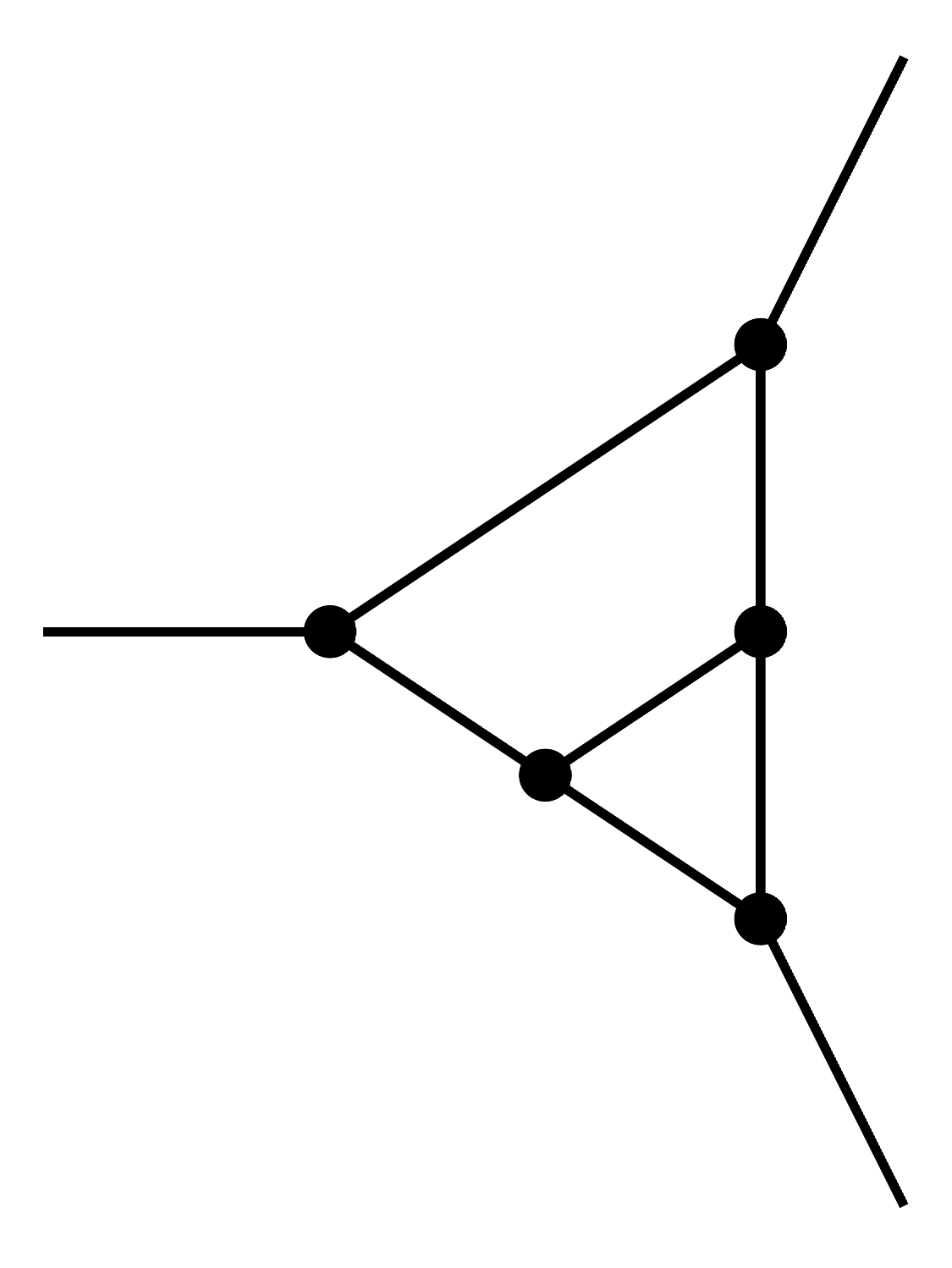}{2.25cm}} 
\newcommand{\trianglerightedgecorrectiongraph}{\feyndiag{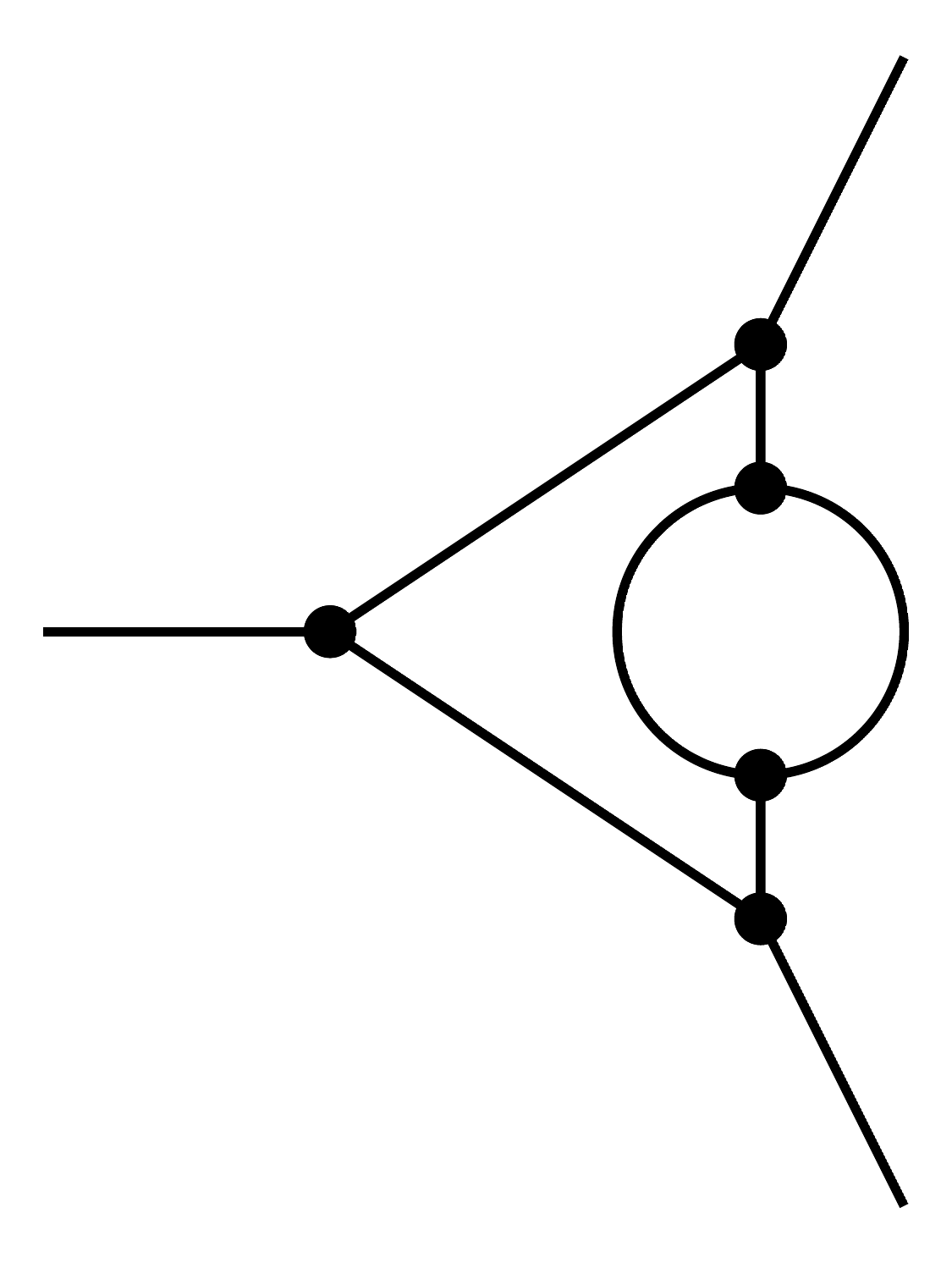}{2.25cm}} 
\newcommand{\trianglerightedgecorrectiondeletedgraph}{\feyndiag{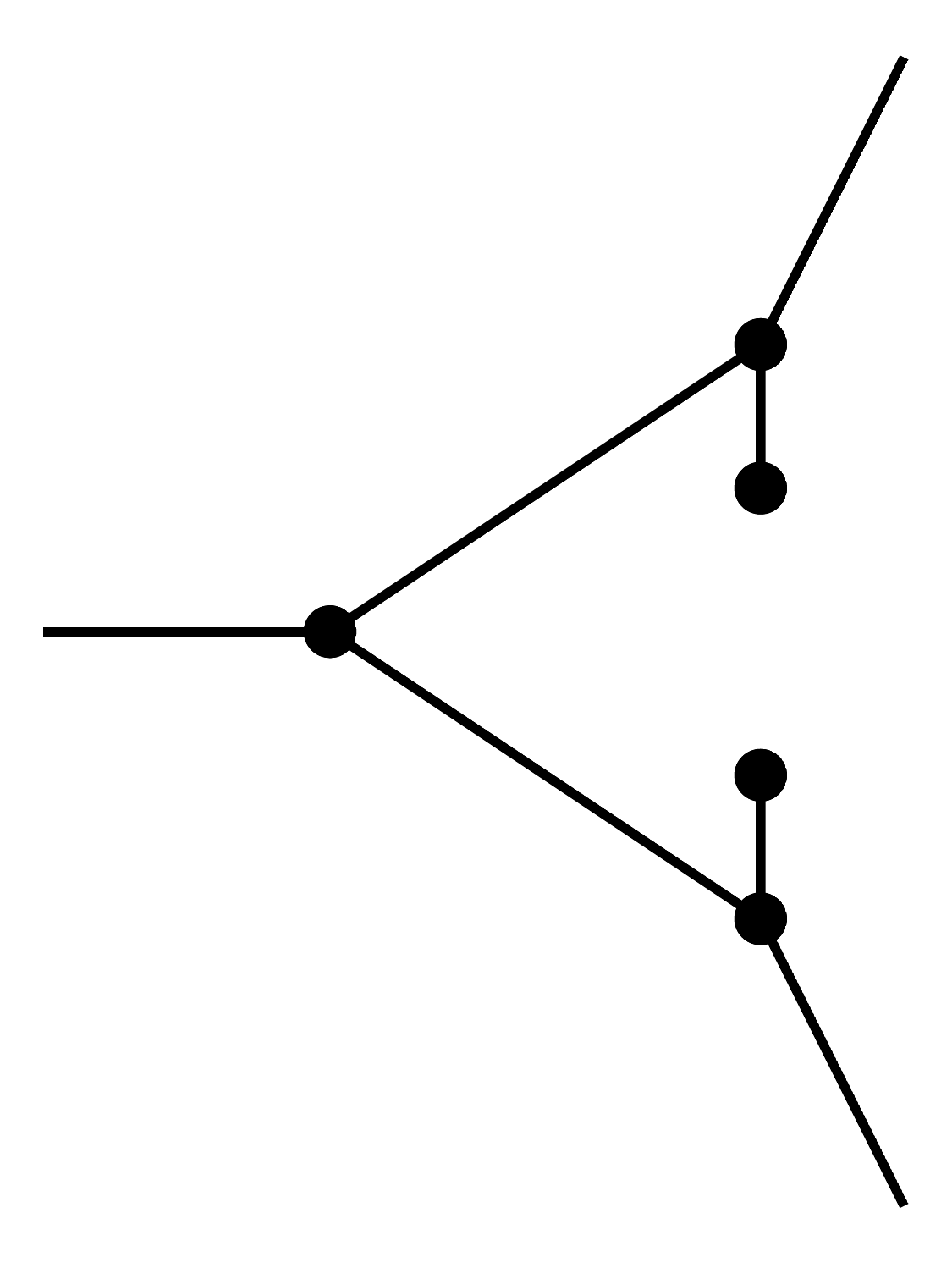}{2.25cm}} 
\newcommand{\triangleupperedgecorrectiongraph}{\feyndiag{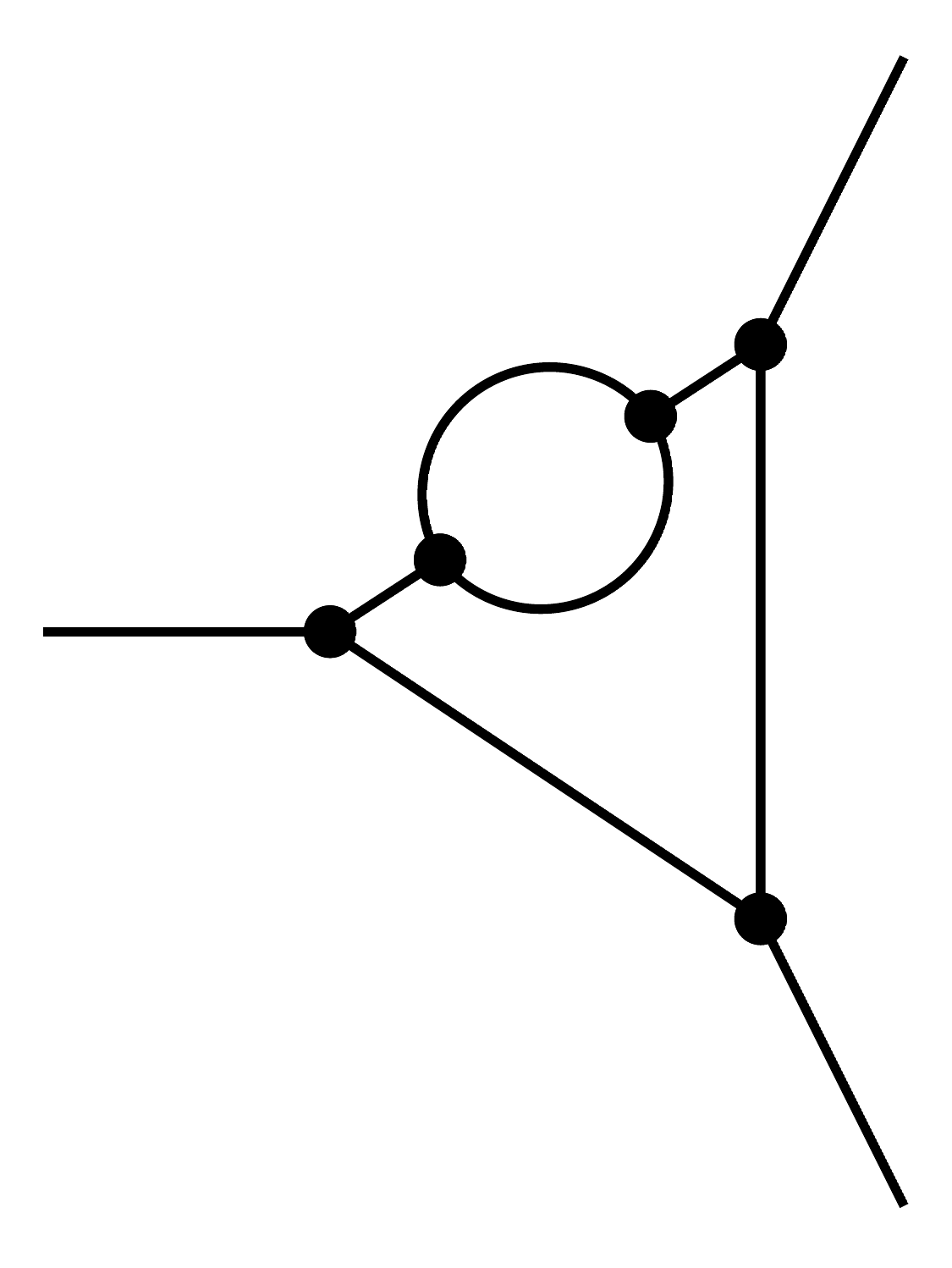}{2.25cm}} 
\newcommand{\triangleloweredgecorrectiongraph}{\feyndiag{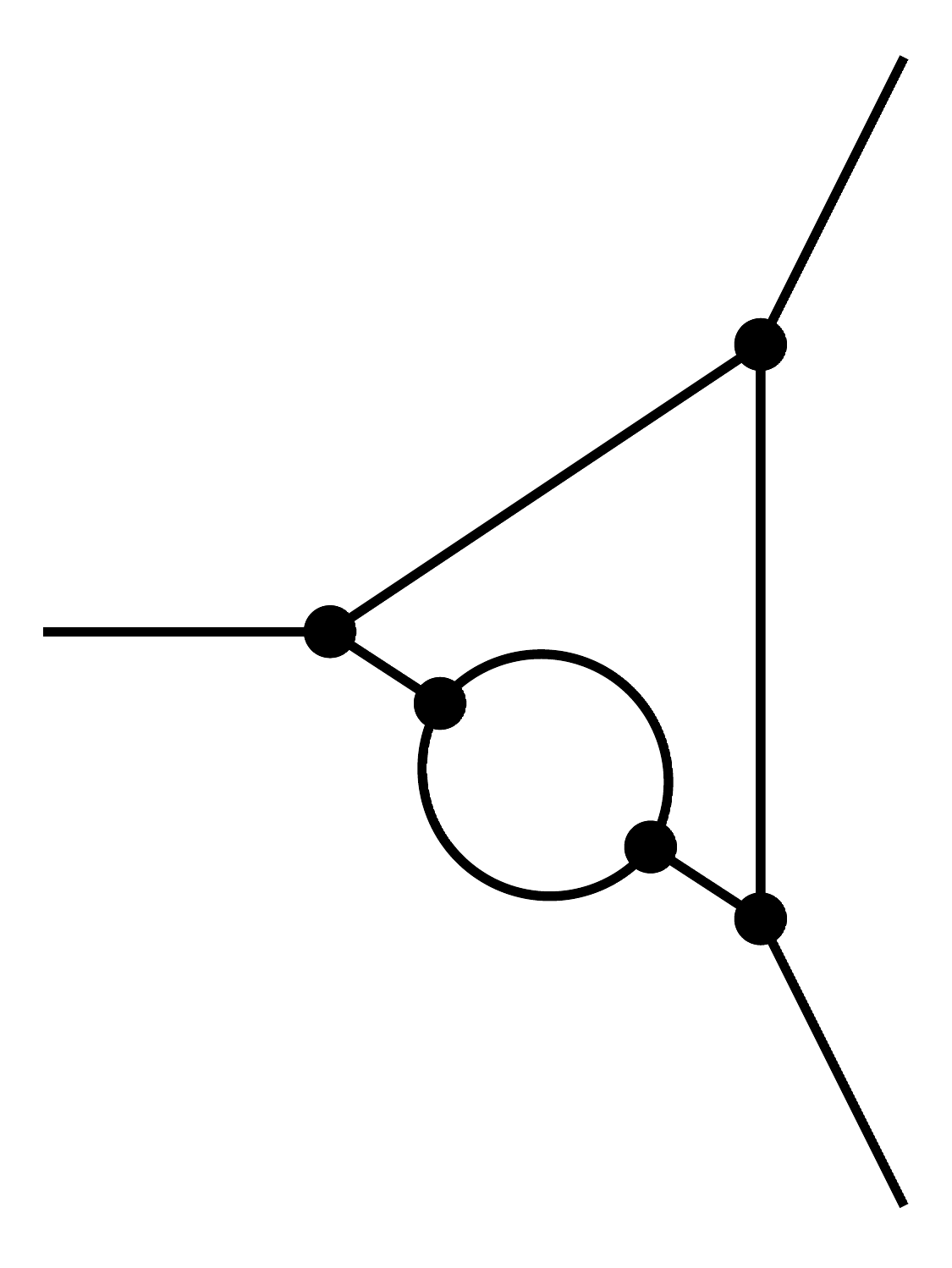}{2.25cm}} 
\newcommand{\trianglenonplanar}{\feyndiag{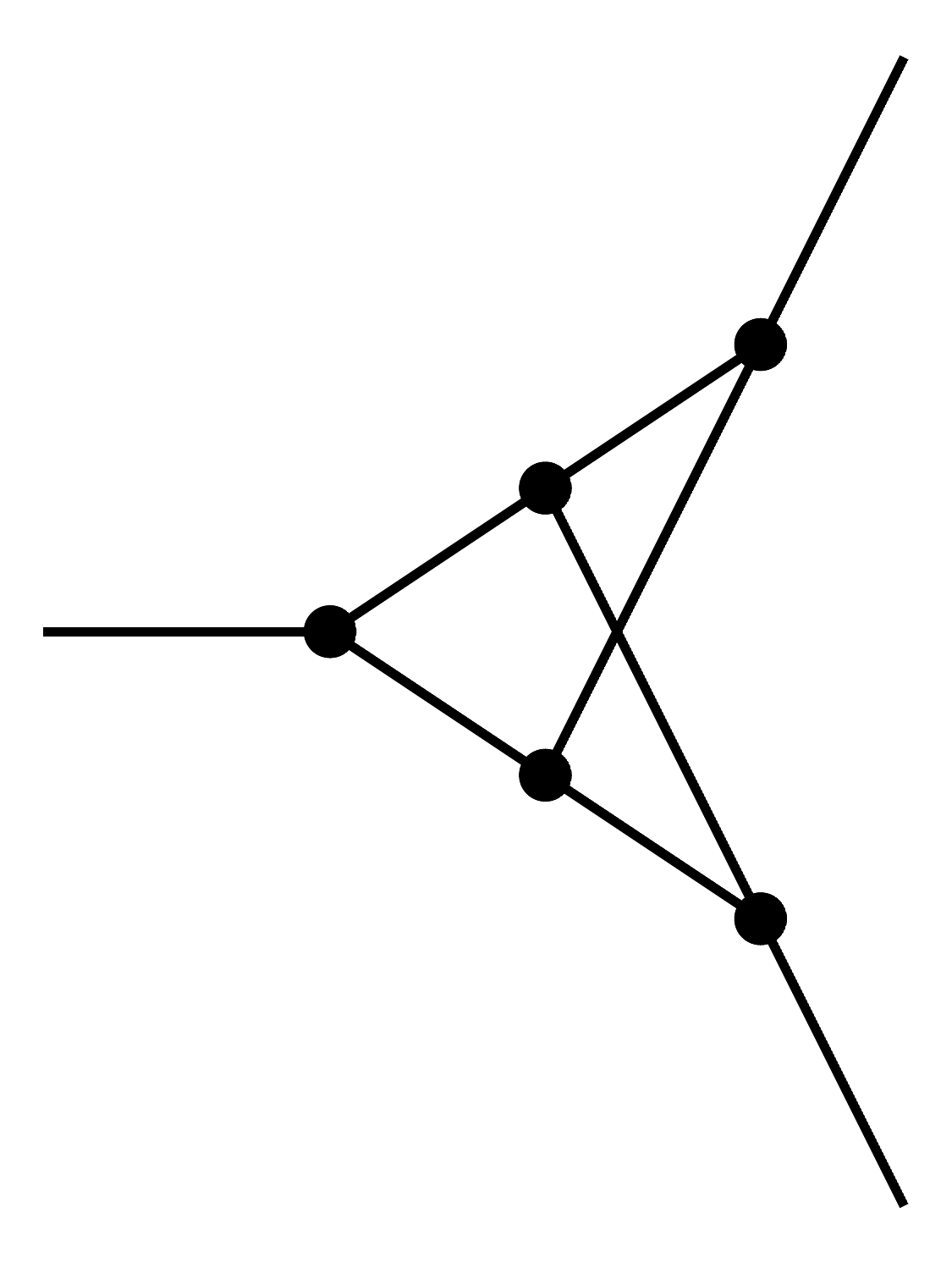}{2.25cm}} 
\newcommand{\insertionproduct}[2]{\delta_{V \left ( #1 \right ) \hookrightarrow V \left ( #2 \right )}}
\newcommand{\mrFR}{\Phi^\texttt{R}}
\newcommand{\renFR}{\Phi^+}
\newcommand{\CTmap}{\Phi^-}

\newcommand{\A}{\mathcal{A}} 
\DeclareMathOperator{\supp}{supp} 
\DeclareMathOperator{\SDD}{SDD} 

\DeclareMathOperator{\Ren}{\mathfrak{R}}
\newcommand{\D}[1][\,]{#1\mathrm{d}}
\DeclareMathOperator{\Hg}{\mathsf{Hg}}
\DeclareMathOperator{\Phg}{\mathsf{PHg}}
\DeclareMathOperator{\Loghg}{\mathsf{LogHg}}

\title{\makebox[\textwidth][c]{Multiplicative Renormalization in Causal Perturbation Theory}}
\author{Jonah Epstein\footnote{University of Bonn and Max Planck Institute for Mathematics, Bonn; e-mail: \href{mailto:jepstein@uni-bonn.de}{jepstein@uni-bonn.de}}, Arne Hofmann\footnote{Leibniz University Hannover and University of Göttingen; e-mail: \href{mailto:arne.hofmann@math.uni-hannover.de}{arne.hofmann@math.uni-hannover.de}} \kern0.2em and David Prinz\footnote{Max Planck Institute for Mathematics, Bonn; e-mail: \href{mailto:prinz@mpim-bonn.mpg.de}{prinz@mpim-bonn.mpg.de}}}
\date{December 10, 2025}

\begin{document}

\maketitle

\begin{abstract}
	We construct multiplicative renormalization for the Epstein--Glaser renormalization scheme in perturbative Algebraic Quantum Field Theory: To this end, we fully combine the Connes--Kreimer renormalization framework with the Epstein--Glaser renormalization scheme. In particular, in addition to the already established position-space renormalization Hopf algebra, we also construct the renormalized Feynman rules and the counterterm map via an algebraic Birkhoff decomposition. This includes a discussion about the appropriate target algebra of regularized distributions and the renormalization scheme as a Rota--Baxter operator thereon. In particular, we show that the Hadamard singular part satisfies the Rota--Baxter property and thus relate factorization in Epstein--Glaser with multiplicativity in Connes--Kreimer. Next, we define \(Z\)-factors as the images of the counterterm map under the corresponding combinatorial Green's functions. This allows us to define the multiplicatively renormalized Lagrange density, for which we show that the corresponding Feynman rules are regular. Finally, we exemplify the developed theory by working out the specific case of \(\phi^3_6\)-theory.
\end{abstract}

\tableofcontents

\section{Introduction} \label{sec:introduction}

Renormalization is a fundamental part of Quantum Field Theory (QFT): The basic idea is to absorb the divergences that appear in the perturbative expansion by making the constants of the theory energy-dependent.\footnote{We deal in this article with UV-divergences only, thus \emph{divergence} means \emph{UV-divergence}.} More precisely, consider the example of \(\phi^3_6\)-theory --- a scalar theory with cubic interaction in six dimensions of spacetime (such that it is renormalizable) --- given via the following Lagrange density:
\begin{equation} \label{eqn:unrenormalized-phi-3-6-theory-intro}
    \mathcal{L}_{\phi^3_6} = \frac{1}{2} \left ( \partial_\mu \phi_0 \right ) \left ( \partial^\mu \phi_0 \right ) - \frac{m_0}{2} \phi_0^2 - \frac{\lambda_0}{3!} \phi_0^3
\end{equation}
Here, the zero subscripts denote the so-called \emph{bare}, i.e.\ unrenormalized, quantities. If we were to derive Feynman rules from \(\mathcal{L}_{\phi^3_6}\), the resulting Feynman amplitudes would be singular.\footnote{More precisely, the propagator Feynman diagrams would be quadratically divergent and the three-point graphs would be logarithmically divergent.} We are now interested in absorbing these divergences multiplicatively, i.e.\ without changing the monomials of the theory. To this end, we introduce the so-called wave-function, mass and coupling constant renormalization, by introducing the so-called \(Z\)-factors:
\begin{align}
    \phi_0 & \rightsquigarrow \phi := \frac{1}{\sqrt{Z_\phi}} \phi_0 \, , \\
    m_0 & \rightsquigarrow m := \frac{1}{\sqrt{Z_m}} m_0
    \intertext{and}
    \lambda_0 & \rightsquigarrow \lambda := \frac{1}{\sqrt{Z_\lambda}} \lambda_0 \, .
\end{align}
Applying these rescalings to the Lagrange density, we obtain the \emph{multiplicatively renormalized Lagrange density}
\begin{equation}
\begin{split}
    \mathcal{L}_{\phi^3_6}^\texttt{R} & = \frac{Z_\phi}{2} \left ( \partial_\mu \phi \right ) \left ( \partial^\mu \phi \right ) - \frac{Z_\phi Z_m m^2}{2} \phi^2 - \frac{Z_\phi^{3/2} Z_\lambda \lambda}{3!} \phi^3 \\
    & = \frac{Z_{\text{Kin}}}{2} \left ( \partial_\mu \phi \right ) \left ( \partial^\mu \phi \right ) - \frac{Z_{\text{Mass}} m}{2} \phi^2 - \frac{Z_{\text{Int}} \lambda}{3!} \phi^3 \, ,
\end{split}
\end{equation}
where we have set
\begin{subequations}
\begin{align}
    Z_{\text{Kin}} & \coloneq Z_\phi \, , \\
    Z_{\text{Mass}} & \coloneq Z_\phi Z_m
    \intertext{and}
    Z_{\text{Int}} & \coloneq Z_\phi^{3/2} Z_\lambda \, .
\end{align}
\end{subequations}
Now, in order to absorb the singularities of the perturbative expansion, the idea is to define these \(Z\)-factors as the singular contributions of the corresponding Feynman diagram expansion: More precisely, the \(Z\)-factors will be an asymptotic series of the form
\begin{equation}
    Z^r \coloneq 1 \pm \sum_{k = 1}^\infty C^r_k \, ,
\end{equation}
where \(r\) denotes the amplitude (for \(\phi^3_6\)-theory either propagator graphs or three-point graphs), the sign on propagator or interaction graphs and \(C^r_k\) the corresponding sum over the singular contributions of all respective Feynman graphs at loop order \(k\). Given this setup, the Feynman rules derived from the multiplicatively renormalized Lagrange density \(\mathcal{L}_{\phi^3_6}^\texttt{R}\) are given as the unrenormalized ones times the respective \(Z\)-factor: Notably, this leads to combinatorial cancellations of all appearing singularities, as we will exemplify in \Cref{thm:multiplicative-renormalization}.

\enter

In this article, we specifically focus on the situation of Epstein--Glaser renormalization for position space Feynman integrals in Algebraic Quantum Field Theory (AQFT), as introduced in \cite{epsteinRoleLocalityPerturbation1973}. Specifically, we aim to unify said setup with the Connes--Kreimer renormalization framework \cite{Kreimer:1997dp,Connes:1998qv,Connes:1999zw}: The current status of the literature only covers the setup of the renormalization Hopf algebra, but misses the central construction of the algebraic Birkhoff decomposition, which produces the renormalized Feynman rules map and the counterterm map, cf.\ \cite{gracia-bondiaCONNESKREIMEREPSTEIN,pinterHopfAlgebraStructure,figueroaUSESCONNESKREIMER2004,Lange:2004mg,bergbauerHopfAlgebraRooted2005,Berghoff:2014xba}. Crucially, we aim for a formulation which clearly separates the analytical input from the combinatorial and algebraic part.

\enter

In particular, the Connes--Kreimer renormalization framework operates as follows: Starting from the set of \emph{one-particle irreducible (1PI)}\footnote{This means graphs which remain connected after the removal of any of its internal edges, sometimes also called \emph{bridge-free}.} Feynman graphs \(\mathcal{G}\), a Hopf algebra \(\mathcal{H} \coloneq \mathbb{Q}[[\mathcal{G}]]\) is constructed as follows \cite{Kreimer:1997dp}:\footnote{Originally, the renormalization Hopf algebra was constructed as a polynomial algebra in decorated rooted trees, i.e.\ \(\HQ \coloneq \mathbb{Q}[\mathcal{T}]\) with \(\mathcal{T}\) the set of rooted trees.} The multiplication is given as disjoint union of Feynman graphs with the empty graph as identity and the coproduct --- the essential ingredient for renormalization --- is constructed to encode the nesting of subdivergences: Given a Feynman graph \(\Gamma \in \mathcal{G}\) together with its set of superficially divergent subgraphs \(\DQ{\Gamma}\), we define:
\begin{equation}
    \Delta \left ( \Gamma \right ) \coloneq \sum_{\gamma \in \DQ{\Gamma}} \delta_{\left ( \gamma \hookrightarrow \Gamma \setminus \gamma \right )} \gamma \otimes_\mathbb{Q} \Gamma \setminus \gamma
\end{equation}
Here, \(\Gamma \setminus \gamma\) denotes the graph \(\Gamma\) with all edges and internal vertices of the subgraph \(\gamma\) deleted --- a particular variation of the usual formulation tailored for the position space approach. In addition, \(\delta_{\left ( \gamma \hookrightarrow \Gamma \setminus \gamma \right )}\) identifies the external vertices of \(\gamma\) with their internal copies in \(\Gamma \setminus \gamma\). Morally speaking, this decomposition now allows the application of the counterterm map on the left hand side of the tensor product together with the ordinary Feynman rules map on the right hand side: Thus, subtracting this from the ordinary Feynman rules map of the whole graph results in a Feynman graph with this particular subdivergence resolved. Notably, iterating this procedure results in a renormalized and local amplitude. More specifically, after setting up an appropriate partial algebra of distributions \(\mathcal{A}\), we can view Feynman rules as \emph{characters}, i.e.\ algebra maps, \(\Phi \colon \mathcal{H} \to \mathcal{A}\). Crucially, the set of algebra maps from \(\mathcal{H}\) to \(\mathcal{A}\) can be turned into a group \(G^\mathcal{H}_\mathcal{A}\), the so-called \emph{character group}, by introducing the \emph{convolution product} \(f_1 \star f_2 \coloneq m_\mathcal{A} \circ \left ( f_1 \otimes_\mathbb{Q} f_2 \right ) \circ \Delta_\mathcal{H}\). Now, the algebraic Birkhoff decomposition --- the main theorem of Connes--Kreimer renormalization \cite[Theorem 1]{Connes:1999zw} --- is the statement that, provided a splitting of the target algebra \(\mathcal{A} \cong \mathcal{A}_+ \oplus \mathcal{A}_-\) into its regular and singular parts, each character \(f \in G^\mathcal{H}_\mathcal{A}\) can be uniquely decomposed as the convolution product \(f \equiv {f^-}^{\star -1} \star f^+ \iff f^+ \equiv f^- \star f\), where the \(f^\pm \in G^\mathcal{H}_\mathcal{A}\) are two characters with respective images in \(\mathcal{A}_\pm\). Ultimately, we want the splitting of \(\mathcal{A}\) as a vector space into \(\mathcal{A}_\pm\) to be induced via a \emph{renormalization scheme} \(\mathscr{R} \colon \mathcal{A} \twoheadrightarrow \mathcal{A}_-\), cf.\ \Cref{defn:renormalization-scheme}. More precisely, in this article we study the specific case of \(\mathscr{R}\) being the Epstein--Glaser scheme.

\enter

The Epstein--Glaser approach is a position space renormalization scheme in which a \emph{locality} axiom (conditions 1 and 2 in \Cref{def:renorm-maps}) derived from physical causality requirements is the most important requirement.
The analytical component of the procedure is encoded in the problem of extending a distribution $ u\in \mathscr{D}'(M\backslash \Sigma)$, where $\Sigma$ is a closed subset of $M$, to a distribution $\overline{u} \in \mathscr{D}'(M)$.
Formulating the construction in terms of these extension maps means that one deals only with finite, well defined quantities, and no ``subtractions of infinities'' occur.
This is manifest in \Cref{construct:ren-maps}.

Of course, singular limits and counterterms do not vanish completely --- they enter in the construction of the extension maps.
The Epstein--Glaser scheme is to some extent agnostic as to the precise method of constructing these maps, so long as they satisfy all the requirements.
However, at least since \cite{brunettiMicrolocalAnalysisInteracting2000} the most common method has been to regularize the distributions involve by cutting out singularities that occur at the points where two or more configuration points coincide, and then taking the finite part as this regulator is removed.
It is a construction in this spirit (see \Cref{def:hadamard-reg-basic}) which we will employ in the present article.

\enter

\textbf{This article is structured as follows:} In \Cref{sec:epstein-glaser}, we start with an overview to the Epstein--Glaser approach to renormalization: This includes the position-space approach to renormalization as well as a discussion on the analytic background including distributions associated to Feynman graphs. Crucially, the central part is about extending a certain class of distributions to a subspace on which they are potentially singular. 

In \Cref{sec:connes-kreimer}, we give an introduction to the Connes--Kreimer framework of renormalization: This includes the Hopf algebra of Feynman graphs and the algebraic Birkhoff decomposition of the Feynman rules, defining renormalized Feynman rules and the counterterm map, with respect to a chosen renormalization scheme. Then we specialize to the position space situation, defining an appropriate target algebra of distributions and the corresponding renormalization scheme. Subsequently, we explain the necessary background to multiplicative renormalization.

In order to motivate and illustrate our main results, we then discuss the example of massless \(\phi^3_6\)-theory in \Cref{sec:examples}: Starting from the corresponding Lagrange density, we display the unrenormalized Feynman rules together with the propagator and three-point Feynman graphs up to two loops. Next, we calculate their coproducts, display the corresponding combinatorial Green's functions and their coproduct identities. From this, we calculate the respective counterterms, renormalized Feynman rules and \(Z\)-factors. Then, we calculate the explicit analytic expressions for the two-loop propagator graphs.

In \Cref{sec:main-results}, we state and prove our main results of the present article: Specifically, in \Cref{thm:Locality+Rota-Baxter implies Epstein-Glaser}, we show that the Rota--Baxter property is essentially equivalent to the Epstein--Glaser factorization property. This justifies the claim that both properties induce a description of locality in Quantum Field Theory. In \Cref{thm:Multiplicative Renormalization scheme yields local extension+ Epstein-Glaser}, we show that the combinatorics of Connes--Kreimer actually define an extension for all distributions coming from Feynman graphs. We also prove that this extension operator satisfies the Epstein--Glaser factorization property. In \Cref{thm:Renormalized Lagrangian gives renormalized Feynman rules}, we conclude by showing that the Feynman rules of the renormalized Lagrange density yield the renormalized Feynman rules. This establishes a renormalized Lagrange density for the Epstein--Glaser scheme.

Finally, we conclude our results in \Cref{sec:conclusion}.

In Appendix \ref{sec:apx-hadamard-finite-part}, we provide the necessary analytic background for the extension of distributions, in particular the \emph{Hadamard finite part}.

\section{Epstein--Glaser and the extension of distributions} \label{sec:epstein-glaser}

The expressions \emph{Epstein--Glaser renormalisation} or \emph{causal perturbation theory}, which are used synonymously, describe a position space approach to renormalisation. Specifically, they are tailored to satisfying the demands of algebraic quantum field theory, in particular locality.
It was developed in \cite{epsteinRoleLocalityPerturbation1973} and builds on earlier work
of Bogoliubov and Stückelberg.
The construction is simple and affords easy proofs of the desired properties.

The drawback is the non-explicit nature of the construction.
The combinatorics in the Epstein--Glaser approach are simple, hardly worth speaking of.
The price for this simplicity is that we are forced to introduce a partition of unity
which makes the construction less explicit than in other approaches.

\subsection{Renormalisation maps and their construction}

In Epstein--Glaser renormalisation, the renormalisation procedure is encoded
in the construction of \emph{renormalisation maps}.
Given a graph $\gamma$ with $V(\gamma) = I \subset \mathbb{N}$,
the Feynman rules of the theory determine unrenormalized (or \emph{bare}) amplitudes
$t_\gamma \in \mathscr{D}'(M^{I}\backslash D_I)$, where
\begin{align}
    D_I = \{ (x_{i_1}, \ldots, x_{i_n})\in M^I| \exists k \neq l: x_{i_k} = x_{i_l}\}
\end{align}
is the \emph{fat diagonal} in $M^I$.
The renormalisation maps transform such an unrenormalised amplitude into a distribution
$\overline{t_\gamma} \in \mathscr{D}'(M^I)$ on the whole space.

The notion of renormalisation maps in the form presented here derives from \cite{nikolovAnomaliesQuantumField2009,nikolovRenormalizationTheoryFeynman2009} and the subsequent developments in \cite{nikolovEuclideanConfigurationSpace2013,nikolovRenormalizationMasslessFeynman2014,dangExtensionDistributionsScalings2014,dangComplexPowersAnalytic2015,dangRenormalizationFeynmanAmplitudes2017,dangRenormalizationQuantumField2019a,hofmannMicrolocalMethodsQuantum2024}.

The distributions on which the renormalisation maps act are defined in terms of Feynman graphs.
Let $\Gamma(I)$ be the set of Feynman graphs $\gamma$ with vertex set $V(\gamma) = I$.
Each edge $e\in E(\gamma)$ is labelled by a distribution $\Phi_{\gamma, M}(e)$, drawn from some space of admissible distributions on $\mathscr{D}'(M\times M)$ called \emph{propagators}.

\enter

\begin{defn}[Propagators]
\label{def:propagators}
Let $\Psi_{\text{cl}}^m(M)$ denote the space of classical pseudodifferential operators of order $m$ on $M$.
We will write $\Psi_{\text{cl}}(M) = \bigcup_{m\in\mathbb{Z}} \Psi_{\text{cl}}^m(M)$ for the space of integer order pseudodifferential operators.
By the Schwartz kernel theorem, any operator $A \in \Psi_{\text{cl}}(M)$
possesses an integral kernel $K_A \in \mathscr{D}'(M\times M)$ and we freely identify the operator $A$ and its kernel $K_A$.

A \emph{propagator} on $M$ is a distribution $u \in \mathscr{D}'(M\times M)$
such that
\begin{itemize}
    \item $u \in \Psi_{\text{cl}}(M)$, and
    \item $u$ is symmetric, $u(x,y) = u(y,x)$.
\end{itemize}
\end{defn}

\henter

\begin{defn} \label{def: Feynman distributions}
Let \(I \subset \mathbb{N}\) be a finite set, and let
\(\Gamma(I)\) be the set of Feynman graphs \(\gamma\) with vertex set \(V(\gamma) = I\).
A \emph{Feynman rule on $M$} is a map $\Phi_{\gamma, M}: E(\gamma) \to \Psi_{\text{cl}}(M)$.
Any such Feynman rule determines an \emph{unrenormalised Feynman amplitude}
\begin{align}
\label{eq:unren-feynm-eg}
t_{\Phi, \gamma}(x_I) = \prod_{e \in E(\gamma)} \Phi_{\gamma, M}(e)(x_{s(e)}, x_{t(e)}) \in \mathscr{C}^\infty(M^I\backslash D_I) \subset \mathscr{D}'(M^I\backslash D_I).
\end{align}
where $(s(e), t(e))$ are the source and target vertices of $e$.

We denote by $\mathscr{O}(M,I) \subset \mathscr{D}'(M^I\backslash D_I)$ the vector space
\begin{align*}
    \mathscr{O}(M,I) := \operatorname{span}_{\mathscr{C}^\infty(M^I)}\{ t_{\Phi, \gamma}| \gamma \in \Gamma(I), \Phi: E(\gamma) \to \Psi_{\text{cl}}(M)\}.
\end{align*}
\end{defn}

\enter

\begin{exmp}
Let $\gamma$ be the ice-cream cone graph, i.e.
\begin{equation}
	\gamma \coloneq  \vcenter{\hbox{\includegraphics[width=0.25\textwidth]{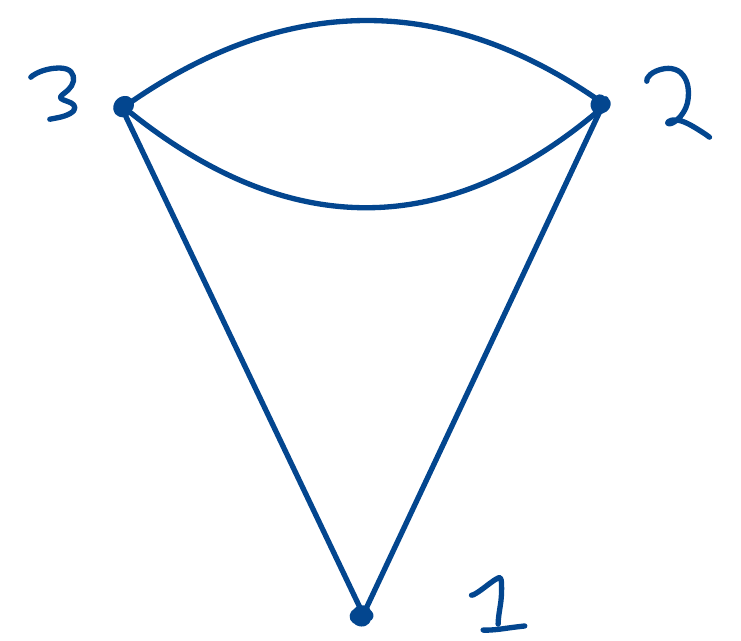}}}
\end{equation}
Enumerate the edges as $e_1, e_2, e_3, e_4$ and assign to each edge a distribution $u_i \in \mathscr{D}'(M\times M)$ which is the kernel of a pseudodifferential operator $P_i$.
Then the unrenormalised Feynman distribution is given via
\begin{align*}
    t_\gamma(x_1,x_2,x_3) = u_1(x_1, x_2) u_2(x_1, x_3) u_3(x_2, x_3) u_4(x_2, x_3).
\end{align*}
\end{exmp}

\enter

The core of the Epstein--Glaser method is to avoid complicated combinatorics
by restricting a Feynman amplitude $t_\gamma$ to certain open subsets where we avoid a part of the
singularities.
We can then patch together these restrictions using a partition of unity.
The basis for this technique is the following result,
known variously as the \emph{geometric lemma}, \emph{Stora's lemma} or the \emph{diagonal lemma}.
To illustrate the cover we return to the ice cream cone graph:
\begin{figure}[h]
    \centering
    \begin{subfigure}[b]{0.25\textwidth}
    \includegraphics[width=\textwidth]{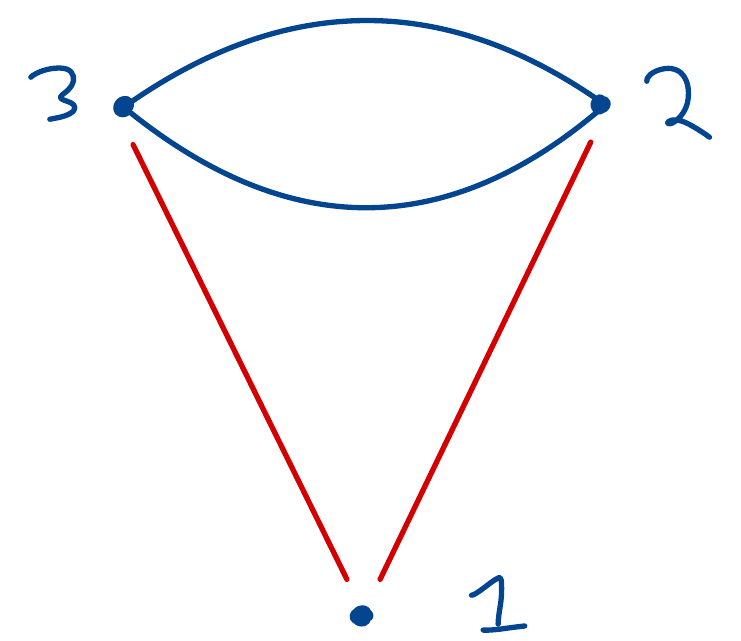}
    \caption{$C_{\{1\}, \{2,3\}}$}
    \label{fig:stora-cover-1}
    \end{subfigure}
    \hfill
     \begin{subfigure}[b]{0.4\textwidth}
         \centering
         \includegraphics[width=\textwidth]{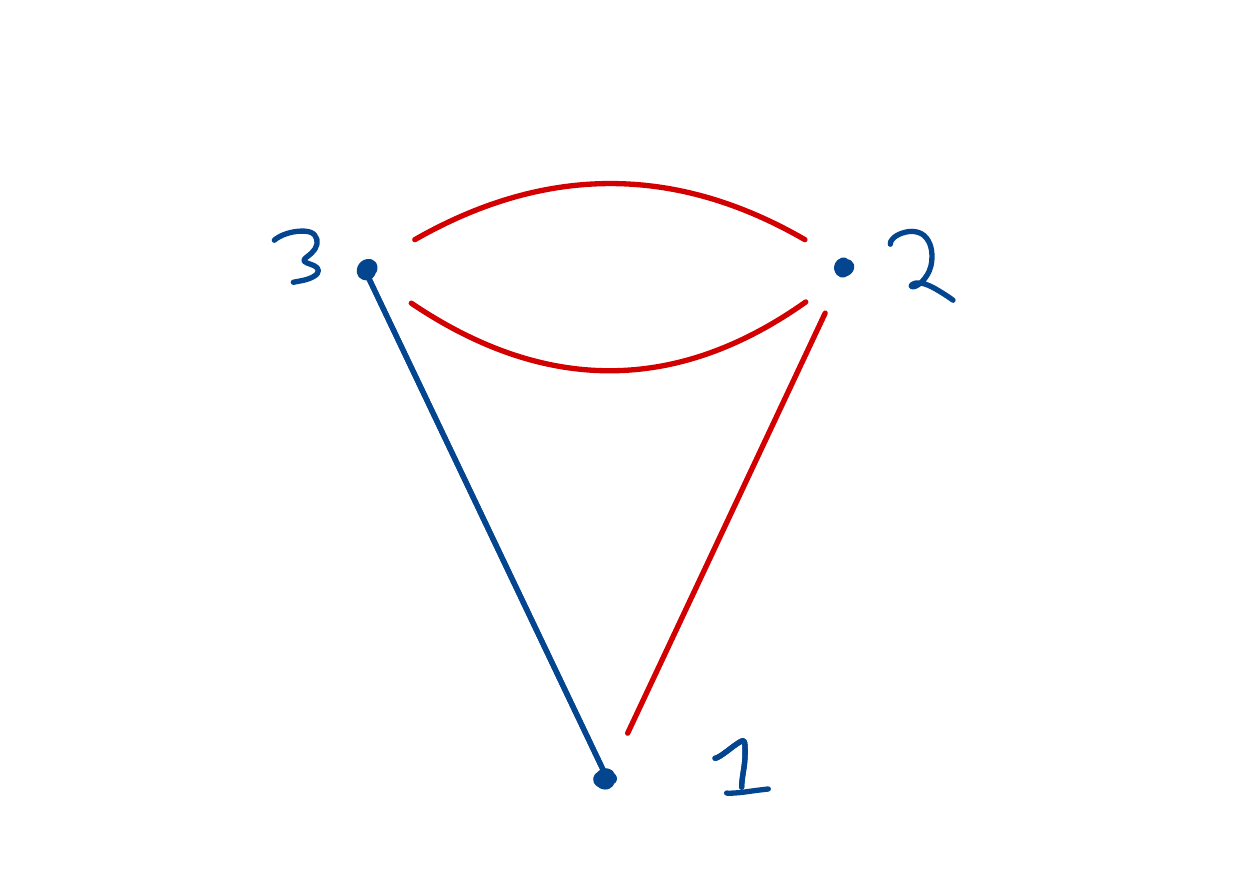}
         \caption{$C_{\{2\}, \{1,3\}}$}
         \label{fig:stora-cover-2}
     \end{subfigure}
    \hfill
     \begin{subfigure}[b]{0.25\textwidth}
         \centering
         \includegraphics[width=\textwidth]{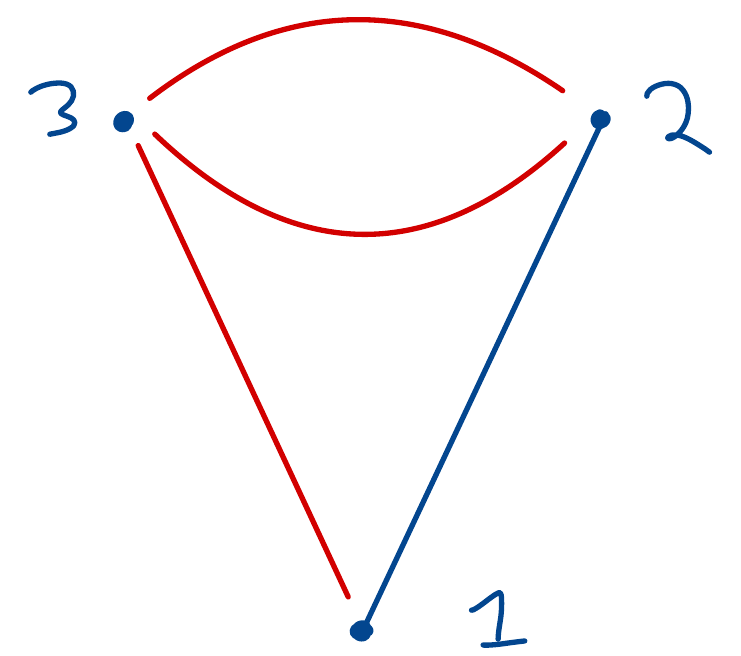}
         \caption{$C_{\{3\}, \{1,2\}}$}
         \label{fig:stora-cover-3}
     \end{subfigure}
\end{figure}

\enter

\begin{lem}
\label{lem:stora}
Let $I \subset \mathbb{N}$ be finite. If $J\subset I$, define
$J^c \coloneq I \backslash J$ and define the open sets
\begin{align}
    C_{I,J} = \{x \in M^I| x_i \neq x_j\ \forall\ i \in J, j \in J^c\} \, .
\end{align}
Then
\begin{align}
    M^I \backslash d_I(M) = \bigcup_{\substack{J \subsetneq I, \\
    J \neq \emptyset}} C_{I,J} \, .
\end{align}
\end{lem}

\enter

The existence of such a partition of unity enters directly into our definition
of renormalisation maps.

\enter

\begin{defn}\label{def:renorm-maps}
Fix a dimension $d$.
A family of \emph{renormalisation maps} is given by linear maps
\begin{align}
\Ren_{I,M}: \mathscr{O}(I,M) \to \mathscr{D}'(M^{I})
\end{align}
for every $d$-dimensional Riemannian manifold $(M, g)$ and each finite set $I \subset \mathbb{N}$,
satisfying the following properties.

\begin{enumerate}
\item \emph{(Extension)} The restriction of \(\Ren_{I,M}({t}_\gamma)\) to
\(M^{I} \backslash D_{I}\) is \({t}_{\gamma}\).
\item \emph{(Factorization)} Let \(\gamma_1,\gamma_2\) be disjoint Epstein--Glaser subgraphs such that
\(V(\gamma_1) = J = I \backslash V(\gamma_2)\).
Then
  \begin{align}
\Ren_{I,M}(t_\gamma)|_{C_J} = \left(\Ren_{J,M}(t_{\gamma_1})
 \otimes \Ren_{J^c,M}(t_{\gamma_{2}})\right)|_{C_{J}}
\prod_{\substack{s(e) \in J\\ t(e) \in J^{c}}} G_{e}(x_{s(e)},x_{t(e)})|_{C_J} \, .
\end{align}
where $\{C_J | J \subset I\}$ is the cover of $M^I \backslash d_I$ from
\Cref{lem:stora}.
\item \emph{(Covariance)} If \(f: M \to N\) is an isometric embedding and \(f^I: M^I \to N^I\)
is the induced isometric embedding,
\begin{align}
\Ren_{I,M}(\gamma) = (f^I)^* \left ( \Ren_{I, N}(\gamma)|_{\operatorname{ran} f^{I}} \right ) \, .
\end{align}
\end{enumerate}
\end{defn}

\enter

\Cref{lem:stora} and the factorization property imply that
$\Ren_{I,M}$ is determined uniquely on $M^I\backslash d_I$ by
the maps $\Ren_{J,M}$ for proper subsets $J\subsetneq I$.
Therefore one can build the maps $\Ren_{I,M}$ inductively.
At each step of the induction, it is necessary to extend a distribution
from $M^I\backslash d_I$ to $M^I$, such that the remaining axioms are satisfied.

\enter

\begin{construction}
\label{construct:ren-maps}
\begin{enumerate}
    \item Proceed inductively. The induction start for $|I| \leq 1$ is trivial. Then suppose that $\Ren_{J,M}(t_\gamma)$ has been defined for all $J\subsetneq I$ and graphs $\gamma$ with $V(\gamma) = J$.
    \item Define $\dot{\Ren}_{I,M}(t_\gamma) \in \mathscr{D}'(M^I\backslash d_I)$ by using the factorization condition. If $\chi_I$ is a partition of unity subordinate to $C_I$, put
	    \begin{align}\label{eq: factorization property}
	        \dot{\Ren}_{I,M}(t_\gamma) = \sum_{J\subset I} \chi_J \cdot \left(\Ren_{J,M}(t_{\gamma_1}) \otimes \Ren_{J^c,M}(t_{\gamma_{2}})\right)
\prod_{\substack{i \in J\\ j \in J^{c}}} G_{ij}(x_i,x_j)^{l_{ij}} \, .
	    \end{align}
	    This is independent of the choice of partition of unity.
    \item Let $\exp_x: T_x M \supset U \xrightarrow{\simeq} V \subset M$ be the Riemannian exponential map. Then there is an induced map $\exp_x^I: U^I \xrightarrow{\simeq} V^I$ and a tubular neighborhood embedding
	    \begin{align}
	    \begin{aligned}
	        &\phi_I: N d_I \to M^I\\
	        &(x, h_1,\ldots, h_n) \mapsto (\exp_x h_1, \ldots, \exp_x h_n) \, .
	    \end{aligned}
	    \end{align}
    \item Given an extension map $E_I$ from $N d_I \backslash d_I \to N d_I$, such that $\phi_I^* \dot{\Ren}_{I,M}(t_\gamma)$ is in the domain of $E_I$,
	\begin{align}
	    \Ren_{I, M}(t_\gamma) \coloneq (\phi_I^{-1})^* E_I(\phi_I^* \dot{\Ren}_{I,M}(t_\gamma)) \, .
	    \end{align}
	\end{enumerate}
\end{construction}

\enter

We have not yet defined what we mean by an extension map, and on what distributions such a map can be defined. This we will do in the next section subsection.

\subsection{Extension maps} \label{sec: extension of distributions}

We will now now provide a definition for the extension problem $\mathbb{R}^d\backslash\{0\} \to \mathbb{R}^d$.
Since $N d_I$ is a vector bundle, Step 4 of \Cref{construct:ren-maps} can indeed be reduced to this case (or the similar case of a neighborhood of the origin in $\mathbb{R}^d$) by extending in each fiber of the bundle.

\enter

\begin{defn} \label{def:ext-map-restrict}
Let $U \subset \mathbb{R}^d$ be a neighborhood of $0$ and write $\dot{U} \coloneq U \backslash 0$.
Let $\mathscr{D}(\dot{U}) \to \mathscr{X}(U) \to \mathscr{D}'(\dot{U})$ be a space of distributions on $\dot{U}$.
An \emph{extension map} is a linear map $E: \mathscr{X}(U) \to \mathscr{D}'(U)$ such that
\begin{enumerate}
\item if $r: \mathscr{D}'(U) \to \mathscr{D}'(\dot{U})$ is the restriction map,
    then $r\circ E = \operatorname{id}_{\mathscr{D}'(\dot{U})}$.
\end{enumerate}We will call $E$ \emph{flat} if moreover 
\begin{enumerate}[resume]
   \item $E$ does not increase supports, i.e. $\operatorname{supp}(Eu) \subset \overline{\operatorname{supp}(u)}^{U}$ (closure in $U$). \label{def:ext-map-supp}
\end{enumerate}
\end{defn}

\enter

\begin{rem}
In order to construct renormalisation maps, one uses extension maps as just defined.
The requirement of \emph{flatness} is required to ensure that the construction is covariant.
If covariance is left out of the requirements for renormalisation maps,
one can dispense with flatness as well.
\end{rem}

\enter

The extension of distributions is closely tied to their scaling behavior at 0.
This is quantified by the \emph{scaling degree}, a generalization of the degree of homogeneity of a distribution.

\enter

\begin{defn}
Let $\mathbb{R}^+$ denote the group $(0,\infty)$ with multiplication.
\begin{enumerate}
\item An action
\begin{align}
\begin{aligned}
    \mathbb{R}^+ \times V &\to V\\
    (\lambda, x) &\mapsto \mathfrak{s}(\lambda)x
\end{aligned}
\end{align}
of $\mathbb{R}^+$ on a set $V$ will be called a \emph{scaling action}.
\item A set $V$ together with a scaling action $\mathfrak{s}$ will be called a
\emph{conic set}. If $V$ is a smooth manifold we will call the pair $(V, \mathfrak{s})$ a \emph{conic manifold}.
\item If $V \subset E$ where $E$ is a vector space, we will call $V$ \emph{conic}
if it is a conic set with respect to the standard action of $\mathbb{R}^+$ on $E$ by scalar multiplication.
\item A subset $U$ of a conic set $\Gamma$ will be called \emph{scaling convex} if $\mathfrak{s}(\lambda)x \in U$ for all $x \in U$ and for all $\lambda \in (0,1]$.
\item If the scaling action $\mathfrak{s}$ is understood, we may occasionally abuse notation and write $\lambda x \coloneq \mathfrak{s}(\lambda)x$.
\end{enumerate}
\end{defn}

\enter

Suppose $\phi \in \mathscr{D}(U)$ and $U$ is scaling convex.
Define $\phi_{\lambda^{-1}}(x) = \phi(\lambda^{-1} x)$, then for $\lambda \in (0,1]$
$\operatorname{supp} \phi_{\lambda^{-1}} = \lambda \operatorname{supp} \phi \subset U$
if $U$ is scaling convex.
Hence we can define the scaling action on distributions in the following way.

\enter

\begin{defn}
Let $(\Gamma, \mathfrak{s})$ be a conic manifold, $U\subset \Gamma$
be open and scaling convex, $u \in \mathscr{D}'(U)$,
    $\lambda \in (0,1]$.
    Then the distribution $\lambda^* u \coloneq \mathfrak{s}(\lambda)^*u \in \mathscr{D}'(U)$ is defined as
    \begin{align}
        \langle \lambda^* u , \phi\rangle = \lambda^{-d} \langle u , \phi_{\lambda^{-1}} \rangle \, .
    \end{align}
\end{defn}

\enter

\begin{rem}
Let $U \subset \mathbb{R}^d$ be scaling convex and $\mathfrak{s}$ the standard scaling action.
Then if $u$ is smooth, $(\lambda^*u)(x) = u(\lambda x)$. A distribution $u\in \mathscr{D}'(U)$ is called \emph{homogeneous} of degree $\alpha$ if $\lambda^* u = \lambda^\alpha u$.
\end{rem}

\enter

\begin{defn} \label{defn:scaling-degree}
Let $\Gamma \subseteq \mathbb{R}^d$ be a conic set.
\begin{enumerate}
    \item The space
$\mathscr{H}_s(\Gamma)$ consists of all $u \in \mathscr{D}'(\Gamma)$
such that
\begin{align}
    p_{\phi} \coloneq \sup_{\lambda \in (0,1]} \lambda^s | \langle \mathfrak{s} (\lambda)^* u , \phi \rangle | < \infty \, ,
    \qquad \phi \in \mathscr{D}(\Gamma) \, .
\end{align}
This defines a family of seminorms $p_\phi$ on $\mathscr{H}_s(\Gamma)$.
\item The \emph{scaling degree} of $u \in \mathscr{D}'(\Gamma)$ is
\begin{align}
    \operatorname{SD} \left ( u \right ) \coloneq \operatorname{inf} \{s \in \mathbb{R}| u \in \mathscr{H}_s(\Gamma)\} \, .
\end{align}
\end{enumerate}
\end{defn}

\enter

\begin{rem}
If $u$ is homogeneous of degree $\alpha \in \mathbb{C}$, then $\operatorname{SD}(u) = -\operatorname{Re}\alpha$.
\end{rem}

\enter

If $u$ has $\operatorname{SD}(u) < d$ then an extension exists by a simple regularization procedure.
This procedure will follow us throughout the article; we will refer to it as \emph{Hadamard regularization}.

\enter

\begin{defn}
\label{def:hadamard-reg-basic}
Let $u \in \mathscr{D}'(\dot{U})$, let $\chi \in \mathscr{D}(U)$ such that $\chi \equiv 1$ in a neighborhood of $0$.
As before we denote $\chi_\epsilon(x) = \chi(x/\epsilon)$.
The \emph{Hadamard regularization} of $u$ is the family $T_\epsilon u := (1-\chi_\epsilon) u \in \mathscr{D}'(U)$.
\end{defn}

\henter

\begin{lem} \label{lem:extension lem of distributions}
If $\operatorname{SD}(u) < d$, then $T_\epsilon u \xrightarrow{\epsilon\to 0} \overline{u}$ where $\overline{u} \in \mathscr{D}'(U)$ and $\overline{u}|_{\dot{U}} = u$.
\end{lem}

\section{Connes--Kreimer and multiplicative renormalization} \label{sec:connes-kreimer}

It was observed by Kreimer that the subdivergence structure of Feynman graphs can be encoded into the coproduct of a Hopf algebra \cite{Kreimer:1997dp}. More precisely, the coproduct takes a Feynman graph and maps it to the sum of all possible combinations of divergent subgraphs tensored with the corresponding cograph, for us in position space the Feynman graph obtained by deleting each connected component of the divergent subgraph. Then, together with Connes, they constructed the renormalized Feynman rules as an algebraic Birkhoff decomposition with respect to a chosen renormalization scheme \cite{Connes:1998qv,Connes:1999zw}. Specifically, their main result was the existence and uniqueness of the renormalized Feynman rules and the counterterm map once the renormalization scheme is chosen, cf.\ \cite{Connes:1999yr,Connes:2000fe}. More formally, the Connes--Kreimer renormalization framework works as follows:

\subsection{The renormalization Hopf algebra}

We start with the definition of the renormalization Hopf algebra: The idea is that we consider the \(\mathbb{Q}\)-vector space of Feynman graphs and construct the coproduct such that it captures the subdivergence structure of Feynman diagrams. Then, we will define the Feynman rules and with that the algebra of formal Feynman distributions as their image.

\enter

\begin{defn}[Set of (1PI) Feynman graphs] \label{defn:feyn-gra}
    A Feynman graph \(\Gamma \equiv (G, \mathfrak{p})\) is a graph \(G \equiv (V,E)\), whose edge set \(E\) is colored by the set of propagating particle types \(\mathcal{E}\) and whose vertex set \(V\) is colored by the set of particle interactions \(\mathcal{V}\), i.e.\ we have the two coloring functions
    \begin{subequations}
    \begin{align}
        \mathfrak{p}_E & \colon E \to \mathcal{E}
        \intertext{and}
        \mathfrak{p}_V & \colon V \to \mathcal{V} \, .
    \end{align}
    \end{subequations}
    Additionally, a (Feynman) graph is called \emph{one-particle irreducible (1PI)}, if it is bridge-free, i.e.\ if it stays connected after the removal of any of its internal edges. We denote the set of all 1PI Feynman graphs for a given theory by \(\mathcal{G}\). Furthermore, we denote by \(\mathcal{G}_\text{conn}\) the set of all connected Feynman graphs, including trees.
\end{defn}

\enter

\begin{defn}[Renormalization Hopf algebra] \label{defn:ren-ha}
Let \(\QFT\) be a Quantum Field Theory with \(\GQ\) denoting the set of its 1PI Feynman graphs. Then, we construct the corresponding renormalization Hopf algebra as the power series algebra\footnote{Strictly speaking, originally the renormalization Hopf algebra was defined as the polynomial algebra \(\HQ \coloneq \mathbb{Q}[\GQ]\) and then completed with respect to one of its gradings, which is equivalent to consider the power series algebra \(\HQ \coloneq \mathbb{Q}[[\GQ]]\) from the start.}
\begin{subequations}
\begin{align}
    \HQ & \coloneq \mathbb{Q}[[\GQ]]
    \intertext{with product \(m \colon \HQ \otimes_\mathbb{Q} \HQ \to \HQ\), given by disjoint union, and corresponding unit \(\one \colon \mathbb{Q} \hookrightarrow \HQ\), given by the empty graph. Additionally, we introduce coproduct \(\Delta\), counit \(\coone\) and antipode \(S\) as follows:\footnotemark}
    \Delta \; & : \quad \HQ \to \HQ \otimes_\mathbb{Q} \HQ \; , \quad \Gamma \mapsto \sum_{\gamma \in \DQ{\Gamma}} \delta_{\left ( \gamma \hookrightarrow \Gamma \setminus \gamma \right )} \gamma \otimes_\mathbb{Q} \Gamma \setminus \gamma \\
    \coone \; & : \quad \HQ \to \mathbb{Q} \; , \quad \Gamma \mapsto \begin{cases} q & \text{if \(\Gamma\) is the empty graph with coefficient \(q \in \mathbb{Q}\)} \\ 0 & \text{else} \end{cases} \\
    S \; & : \quad \HQ \to \HQ \; , \quad \Gamma \mapsto - \sum_{\substack{\gamma \in \DQ{\Gamma} \\ \gamma \subsetneq \Gamma}} S \left ( \gamma \right ) \left ( \Gamma \setminus \gamma \right ) \quad \text{and} \quad S \left ( \one \right ) = 1
\end{align}
\end{subequations}
\footnotetext{We emphasize the differences for the position space and momentum space formulations: While the cograph in the momentum space representation is given via \(\Gamma / \gamma\), i.e.\ shrinking the subgraph to a new vertex, in position space it is given via \(\Gamma \setminus \gamma\), i.e.\ removing the subgraph.}where \(\DQ{\cdot}\) denotes the set of superficially divergent subgraphs (properly introduced below), \(\Gamma \setminus \gamma\) denotes the deletion of each connected component of \(\gamma\) inside \(\Gamma\) but keeping the adjacent vertices. Sometimes we also write \(\otimes^\delta_\mathbb{Q}\) to include the vertex identification map \(\delta\) implicitly into the tensor product. Finally, we remark that the definition of the antipode is by recursion.
\end{defn}

\enter

With the renormalization Hopf algebra settled, we now proceed to define the appropriate target algebra for the Feynman rule map:

\enter

\begin{defn}[Scaling degree and superficial degree of divergence]
Recall \Cref{defn:scaling-degree}, where we defined the \emph{scaling degree} \(\operatorname{SD} \left ( u \right )\) of a distribution \(u \in \mathscr{D}'(M)\). Building on this, we now define its associated \emph{superficial degree of divergence}\footnote{We remark that there exist different sign conventions for the superficial degree of divergence: Our convention is such that it matches the divergence degree of the corresponding Feynman integral.}
\begin{equation}
    \operatorname{SDD} \left ( u \right ) \coloneq d - \operatorname{SD} \left ( u \right ) \, .
\end{equation}
We refer to e.g.\ \cite{Brunetti:1999jn} for a more detailed discussion.
\end{defn}

\enter

\begin{defn}[Set of superficially divergent subgraphs]
Let \(\Gamma \in \mathcal{G}\) be a 1PI Feynman graph. Then we denote by \(\DQ{\Gamma}\) the set of disjoint unions of 1PI subgraphs, where every connected component is superficially divergent:
\begin{equation}
    \DQ{\Gamma} \coloneq \left \{ \gamma \subseteq \Gamma \; \Bigg \vert \; \gamma \cong \bigsqcup_{m = 1}^n \gamma_m \; \text{ such that } \; \operatorname{SDD} \left ( u_{\gamma_m} \right ) \geq 0 \; \text{ for all } \; m \in \{ 1, \dots n \} \right \}
\end{equation}
Observe that this includes \(\Gamma \in \DQ{\Gamma}\) if \(\Gamma\) is itself superficially divergent. Furthermore, it is also convenient to include \(\one \in \DQ{\Gamma}\).
\end{defn}

\enter

A natural object of interest in the renormalization Hopf algebra is the combinatorial Green's function, as the sum over all Feynman diagrams with a specific external leg structure, weighted by their symmetry factor: Evaluating this linear combination of Feynman graphs under the (renormalized) Feynman rules map yields the (renormalized) Green's function, from which the corresponding cross section can be calculated.

\enter

\begin{defn}[(Restricted) Combinatorial Green's function] \label{defn:combinatorial-greens-function}
    Let \(r\) denote an amplitude, represented via the external leg structure (ELS) of a Feynman diagram. Then, we define its corresponding combinatorial Green's function \(\mathfrak{X}^r \in \HQ\) as follows: First, we define the \emph{precombinatorial Green's function} as the sum over all such graphs
    \begin{subequations}
    \begin{align}
        \mathfrak{x}^r & \coloneq \sum_{\substack{\Gamma \in \GQ \\ \operatorname{ELS} \left ( \Gamma \right ) = r}} \frac{1}{\operatorname{Sym} \left ( \Gamma \right )} \Gamma
        \intertext{and then obtain the full \emph{combinatorial Green's function} as the unit plus or minus the precombinatorial Green's function, depending on whether it is a vertex or a propagator amplitude, as follows:\footnotemark}
        \mathfrak{X}^r & \coloneq \begin{cases} \one + \mathfrak{x}^r & \text{If \(r\) is a vertex amplitude} \\ \one - \mathfrak{x}^r & \text{If \(r\) is a propagator amplitude} \end{cases} \label{eqn:full-combinatorial-greens-function}
        \intertext{Finally, we define the \emph{restricted combinatorial Green's function} of loop number \(L\) as the restriction of the above expression to graphs with loop number \(L\):}
        \mathfrak{X}^r_L & \coloneq \eval{\mathfrak{X}^r}_{b_1 \left ( \Gamma \right ) = L}
    \end{align}
    \end{subequations}
    \footnotetext{The reason for the following sign-rule is explained in \Cref{rem:propagator-geom-series}.}%
    We emphasize that Quantum Field Theories with more than one interaction vertex --- especially gauge theories --- require a finer grading than the loop number, cf.\ \cite[Definition 2.18]{Prinz:2019awo} and \cite[Remark 3.31]{Prinz:2018dhq}. Since we only consider scalar theories with one interaction term in the present article, we avoid these technical complications here.
\end{defn}

\enter

\begin{defn}[Cutting and inserting graphs] \label{defn:cutting-graphs}
    Let \(\Gamma \in \GQ\) be a Feynman graph. Then, we define the following cutting operation:
    \begin{equation}
         \cut \, : \quad \HQ \to \HQ \, , \quad \Gamma \mapsto \Gamma_{\mathfrak{e}, \mathfrak{v}} \, ,
    \end{equation}
    where \(\Gamma_{\mathfrak{e}, \mathfrak{v}}\) denotes the sum over all possibilities to cut \(\mathfrak{e}\) edges and remove \(\mathfrak{v}\) vertices from \(\Gamma\).\footnote{See the cographs in \Cref{ssec:coproduct-identities-examples} for examples.}
\end{defn}

\enter

\begin{defn}[Cut (restricted) combinatorial Green's function]
    Given the situation of Definitions \ref{defn:combinatorial-greens-function} and \ref{defn:cutting-graphs}, then we define the \emph{cut precombinatorial Green's function} via
    \begin{subequations}
    \begin{align}
        \mathfrak{y}^r_{\mathfrak{e}, \mathfrak{v}} & \coloneq \cut \left ( \mathfrak{x}^r \right )
        \intertext{and the full \emph{cut combinatorial Green's function} as follows:}
        \mathfrak{Y}^r_{\mathfrak{e}, \mathfrak{v}} & \coloneq \cut \left ( \mathfrak{X}^r \right )
        \intertext{Finally, we define the \emph{cut restricted combinatorial Green's function}:}
        \mathfrak{Y}^r_{L, \mathfrak{e}, \mathfrak{v}} & \coloneq \cut \left ( \mathfrak{X}^r_L \right )
    \end{align}
    \end{subequations}
\end{defn}

\enter

The cut restricted combinatorial Green's functions allow us now to write the position-space \emph{coporduct identities}, which are an important part of \emph{multiplicative renormalization}, as will be remarked in \Cref{ssec:multiplicative-renormalization}:

\enter

\begin{prop}[Coproduct identities] \label{prop:coproduct-identities}
    Given a combinatorial Green's function \(\mathfrak{X}^r \in \mathcal{H}\) of a renormalizable scalar theory with only one interaction vertex.\footnote{Examples are e.g.\ \(\phi^3_6\)-theory, which we discuss throughout the article, and \(\phi^4_4\)-theory. However, we remark that it is possible to relax both conditions, the \emph{renormalizable} as well as the \emph{scalar theory with only one interaction-vertex}: The first needs to be replaced by being \emph{cograph-divergent}, cf.\ \cite[Definition 3.6]{Prinz:2019awo}, and the second requires the restriction to a finer grading, cf.\ \cite[Definition 2.18]{Prinz:2019awo}: This is shown in \cite[Proposition 4.2]{Prinz:2019awo}.} Then, the following coproduct identities hold:
    \begin{subequations}
    \begin{align}
        \Delta \left ( \mathfrak{X}^r \right ) & = \sum_{\mathfrak{e} = 0}^\infty \sum_{\mathfrak{v} = 0}^\infty \left ( \mathfrak{x}^v \right )^\mathfrak{v} \left ( \mathfrak{x}^e \right )^\mathfrak{e} \otimes^\delta_\mathbb{Q} \mathfrak{Y}^r_{\mathfrak{e}, \mathfrak{v}}
        \intertext{and}
        \Delta \left ( \mathfrak{X}^r_L \right ) & = \sum_{l = 0}^L \sum_{\mathfrak{e} = 0}^\infty \sum_{\mathfrak{v} = 0}^\infty \left [ \left ( \mathfrak{x}^v \right )^\mathfrak{v} \left ( \mathfrak{x}^e \right )^\mathfrak{e} \right ]_l \otimes^\delta_\mathbb{Q} \mathfrak{Y}^r_{L - l, \mathfrak{e}, \mathfrak{v}} \, .
    \end{align}
    \end{subequations}
\end{prop}

\begin{proof}
    Since the theory is renormalizable and has only one interaction vertex, all superficially divergent subgraphs are either edge-corrections, vertex-corrections or disjoint unions thereof: In particular, non-trivial edge-corrections are given by \(\mathfrak{x}^e\) while non-trivial vertex-corrections are given by \(\mathfrak{x}^v\). Thus, the most general case is given by summing over all numbers of edge- and vertex corrections and considering their product.\footnote{We remark the similarity to \cite[Proposition 2.30]{Prinz:2019awo} and \cite[Proposition 2.15]{Prinz:2025obr}, which are the most general 1PI and connected momentum-space versions, including super- and non-renormalizable theories which are \emph{cograph divergent} as well as general gradings, in particular the finer \emph{vertex-} and \emph{coupling-grading}, cf.\ \cite[Remark 3.31]{Prinz:2018dhq} and \cite[Section 3.5]{Prinz:2022qll}.}.
\end{proof}

\enter

\begin{rem} \label{rem:propagator-geom-series}
    Generally, in momentum space the edge corrections are given via \(1 / \mathfrak{X}^e\): This takes into account that it is possible to insert multiple edge-corrections into a single edge by employing the geometric series:
    \begin{equation}
        \frac{1}{\mathfrak{X}^e} \equiv \frac{1}{\one - \mathfrak{x}^e} \equiv \one + \sum_{k = 1}^\infty \left ( \mathfrak{x}^e \right )^k
    \end{equation}
    Since in the position-space representation the cographs are given via \(\Gamma \setminus \gamma\) instead of \(\Gamma / \gamma\) for \(\Gamma \in \GQ\) and \(\gamma \in \DQ{\Gamma}\) the best we can do is to count the number of connected vertex and edge graph components of \(\gamma\), as we need to cut the Green's function accordingly: Thus, here in the position space the famous coproduction formula takes on the rather trivial form of \Cref{prop:coproduct-identities}.
\end{rem}

\subsection{The algebraic Birkhoff decomposition}

Having set up the renormalization Hopf algebra \(\HQ\), the algebra of formal integral expressions \(\AQ\) and the Feynman rule map \(\Phi \colon \HQ \to \AQ\) between them, we now construct the renormalized Feynman rules via an algebraic Birkhoff decomposition. Specifically, we start with the set of characters \(\CG\), i.e.\ the group of algebra morphisms from \(\HQ\) to \(\AQ\), which we turn into a group by means of the convolution product \(\star\) and the respective identity \(\mathbf{1}_\star\): This will turn out to be the natural setting for the renormalization group. Then, we define a renormalization scheme \(\mathscr{R} \in \operatorname{End} \left ( \AQ \right )\) as a splitting of the target algebra \(\AQ \cong \mathcal{A_+ \oplus \AQ_-}\), where \(\AQ_+\) are the convergent integral expressions and \(\AQ_-\) are the purely divergent expressions, as seen via the renormalization scheme \(\mathscr{R}\).

\enter

\begin{defn}[Character group] \label{defn:character-group}
    Let \(\HQ\) be a Hopf algebra and \(\AQ\) be an algebra, then we denote the set of algebra morphisms, also called \emph{characters}, as follows
    \begin{subequations}
    \begin{align}
        \CG & \coloneq \left \{ f \colon \HQ \to \AQ \mid \text{\(f\) is an algebra morphism} \right \} \, .
        \intertext{Additionally, we turn them into a semi-group by introducing the following product, using the coproduct \(\Delta_\HQ\) on the Hopf algebra and the product \(m_\AQ\) on the algebra,}
        f \star g & \coloneq m_\AQ \circ \left ( f \otimes_\mathbb{Q} g \right ) \circ \Delta_\HQ \, ,
        \intertext{then lift it to a monoid by introducing the following unit, as the composition of the counit \(\coone_\HQ\) of the Hopf algebra and the unit \(\one_\AQ\) of the algebra}
        \mathbf{1}_\star & \coloneq \one_\AQ \circ \coone_\HQ
        \intertext{and then finally turn it into a group via the following inversion law, using the antipode \(S_\HQ\) of the Hopf algebra,}
        f^{(\star -1)} & \coloneq f \circ S_\HQ \, .
    \end{align}
    \end{subequations}
    We remark that this construction also works for endocharacters \(G^\HQ_\HQ\) and produces a monoid, if \(\HQ\) is only a bialgebra, as the antipode is required for the group inversion.
\end{defn}

\enter

The idea of the \emph{algebraic Birkhoff decomposition} is now to split the target algebra \(\AQ \cong \AQ_+ \oplus \AQ_-\) by means of a renormalization scheme \(\mathscr{R}\) and then decompose any character \(f \in \CG\) into the following convolution product \(f \equiv {f^-}^{(\star -1)} \star f^+\), where \(f^\pm \in \CG\) are also characters whose image is contained in \(\AQ_\pm\), respectively. To this end, we introduce a renormalization scheme as the splitting of the target algebra \(\AQ\):

\enter

\begin{defn} \label{defn:renormalization-scheme}
    A renormalization scheme is a splitting of the short exact sequence of vector spaces
    \begin{equation}
    \begin{tikzcd}
        0 \arrow[r] & \A_+ \arrow[r] & \A \arrow[l, dashed, "\mkern-27mu 1 - \mathscr{R}", bend left=33] \arrow[r, "\mathscr{R}"] & \A_- \arrow[r] & 0 \, ,
    \end{tikzcd}
    \end{equation}
    such that
    \begin{enumerate}
        \item \(\mathscr{R}\) is a projector onto \(\mathcal{A}_-\), and
        \item \(\mathscr{R}\) satisfies the Rota--Baxter property
    \begin{equation}
        \mathscr{R} \circ m_\mathcal{A} + m_\mathcal{A} \circ \left ( \mathscr{R} \otimes \mathscr{R} \right ) = \mathscr{R} \circ m_\mathcal{A} \circ \left ( \operatorname{Id}_\mathcal{A} \otimes \, \mathscr{R} + \mathscr{R} \otimes \operatorname{Id}_\mathcal{A} \right ) \, ,
    \end{equation}
    \end{enumerate}
    i.e.\ \(\mathscr{R}\) is a linear projector turning the tuple \((\mathcal{A}, \mathscr{R})\) into a Rota--Baxter algebra.
\end{defn}

\enter

\begin{rem}
    The Rota--Baxter property of the renormalization scheme ensures that the renormalized Feynman rules \(\renFR\) as well as the counterterm map \(\CTmap\) are characters themselves, i.e.\ \(\Phi^\pm \in \CG\), for any given Feynman rules character \(\Phi \in \CG\).
\end{rem}

\enter

\begin{defn}[Algebraic Birkhoff decomposition] \label{defn:algebraic-birkhoff-decomposition}
    Given the situation of \Cref{defn:character-group}, let \(f \in \CG\) be a character and \(\mathscr{R} \in \operatorname{End}_\texttt{Vect} \left ( \AQ \right )\) be a renormalization scheme. Furthermore, we set
    \begin{subequations}
    \begin{align}
        \AQ_+ & \coloneq \operatorname{Ker} \left ( \mathscr{R} \right ) \, , \\
        \AQ_0 & \coloneq \operatorname{Im} \left ( \one_\AQ \right ) \, , \\
        \AQ_- & \coloneq \operatorname{Im} \left ( \mathscr{R} \right )
        \intertext{and finally what we call the \emph{unital counterterm algebra}\footnotemark}
        \AQ^\texttt{R} & \coloneq \AQ_0 \oplus \AQ_- \, .
    \end{align}
    \end{subequations}
    \footnotetext{With the following definitions, we could equivalently define \(\AQ^\texttt{R} \coloneq \operatorname{Im} \left ( f^- \right )\).}%
    Then, the \emph{algebraic Birkhoff decomposition} of \(f\) with respect to \(\mathscr{R}\) is the tuple \((f^+, f^-)\), with the following properties:
    \begin{subequations}
    \begin{align}
        f \equiv {f^-}^{(\star -1)} \star f^+ \quad & \iff \quad f^+ \equiv f^- \star f
        \intertext{and}
        f^\pm \bigl ( \widetilde{\HQ} \bigr ) & \subseteq \AQ_\pm \, , \label{eqn:algebraic-Birkhoff-decomposition-evaluation-property}
    \end{align}
    \end{subequations}
    where we have set \(\HQ \cong \HQ_0 \oplus \widetilde{\HQ}\), i.e.\ split \(\HQ\) into the image of the unit \(\HQ_0 \coloneq \operatorname{Im} \left ( \one_\HQ \right )\) and its non-trivial complement \(\widetilde{\HQ}\).
\end{defn}

\enter

\begin{rem}
    The reason for the splitting \(\HQ \cong \HQ_0 \oplus \widetilde{\HQ}\) in \Cref{eqn:algebraic-Birkhoff-decomposition-evaluation-property} is linked to the fact that, being characters, both maps \(f^\pm \in \CG\) map the unit \(\one_\HQ \in \HQ\) to the unit \(\one_\AQ \in \AQ\), i.e.\ \(f^\pm \left ( \one_\HQ \right ) = \one_\AQ \in \AQ_0\). Since the unit \(\one_\AQ\) corresponds to tree-level expressions, they are typically convergent. This suggests choosing renormalization schemes \(\mathscr{R}\) with the property \(\mathscr{R} \left ( \one_\AQ \right ) = 0\), which implies the inclusion \(\AQ_0 \subset \AQ_+\). Thus,  \Cref{eqn:algebraic-Birkhoff-decomposition-evaluation-property} is generally false on \(\HQ_0\).
\end{rem}

\enter

With this, we arrive finally at the following:

\henter

\begin{thm}[Algebraic Birkhoff decomposition, {\cite[Theorem 1]{Connes:1999zw}}] \label{thm:algebraic-birkhoff-decomposition}
    Given the situation of \Cref{defn:algebraic-birkhoff-decomposition}, then every \(f \in \CG\) admits a unique Birkhoff decomposition. Furthermore, the characters \(f^\pm \in \CG\) are given as follows:
    \begin{subequations}
    \begin{align}
        f^- & \equiv - \mathscr{R} \circ \overline{f}
        \intertext{and}
        f^+ & \equiv \left ( \operatorname{Id}_\AQ - \mathscr{R} \right ) \circ \overline{f}
        \intertext{using the \emph{Bogoliubov map}}
        \overline{f} & \coloneq f + f^- \tilde{\star} f \, ,
    \end{align}
    \end{subequations}
    where \(\tilde{\star}\) denotes the convolution product using the reduced coproduct \(\widetilde{\Delta}_\HQ\).
\end{thm}

\begin{proof}
    A proof can be found in \cite[Theorem 1]{Connes:1999zw}.
\end{proof}

\enter

Now, we can finally define the renormalized Feynman rules and the counterterm map as follows:

\enter

\begin{defn} \label{defn:renormalized-feynman-rules}
    Given the situation of \Cref{thm:algebraic-birkhoff-decomposition} and let \(\Phi \in \CG\) be the unrenormalized Feynman rules map. Then, the renormalized Feynman rules are given via
    \begin{subequations}
    \begin{align}
        \Phi_\mathscr{R} & \coloneq \Phi^+ \in \CG
        \intertext{and the counterterm map is given via}
        S_\mathscr{R}^\Phi & \coloneq \Phi^- \in \CG \, ,
    \end{align}
    \end{subequations}
    which is the commonly used notation for the renormalized Feynman rules and the counterterm map in the Connes--Kreimer community. We also remark that the counterterm map is sometimes also called \emph{twisted antipode}, as both maps share the same combinatorics.
\end{defn}

\enter

\begin{rem}
    Notice that the finiteness of the renormalized Feynman rules depends on the analytic properties of the renormalization scheme \(\mathscr{R} \in \operatorname{End}_\texttt{Vect} \left ( \AQ \right )\) from \Cref{defn:renormalization-scheme}: Specifically, we need to choose \(\mathscr{R}\) such that all divergent expressions of \(\AQ\) are contained in its image \(\AQ_- \coloneq \operatorname{Im} \left ( \mathscr{R} \right )\), but at the same time we wish to keep its kernel \(\AQ_+ \coloneq \operatorname{Ker} \left ( \mathscr{R} \right )\) as large as possible, as it corresponds to the finite and thus physically relevant expressions. Such renormalization schemes have been called \emph{proper renormalization schemes} in \cite[Definition 2.41]{Prinz:2019awo}.
\end{rem}

\enter

To illustrate this remark, we close this subsection with two extreme cases:

\henter

\begin{exmp}
    Consider the renormalization scheme \(\mathscr{R} \equiv 0\), i.e.\ the renormalization scheme that considers everything as convergent \(\AQ_+ \cong \AQ\) and \(\AQ_- \cong \{0\}\). Then, we obtain \(\renFR \equiv \Phi\) and \(\CTmap \equiv 0\).
\end{exmp}

\begin{exmp}
    Now, consider the renormalization scheme \(\mathscr{R} \equiv \operatorname{Id}_\AQ\), i.e.\ the renormalization scheme that considers everything as divergent \(\AQ_- \cong \AQ\) and \(\AQ_+ \cong \{0\}\). Then, we obtain \(\CTmap \equiv \Phi\) and \(\renFR \equiv 0\).
\end{exmp}

\subsection{Position space renormalization data} \label{ssec:position-space-data}

Our aim is to prove that the algebraic Birkhoff decomposition is capable of producing Epstein--Glaser renormalisation maps as in \Cref{def:renorm-maps}.
However, a slight modification to the scheme of \Cref{thm:algebraic-birkhoff-decomposition} is required.
We will construct a target algebra whose product is built on the tensor product of distributions, which is always well defined.
However, we require a second product for distributions which arise from graphs with overlapping vertices.
This product is only partially defined, and it only needs to be formed for specific terms arising in the algebraic Birkhoff decomposition; these terms are always well defined.
The modification this introduces does not affect the combinatorics, and we can reproduce the Connes--Kreimer picture with this double choice of products.

Recall that we defined the set of propagators in \Cref{def:propagators}.
For simplicity, we will assume from now on that we are dealing with a single type of propagator which we call $G \in \mathscr{D}'(M \times M)$.

We will use the Hadamard regularization procedure described in Appendix \ref{sec:apx-hadamard-finite-part}.
For every integer $n$, choose a cutoff function $\chi^{(n)} \in \mathscr{D}'(M^n)$ such that $\chi^{(n)} \equiv 1$ in a neighborhood of the small diagonal in $M^n$.
Then we put $T_\epsilon G = (1-\chi^{(2)}_\epsilon) G \in \mathscr{C}^\infty(M \times M)$.

We define the Feynman rules as a regularized version of \Cref{eq:unren-feynm-eg}. Recall that for any graph $\Gamma$ and any edge $e \in E(\Gamma)$, there are unique source and range vertices $s(e), r(e) \in V(\Gamma)$.

\enter

\begin{defn}
    The \emph{regularized Feynman rules} are given by a map $\Phi_\epsilon: \mathcal{H} \to \mathscr{C}^\infty(M^\bullet)$, defined by
    \begin{align*}
        \Phi_\epsilon(\Gamma)(x_1,\ldots, x_n) = \prod_{e \in E(\Gamma)} (T_\epsilon G)(x_{s(e)}, x_{r(e)}).
    \end{align*}
    for any graph $\Gamma$ with $n$ vertices.
\end{defn}

\enter

Our target algebra will consist of functions of one or several regularization parameters $\epsilon$ with values in distributions.
The behavior as $\epsilon \to 0$ must be controlled.
We do this by introducing the class of \emph{partially finite functions}.

\enter

\begin{defn}
To define the space $\mathsf{PF}_{a, M}$ we need to distinguish $a \in -\mathbb{N}$ and $a \in (\mathbb{C} \setminus -\mathbb{N})$.
A function $f: (0,1] \to \mathbb{C}$ is in $\widetilde{\mathsf{PF}}_{a,M}$ if for any $N\in \mathbb{N}$ there exist constants $c_{k,j}\in\mathbb{C}$ such that
\begin{align}
    \label{eq:pf-asymptotic-expansion}
    f(\epsilon) - \sum_{k=0}^N\sum_{j=0}^M c_{k,j} \epsilon^{a + k} \log^j\epsilon = O(\epsilon^{a+N+1}\log^M\epsilon).
\end{align}
Then for $a \in -\mathbb{N}$ we define $\mathsf{PF}_{a, M} := \widetilde{\mathsf{PF}}_{a, M}$, but for $a \in (\mathbb{C} \setminus -\mathbb{N})$, $f \in \mathsf{PF}_{a,M}$ if and only if $f = C + g$ with $C \in \mathbb{C}$ and $g \in \widetilde{\mathsf{PF}}_{a,M}$.

We will call $f \in \mathsf{PF}_{a,M}$ \emph{partially finite} of order $a$ and degree $M$.
\end{defn}

\enter

We extend this definition to distributions in the obvious weak sense.

\enter

\begin{defn}
    The space $\mathsf{PF}_{a,M}((0,1], \mathscr{D}'(M))$ is defined to be the space of functions $f: (0,1] \to \mathscr{D}'(M)$ such that for every test function $\phi \in \mathscr{D}(M)$, $\epsilon \mapsto \langle f(\epsilon), \phi\rangle$ is in $\mathsf{PF}_{a,M}$.
\end{defn}

\enter

To deal with overlapping divergences, we need to consider several regularization parameters.
To deal with the resulting expansion in several parameters, we need some additional notation.

\enter

\begin{defn}
    A \emph{power index} of length $n$ is a pair $I = (\alpha, \beta)$, where $\beta$ is a multi-index of length $n$, $\beta_i \in \mathbb{N}$, $i=1,\ldots, n$, and $\alpha \in \mathbb{C}^n$.
    We will write $n = |I| = |\alpha| = |\beta|$.
    
    For a power index $I$ with $|I| = n$, we define the standard function
    \begin{align*}
        p_{\alpha, \beta}(\epsilon_1, \ldots, \epsilon_n)
        := \epsilon_1^{\alpha_1}\ldots \epsilon_n^{\alpha_n}
        \log^{\beta_1} \epsilon_1 \ldots \log^{\beta_n} \epsilon_n.
    \end{align*}
    
    Let $\mathsf{Pow}_n$ be the set of power indices of length $n$.
    We define on $\mathsf{Pow}_n$ a relation given by lexicographical ordering, i.e.
    $(\alpha, \beta) < (\gamma, \delta)$ if $\alpha < \gamma$ or $\alpha = \gamma$ and $\beta < \delta$.
    Here, $\alpha < \gamma$ means that $\sum_{i=1}^n \alpha_i < \sum_{i=1}^n \gamma_i$, and $\beta < \delta$ likewise means $\sum_{i=1}^n \beta_i < \sum_{i=1}^n \delta_i$.
    We will say that
    \begin{align*}
        \operatorname{deg}(\alpha,\beta) < 0, \qquad \text{if}\ (\alpha,\beta) < (0,0)\\
        \operatorname{deg}(\alpha,\beta) > 0, \qquad \text{if}\ (\alpha,\beta) > (0,0)\\
        \operatorname{deg}(\alpha,\beta) =0, \qquad \text{otherwise.}
    \end{align*}
    
    We will say that $p_{\alpha,\beta}$ is \emph{singular} if $\operatorname{deg}(\alpha,\beta) < 0$ and it is \emph{regular} otherwise.
\end{defn}

\enter

\begin{defn}
    Let $(\alpha,\beta)$ be a power index.
    We will say that a function $f: (0,1]^n \to \mathscr{D}'(M)$ is in
    $\widetilde{\mathsf{PF}}_{\alpha,\beta}((0,1]^n, \mathscr{D}'(M))$ if for all
    test functions $\phi \in \mathscr{D}(M)$, we have
    \begin{align}
    \label{eq:partially-finite-several-parameters}
        \langle f(\epsilon), \phi\rangle 
        - \sum_{(\mathbf{a}, \gamma) \leq (N^{\times n}, \delta)} c_{\mathbf{a},\gamma}[\phi] p_{\alpha + \mathbf{a}, \gamma}(\epsilon_1, \ldots, \epsilon_n) + O(\epsilon^{\alpha + N^{\times n} + 1^{\times n}} \log^\beta \epsilon).
    \end{align}
    
Analogously to the one-parameter case we define $\mathsf{PF}_{\alpha,\beta}((0,1]^n, \mathscr{D}'(M))$ to include a constant (independent of $\epsilon$) term if it is not already present in the expansion.
\end{defn}

\enter

\begin{defn}
    Define $\mathcal{A}_{n, k}$ to be the set of all $f: (0,1]^n \to \mathscr{D}'(M^k)$ such that $f \in \mathsf{PF}_{\alpha,\beta}$ for some power index $(\alpha,\beta) \in \mathsf{Pow}_n$, such that $\alpha_i \in \mathbb{Z}$ for all $i=1,\ldots, n$.
    Define $\mathcal{A}_k := \bigoplus_{n=1}^\infty \mathcal{A}_{n,k}$ and $\mathcal{A} := \bigoplus_{k=1}^\infty \mathcal{A}_k$.
    Then $\mathcal{A}$ is an algebra with the product
    \begin{align*}
        m: \mathcal{A}_{n,k} \otimes \mathcal{A}_{n', k'} &\to \mathcal{A}_{n+n', k+k'}\\
        m(f, g)(\epsilon_1,\ldots, \epsilon_{n+n'})
        &= f(\epsilon_1, \ldots, \epsilon_n) \otimes g(\epsilon_{n+1}, \ldots, \epsilon_{n+n'}).
    \end{align*}
\end{defn}

\enter

\begin{defn}
    Let $\iota:J\hookrightarrow I$ be an inclusion of finite subsets of $\N$. We define a partial insertion product 
    \begin{equation*}
        \delta_{J\hookrightarrow I}:\mathcal{A}\times \mathcal{A} \dashrightarrow \mathcal{A}
    \end{equation*}
    by
    \begin{equation}
        \delta_{J\hookrightarrow I}(u,v)\coloneq u \otimes v \prod_{j\in J}\delta_{x_j,x_{\iota(j)}},
    \end{equation}
    whenever defined (in the sense of the Hörmander product).
\end{defn}

\enter

We can now define our renormalization scheme.

\enter

\begin{defn}
\label{def:hadamard-singular-part-scheme}
    Let $(\alpha, \beta) \in \mathsf{Pow}_n$ and $f \in  \operatorname{PF}_{\alpha,\beta}((0,1]^n, \mathscr{D}'(M^k))$ have the asymptotic expansion of \Cref{eq:partially-finite-several-parameters}.
    Then we define
    \begin{align*}
        \langle\mathscr{R} f(\epsilon_1,\ldots, \epsilon_n),\phi\rangle = \sum_{\operatorname{deg}(\gamma, \delta) < 0} c_{\gamma,\delta}[\phi] p_{\gamma,\delta}(\epsilon_1,\ldots, \epsilon_n),
    \end{align*}
    i.e. $\mathscr{R}$ keeps only the singular terms in the expansion of $f(\epsilon_1,\ldots, \epsilon_n)$.
\end{defn}

\enter

The following proposition guarantees that the operations in the algebraic Birkhoff decomposition are well defined:

\enter

\begin{prop}
	Let $\Gamma \in \mathcal{H}$ be a Feynman graph with $n$ vertices, then $\Phi_\epsilon(\Gamma) \in \mathcal{A}_n$. Furthermore, let $\Gamma_1, \Gamma_2 \in \mathcal{H}$ be Feynman graphs with $k$ and $l$ vertices, respectively, then $\Phi_\epsilon^-(\Gamma_1) \Phi_\epsilon(\Gamma_2) \in \mathcal{A}_{k+l}$.
\end{prop}

\subsection{Multiplicative renormalization} \label{ssec:multiplicative-renormalization}

Using the framework developed in the previous two subsections, we can now state the renormalized Lagrange density with the \(Z\)-factors of multiplicative renormalization: In a nutshell, the idea is that each monomial in the Lagrange density is multiplied via a divergent expression such that the Feynman rules derived from it give rise to finite integral expressions.

\enter

\begin{defn}[\(Z\)-factors]\label{defn:z-factors}
    Let \(\mathfrak{X}^r_L\) be the restricted combinatorial Green's function and \(\CTmap \colon \HQ \to \AQ_-\) be the counterterm map. Then, we define the counterterms for the external leg structure \(r\) and loop-grading \(L\) as follows:
    \begin{subequations}
    \begin{align}
        C^r_L & \coloneq \CTmap \left ( \mathfrak{X}^r_L \right )
        \intertext{and the corresponding \(Z\)-factors via}
        Z^r & \coloneq \CTmap \left ( \mathfrak{X}^r_L \right ) \equiv 1 \pm \sum_{L = 1}^\infty C^r_L
        \intertext{and finally}
        Z^r_L & \coloneq Z^r \Bigg \vert_L \equiv \pm C^r_L \, ,
    \end{align}
    \end{subequations}
    where plus corresponds to vertex residues and minus to propagator residues, cf.\ \Cref{eqn:full-combinatorial-greens-function}.\footnote{The reason for the minus sign is that for edge residues \(e\) the fractions \(1 / \mathfrak{X}^e_L\) can be interpreted as formal geometric series and thus producing \emph{connected} Green's functions out of the \emph{1PI} Green's functions \(\mathfrak{X}^e_L\) and likewise for the corresponding \(Z\)-factors \(Z^e\).}
\end{defn}

\enter

\begin{defn}[Multiplicatively renormalized Lagrange density] \label{defn:mult-ren-lagr-dens}
    Let \(\mathcal{L}_\QFT\) denote the Lagrange density of a Quantum Field Theory \(\QFT\), decomposed into monomials \(\mathcal{M}^r\) describing the classical interaction for the amplitude \(r\), i.e.
    \begin{subequations}
    \begin{align}
        \LQ & \equiv \sum_{r \in \RQ} \mathcal{M}^r \, ,
        \intertext{where \(\RQ\) denotes the set of all interactions of \(\QFT\). Then, its renormalized variant is given via}
        \LQR & \coloneq \sum_{r \in \RQ} Z^r \mathcal{M}^r \, ,
    \end{align}
    \end{subequations}
    where the \(Z^r\)'s are the \(Z\)-factors from \Cref{defn:z-factors}. Again, this expression needs to be considered formal, unless a regularization scheme has been used, in which case the renormalized Lagrange density is a function in the regulator \(\epsilon\) with a pole in \(\epsilon = 0\).
\end{defn}

\enter

With this setup, the main result of multiplicative renormalization reads as follows:

\enter

\begin{thm}[Multiplicative renormalization] \label{thm:multiplicative-renormalization}
    The Feynman rules calculated from the renormalized Lagrange density \(\Phi^\emph{\texttt{R}}\) reproduce the renormalized Feynman rules of the algebraic Birkhoff decomposition \(\renFR\) when evaluated on combinatorial Green's function \(\mathfrak{X}^r\). Specifically, we have
    \begin{equation}
        \Phi^\emph{\texttt{R}} \left ( \mathfrak{X}^r \right )_L = \renFR \left ( \mathfrak{X}^r_L \right ) \, ,
    \end{equation}
    where the subscript \(L\) denotes restriction to a specific loop number.\footnote{Or a finer grading, see the remark at the beginning of the proof.}
\end{thm}

\begin{proof}
    We state the proof for scalar theories with one interaction term only: This means that we can work with the loop-grading for Feynman graphs, cf.\ \cite[Remark 3.31]{Prinz:2018dhq}. The general case can then be derived in a similar manner using either the vertex-grading or the coupling-grading, cf.\ \cite[Definition 2.18]{Prinz:2019awo}. The main point is that we need a grading that is fine enough such that it uniquely determines the edge- and vertex-types for a given residue \(r \in \RQ\) in the sense of \cite[Lemma 2.27]{Prinz:2019awo}. In the scalar case with only one interaction term this is given via the loop-grading using the Euler characteristic and noting that all vertices are univalent: To see this explicitly, fix any Feynman graph \(\Gamma \in \GQ\) and consider the Euler characteristic
    \begin{equation} \label{eqn:euler-characteristic}
        b_1 \left ( \Gamma \right ) - b_0 \left ( \Gamma \right ) = \mathfrak{e} - \mathfrak{v} \, ,
    \end{equation}
    where \(b_1 \left ( \Gamma \right )\) is its first Betti number, i.e.\ number of cycles, \(b_0 \left ( \Gamma \right )\) is its zeroth Betti number, i.e.\ number of connected components. Furthermore, \(\mathfrak{e} \coloneq \# \Gamma^{[1]}\) denotes the number of edges and \(\mathfrak{v} \coloneq \# \Gamma^{[0]}\) the number of vertices. As in \(\phi^m\)-theory every vertex is \(m\)-valent and its residue \(n \coloneq \left \vert r \right \vert\) fixes the number of external half-edges, we can express the number of internal edges in terms of the number of vertices via
    \begin{equation} \label{eqn:edge-vertex-relation}
        \mathfrak{e} = \frac{1}{2} \left ( m \mathfrak{v} - n \right ) \, .
    \end{equation}
    Inserting \Cref{eqn:edge-vertex-relation} into \Cref{eqn:euler-characteristic} and setting \(b_0 \left ( \Gamma \right ) = 1\), which holds obviously for connected graphs, we obtain
    \begin{subequations} \label{eqns:vertex-and-edge-sets}
    \begin{align}
        \mathfrak{v} & = \frac{2 \left ( b_1 \left ( \Gamma \right ) - 1 \right ) + n}{m - 2}
        \intertext{and}
        \mathfrak{e} & = \frac{m \left ( b_1 \left ( \Gamma \right ) - 1 \right ) + n}{m - 2} \, .
    \end{align}
    \end{subequations}
    Since these expressions depend only on the residue \(r\) and the loop number, both expressions also hold for the restricted combinatorial Green's functions \(\mathfrak{X}^r_L\). Now, the Feynman rules derived from the multiplicatively renormalized Lagrange density \(\LQR\) for the residue \(r\) are given via
    \begin{equation}
    \begin{split}
        \mrFR \left ( r \right ) & \coloneq Z^r \Phi \left ( r \right ) \\
        & \equiv \Phi \left ( r \right ) \left ( 1 \pm \sum_{L = 1}^\infty C^r_L \right ) \, ,
    \end{split}
    \end{equation}
    where \(\Phi \left ( r \right )\) denotes the unrenormalized Feynman rules of edges and vertices, \(Z^r \equiv \left ( 1 \pm \sum_{L = 1}^\infty C^r_L \right )\) the corresponding \(Z\)-factor with \(C^r_L\) the counterterm for loop order \(L\). Using the above derivations, we obtain for the restricted combinatorial Green's functions
    \begin{subequations}
    \begin{align}
        \mrFR \left ( \mathfrak{X}^r \right )_L = \sum_{l = 0}^L \left ( \frac{\left ( Z^v \right )^{\mathfrak{v} \left ( r, l \right )}}{\left ( Z^e \right )^{\mathfrak{e} \left ( r, l \right )}} \Phi \left ( \mathfrak{X}^r_l \right ) \right )_L \, ,
    \end{align}
    \end{subequations}
    where \(e\) is the edge-type of the theory and \(v\) its vertex-type. Specifically, we remark that the counterterms in the \(Z\)-factors correspond itself to loop-expressions, such that the global restriction to loop number \(L\) is understood as the sum of the contributions. Now, we want to show that the left hand side exactly corresponds to the positive contribution of the algebraic Birkhoff decomposition in the Connes--Kreimer framework, i.e.
    \begin{equation} \label{eqn:renormalized-feynman-rules}
        \mrFR \left ( \mathfrak{X}^r \right )_L \equiv \renFR \left ( \mathfrak{X}^r_L \right ) \, .
    \end{equation}
    Indeed, using the \emph{restricted coproduct formula}, given in \cite[Lemma 4.6]{Yeats:2008zy} and \cite[Proposition 4.2]{Prinz:2019awo},\footnote{We remark that the complete coproduct formula, i.e.\ without restriction to a specific loop number, has been additionally proven in \cite[Proposition 16]{vanSuijlekom:2006fk} and \cite[Theorem 1]{Borinsky:2014xwa}. Furthermore, the (restricted) version for connected Green's functions is provided in \cite[Proposition 2.15]{Prinz:2025obr}.}
    \begin{equation} \label{eqn:restricted-coproduct-formula}
        \Delta \left ( \mathfrak{X}^r_L \right ) = \sum_{l = 1}^L \pol \otimes^\delta_\mathbb{Q} \mathfrak{X}^r_{L - l} \, ,
    \end{equation}
    where \(\pol\) is a polynomial in Feynman graphs with loop number \(l\) which depends on \(r\): Specifically, we can express \(\pol\) in terms of \emph{combinatorial charges}, \cite[Lemma 4.6]{Yeats:2008zy} and \cite[Proposition 4.2]{Prinz:2019awo}, which are formal rational functions in the combinatorial Green's functions. Referring to the cited literature, we obtain
    \begin{equation}
        \pol \equiv \left ( \frac{\left ( \mathfrak{X}^v \right )^{\mathfrak{v} \left ( r, L-l \right )}}{\left ( \mathfrak{X}^e \right )^{\mathfrak{e} \left ( r, L-l \right )}} \right ) \! \Bigg \vert_{l} \, .
    \end{equation}
    With this, we obtain
    \begin{equation}
    \begin{split}
        \renFR \left ( \mathfrak{X}^r_L \right ) & = \left ( \CTmap \star \Phi \right ) \left ( \mathfrak{X}^r_L \right ) \\
        & = \sum_{l = 1}^L \CTmap \bigl ( \pol \bigr ) \Phi \left ( \mathfrak{X}^r_{L - l} \right ) \\
        & = \sum_{l = 1}^L \left ( \frac{\left ( Z^v \right )^{\mathfrak{v} \left ( r, L-l \right )}}{\left ( Z^e \right )^{\mathfrak{e} \left ( r, L-l \right )}} \right ) \! \Bigg \vert_{l} \Phi \left ( \mathfrak{X}^r_{L - l} \right ) \\
        & = \mrFR \left ( \mathfrak{X}^r \right )_L \, . \vphantom{\Big \vert}
    \end{split}
    \end{equation}
    Summing over all \(L\) now indeed verifies \Cref{eqn:renormalized-feynman-rules}, and thus completes the proof.
\end{proof}

\section{Example: Massless \texorpdfstring{\(\phi^3_6\)}{phi-3-6}-theory} \label{sec:examples}

To fill the above construction with life, we apply them to the simplest nontrivial example: Massless \(\phi^3_6\)-theory. This theory is very well suited as a pedagogical example, since the graphs are quite easy to visualize. However, it is not a physical theory, since its Hamiltonian is unbounded from below, meaning that the theory does not allow for a stable vacuum state.\footnote{Nevertheless, we emphasize the relation of \(\phi^3_4\)-theory amplitudes to Quantum Yang--Mills theory amplitudes via the \emph{Corolla polynomial} and the \emph{Corolla differential}, cf.\ \cite{Kreimer:2012eh,Kreimer:2012jw,Prinz:2016fka}.} Furthermore, being massless avoids the necessity of introducing further renormalization conditions which relate the quadratic \(Z\)-factor to the kinetic and the mass \(Z\)-factor. Additionally, it ensures that we obtain a homogenous scaling for the propagator distributions --- at the cost of introducing additional infrared-singularities, which we ignore in this article. We consider spacetime dimension \(d = 6\) such that the theory is renormalizable.\footnote{The Feynman diagrams in this section are drawn using \texttt{FeynGame} \cite{Harlander:2020cyh,Harlander:2024qbn,Bundgen:2025utt}.}

\subsection{Lagrange density}

This theory is given via the following unrenormalized Lagrange density:
\begin{equation} \label{eqn:unrenormalized-phi-3-6-theory}
    \mathcal{L}_{\phi^3_6} = \frac{1}{2} \left ( \partial_\mu \phi_0 \right ) \left ( \partial^\mu \phi_0 \right ) - \frac{\lambda_0}{3!} \phi_0^3 \, ,
\end{equation}
In the process of multiplicative renormalization, the field and the coupling constant are rescaled as follows: We have the \emph{wavefunction renormalization}
\begin{subequations}
\begin{align}
    \phi_0 & \rightsquigarrow \phi \coloneq \sqrt{Z_\phi} \phi_0
    \intertext{and the \emph{coupling constant renormalization}}
    \lambda_0 & \rightsquigarrow \lambda \coloneq \sqrt{Z_\lambda} \lambda_0 \, .
\end{align}
\end{subequations}
This then leads to the multiplicatively renormalized Lagrange density:
\begin{equation} \label{eqn:renormalized-phi-3-6-theory}
    \mathcal{L}_{\phi^3_6}^\texttt{R} = \frac{Z_\text{Kin}}{2} \left ( \partial_\mu \phi_0 \right ) \left ( \partial^\mu \phi_0 \right ) - \frac{Z_\text{Int} \lambda_0}{3!} \phi_0^3 \, ,
\end{equation}
where we have set
\begin{subequations}
\begin{align}
    Z_\text{Kin} & \coloneq Z_\phi
    \intertext{and}
    Z_\text{Int} & \coloneq Z_\phi^{3/2} Z_\lambda \, .
\end{align}
\end{subequations}
As we will see in this subsection, this renormalized Lagrange density leads to finite Feynman rules.

\subsection{Feynman rules}

Starting from the Lagrange density given in \Cref{eqn:unrenormalized-phi-3-6-theory}, we obtain the following Feynman rules:
\begin{subequations} \label{eqn:feynman-rules}
\begin{align}
    \Phi \left (x\ \propagator\ y\right ) & \coloneq \frac{1}{|x-y|^{4}}
    \intertext{and}
    \Phi \left ( \vertex \right ) & \coloneq \lambda
\end{align}
\end{subequations}

\subsection{Combinatorial Green's functions}

From the Feynman rules, given in \Cref{eqn:feynman-rules}, we obtain the following propagator Green's functions with one and two loops:\footnote{We remark that the symmetry factor for the double bubble is actually \(1/2\), however we have displayed the upper and lower insertions for a symmetric presentation, which are actually equivalent as Feynman graphs, i.e. \begin{equation} \bubbleupperedgecorrectiongraphsmall \cong \bubbleloweredgecorrectiongraphsmall \end{equation} and thus resulting in the prefactor \(1/4\).}
{\allowdisplaybreaks
\begin{subequations} \label{eqn:propagator-greens-functions}
\begin{align}
    \mathfrak{X}^{\propagatorresidue}_1 & \coloneq \frac{1}{2} \bubblegraph \\
    \begin{split}
        \mathfrak{X}^{\propagatorresidue}_2 & \coloneq \bubblevertexcorrectiongraph + \frac{1}{4} \left ( \bubbleupperedgecorrectiongraph + \bubbleloweredgecorrectiongraph \right )
    \end{split}
\end{align}
\end{subequations}
}
Additionally, we obtain the following vertex Green's functions with one and two loops:
{\allowdisplaybreaks
\begin{subequations} \label{eqn:vertex-greens-functions}
\begin{align}
    \mathfrak{X}^{\vertexresidue}_1 & \coloneq \trianglegraph \\
    \begin{split}
        \mathfrak{X}^{\vertexresidue}_2 & \coloneq \triangleleftvertexcorrectiongraph + \triangleuppervertexcorrectiongraph + \trianglelowervertexcorrectiongraph \\ & \phantom{\coloneq} + \frac{1}{2} \left ( \trianglerightedgecorrectiongraph + \triangleupperedgecorrectiongraph + \triangleloweredgecorrectiongraph \right ) \\ & \phantom{\coloneq} + \trianglenonplanar
    \end{split}
\end{align}
\end{subequations}
}

\subsection{Coproduct identities} \label{ssec:coproduct-identities-examples}

Ultimately, we want to apply the algebraic Birkhoff decomposition to the combinatorial Green's functions, given in \Cref{eqn:propagator-greens-functions} and \Cref{eqn:propagator-greens-functions} in order to obtain the renormalized Feynman rules map and the coproduct map. To this end, we observe that the one-loop graphs are primitive elements in the renormalization Hopf algebra, i.e.\ satisfy \(\Delta \left ( \Gamma \right ) = \one \otimes_\mathbb{Q} \Gamma + \Gamma \otimes_\mathbb{Q} \one\), and we obtain
\begin{subequations} \label{eqn:coproduct-identities-oneloop}
\begin{align}
    \Delta \left ( \mathfrak{X}^{\propagatorresidue}_1 \right ) & = \one \otimes_\mathbb{Q} \mathfrak{X}^{\propagatorresidue}_1 + \mathfrak{X}^{\propagatorresidue}_1 \otimes_\mathbb{Q} \one
    \intertext{and}
    \Delta \left ( \mathfrak{X}^{\vertexresidue}_1 \right ) & = \one \otimes_\mathbb{Q} \mathfrak{X}^{\vertexresidue}_1 + \mathfrak{X}^{\vertexresidue}_1 \otimes_\mathbb{Q} \one \, .
\end{align}
\end{subequations}
The coproduct formula for the two-loop variants is a little bit more involved and, as a preparation therefor, we compute the reduced coproduct \(\widetilde{\Delta}\), i.e.\ the non-trivial contributions omitting the \(\one \otimes_\mathbb{Q} \Gamma + \Gamma \otimes_\mathbb{Q} \one\) terms. Furthermore, we write \(\otimes_\mathbb{Q}^\delta\) as a shorthand for the tensor product together with the identification map \(\delta\):
{\allowdisplaybreaks
\begin{subequations}
\begin{align}\label{eqn:coproduct-of-two-loop-graphs}
    \begin{split}
        \widetilde{\Delta} \left ( \bubblevertexcorrectiongraph \right ) & = \trianglegraph \otimes_\mathbb{Q}^\delta \frac{1}{2} \bubblevertexcorrectiongraphdeletedleft \\
        & \phantom{=} + \trianglegraphmirrored \otimes_\mathbb{Q}^\delta \frac{1}{2} \bubblevertexcorrectiongraphdeletedright
    \end{split} \\ 
        \widetilde{\Delta} \left ( \frac{1}{2} \bubbleupperedgecorrectiongraph \right ) & = \frac{1}{2} \bubblegraph \otimes_\mathbb{Q}^\delta \bubbleupperedgecorrectiondeletedgraph  \\
    \widetilde{\Delta} \left ( \triangleleftvertexcorrectiongraph \right ) & = \trianglegraph \otimes_\mathbb{Q}^\delta \triangleleftvertexcorrectiondeletedgraph \\
    \widetilde{\Delta} \left ( \frac{1}{2} \trianglerightedgecorrectiongraph \right ) & = \frac{1}{2} \bubblegraph \otimes_\mathbb{Q}^\delta \trianglerightedgecorrectiondeletedgraph
\end{align}
\end{subequations}
}%
We remark that the multiplicities and symmetry factors conveniently combine on the left and right hand side, which is an example of the general results \cite[Lemma 12]{vanSuijlekom:2006fk} and \cite[Lemma 2.34]{Prinz:2019awo}. Furthermore, we emphasize that from a renormalization perspective the orientation of a subgraph is irrelevant, as its kinematics are typically set to a symmetric reference value. Thus, we will treat subgraphs coming from different insertion points as equivalent, which gives the multiplicities occurring in the following formula as the \emph{number of insertion points}. From this, we finally obtain
\begin{subequations} \label{eqn:coproduct-identities-twoloop}
\begin{align}
    \Delta \left ( \mathfrak{X}^{\propagatorresidue}_2 \right ) & = \one \otimes_\mathbb{Q} \mathfrak{X}^{\propagatorresidue}_2 + 2 \mathfrak{X}^{\propagatorresidue}_1 \otimes_\mathbb{Q}^\delta \mathfrak{Y}^{\propagatorresidue}_{1,1,0} + 2 \mathfrak{X}^{\vertexresidue}_1 \otimes_\mathbb{Q}^\delta \mathfrak{Y}^{\propagatorresidue}_{1,0,1} + \mathfrak{X}^{\propagatorresidue}_2 \otimes_\mathbb{Q} \one \\
    \Delta \left ( \mathfrak{X}^{\vertexresidue}_2 \right ) & = \one \otimes_\mathbb{Q} \mathfrak{X}^{\vertexresidue}_2 + 3 \mathfrak{X}^{\propagatorresidue}_1 \otimes_\mathbb{Q}^\delta \mathfrak{Y}^{\vertexresidue}_{1,1,0} + \mathfrak{X}^{\vertexresidue}_1 \otimes_\mathbb{Q}^\delta \mathfrak{Y}^{\vertexresidue}_{1,0,1} + \mathfrak{X}^{\vertexresidue}_2 \otimes_\mathbb{Q} \one
\end{align}
\end{subequations}
This is an example of the \emph{restricted coproduct formula}, given in \cite[Lemma 4.6]{Yeats:2008zy} and \cite[Proposition 4.2]{Prinz:2019awo}, which was mentioned in \Cref{eqn:restricted-coproduct-formula}.

\subsection{Counterterms, renormalized Feynman rules and \texorpdfstring{\(Z\)}{Z}-factors} \label{ssec:counterterms-renormalized-feynman-rules-z-factors}

Using the results from \Cref{eqn:coproduct-identities-oneloop} and \Cref{eqn:coproduct-identities-twoloop}, we can now calculate the renormalized Feynman rules and the counterterm map, which gives us the corresponding \(Z\)-factors: In the following, we denote via \(\Phi_\mathscr{A} := \Phi \circ \mathscr{A}\) the Feynman rules that vanish on trivial graphs, i.e.\ \(\mathscr{A}\) is the projector onto the augmentation ideal. Furthermore, it is easy to see that on primitive elements \(\mathfrak{G}\) the counterterm map simplifies to \(\CTmap \left ( \mathfrak{G} \right ) \equiv - \left [ \mathscr{R} \circ \Phi \right ] \left ( \mathfrak{G} \right )\). In the following, we denote the counterterms for residue \(r\) and loop order \(L\) via \(C^r_L\) and its product with the corresponding form factor via \(C^r_L \coloneq \Phi \left ( r \right ) \cdot_\AQ C^r_L\). Thus, we obtain the following expressions for the one-loop counterterm:
\begin{subequations} \label{eqn:counterterm-map-oneloop}
\begin{align}
    \begin{split}
        C^{\propagatorresidue}_1 & \coloneq \CTmap \left ( \mathfrak{X}^{\propagatorresidue}_1 \right ) \\ & \phantom{:} = - \mathscr{R} \left ( \left [ \CTmap \star \Phi_\mathscr{A} \right ] \left ( \mathfrak{X}^{\propagatorresidue}_1 \right ) \right ) \\ & \phantom{:} = - \left [ \mathscr{R} \circ \Phi \right ] \left ( \mathfrak{X}^{\propagatorresidue}_1 \right )
    \end{split}
    \intertext{and}
    \begin{split}
        C^{\vertexresidue}_1 & \coloneq \CTmap \left ( \mathfrak{X}^{\vertexresidue}_1 \right ) \\ & \phantom{:} = - \mathscr{R} \left ( \left [ \CTmap \star \Phi_\mathscr{A} \right ] \left ( \mathfrak{X}^{\vertexresidue}_1 \right ) \right ) \\ & \phantom{:} = - \left [ \mathscr{R} \circ \Phi \right ] \left ( \mathfrak{X}^{\vertexresidue}_1 \right ) \, .
    \end{split}
\end{align}
\end{subequations}
In addition, we set
\begin{subequations}
\begin{align}
    Y^{\propagatorresidue}_{i,j,k} & \coloneq \Phi \left ( \mathfrak{Y}^{\propagatorresidue}_{i,j,k} \right )
    \intertext{and}
    Y^{\vertexresidue}_{i,j,k} & \coloneq \Phi \left ( \mathfrak{Y}^{\vertexresidue}_{i,j,k} \right ) \, ,
\end{align}
\end{subequations}
where \(i\) denotes the number of loops, \(j\) the number of edge-cuttings and \(k\) the number of vertex removals. Then, we obtain the following expressions for the two-loop counterterm:
\begin{subequations} \label{eqn:counterterm-map-twoloop}
\begin{align}
    \begin{split}
        C^{\propagatorresidue}_2 & \coloneq \CTmap \left ( \mathfrak{X}^{\propagatorresidue}_2 \right ) \\
        & \phantom{:} = - \mathscr{R} \left ( \left [ \CTmap \star \Phi_\mathscr{A} \right ] \left ( \mathfrak{X}^{\propagatorresidue}_2 \right ) \right )
        \\ & \phantom{:} = 2 \, \mathscr{R} \left ( \left [ \mathscr{R} \circ \Phi \right ] \left ( \mathfrak{X}^{\propagatorresidue}_1 \cdot_\AQ Y^{\propagatorresidue}_{1,0,1} + \mathfrak{X}^{\vertexresidue}_1 \cdot_\AQ Y^{\propagatorresidue}_{1,0,1} \right ) \right ) - \left [ \mathscr{R} \circ \Phi \right ] \left ( \mathfrak{X}^{\propagatorresidue}_2 \right )
        \\ & \phantom{:} = 2 \, \mathscr{R} \left ( C^{\propagatorresidue}_1 \cdot_\AQ Y^{\propagatorresidue}_{1,1,0} + C^{\vertexresidue}_1 \cdot_\AQ Y^{\propagatorresidue}_{1,0,1} \right ) - \left [ \mathscr{R} \circ \Phi \right ] \left ( \mathfrak{X}^{\propagatorresidue}_2 \right )
    \end{split}
    \intertext{and}
    \begin{split}
        C^{\vertexresidue}_2 & \coloneq \CTmap \left ( \mathfrak{X}^{\vertexresidue}_2 \right ) \\
        & \phantom{:} = - \mathscr{R} \left ( \left [ \CTmap \star \Phi_\mathscr{A} \right ] \left ( \mathfrak{X}^{\vertexresidue}_2 \right ) \right ) \\
        & \phantom{:} = 3 \, \mathscr{R} \left ( \left [ \mathscr{R} \circ \Phi \right ] \left ( \mathfrak{X}^{\propagatorresidue}_1 \cdot_\AQ Y^{\vertexresidue}_{1,1,0} + \mathfrak{X}^{\vertexresidue}_1 \cdot_\AQ Y^{\vertexresidue}_{1,0,1} \right ) \right ) - \left [ \mathscr{R} \circ \Phi \right ] \left ( \mathfrak{X}^{\vertexresidue}_2 \right ) \\
        & \phantom{:} = 3 \, \mathscr{R} \left ( C^{\propagatorresidue}_1 \cdot_\AQ Y^{\vertexresidue}_{1,1,0} + C^{\vertexresidue}_1 \cdot_\AQ Y^{\vertexresidue}_{1,0,1} \right ) - \left [ \mathscr{R} \circ \Phi \right ] \left ( \mathfrak{X}^{\vertexresidue}_2 \right ) \, ,
    \end{split}
\end{align}
\end{subequations}
where we have emphasized the product in \(\mathcal{A}\) via \(\cdot_\AQ\). Building on this, we obtain the following expressions for the one-loop renormalized Feynman rules:
\begin{subequations} \label{eqn:renormalized-feynman-rules-oneloop}
\begin{align}
    \begin{split}
        \renFR \left ( \mathfrak{X}^{\propagatorresidue}_1 \right ) & = \left [ \CTmap \star \Phi \right ] \left ( \mathfrak{X}^{\propagatorresidue}_1 \right ) \\
        & = \Phi \left ( \mathfrak{X}^{\propagatorresidue}_1 \right ) - \left [ \mathscr{R} \circ \Phi \right ] \left ( \mathfrak{X}^{\propagatorresidue}_1 \right ) \\
        & = \left ( \operatorname{Id}_\AQ - \mathscr{R} \right ) \left ( \Phi \left ( \mathfrak{X}^{\propagatorresidue}_1 \right ) \right ) \\
        & = \Phi \left ( \mathfrak{X}^{\propagatorresidue}_1 \right ) + C^{\propagatorresidue}_1
    \end{split}
    \intertext{and}
    \begin{split}
        \renFR \left ( \mathfrak{X}^{\vertexresidue}_1 \right ) & = \left [ \CTmap \star \Phi \right ] \left ( \mathfrak{X}^{\propagatorresidue}_1 \right ) \\
        & = \Phi \left ( \mathfrak{X}^{\vertexresidue}_1 \right ) - \left [ \mathscr{R} \circ \Phi \right ] \left ( \mathfrak{X}^{\vertexresidue}_1 \right ) \\
        & = \left ( \operatorname{Id}_\AQ - \mathscr{R} \right ) \left ( \Phi \left ( \mathfrak{X}^{\vertexresidue}_1 \right ) \right ) \\
        & = \Phi \left ( \mathfrak{X}^{\vertexresidue}_1 \right ) + C^{\vertexresidue}_1
    \end{split}
\end{align}
\end{subequations}
and, where again \(\cdot_\AQ\) emphasizes the product in \(\AQ\), the following expressions for the two-loop renormalized Feynman rules
\begin{subequations} \label{eqn:renormalized-feynman-rules-twoloop}
\begin{align}
    \begin{split}
        \renFR \left ( \mathfrak{X}^{\propagatorresidue}_2 \right ) & = \left [ \CTmap \star \Phi \right ] \left ( \mathfrak{X}^{\propagatorresidue}_2 \right ) \\
        & = \Phi \left ( \mathfrak{X}^{\propagatorresidue}_2 \right ) + 2 \, C^{\propagatorresidue}_1 \cdot_\AQ Y^{\propagatorresidue}_{1, 1, 0} + 2 \, C^{\vertexresidue}_1 \cdot_\AQ Y^{\propagatorresidue}_{1,0,1} + C^{\propagatorresidue}_2
    \end{split}
    \intertext{and}
    \begin{split}
        \renFR \left ( \mathfrak{X}^{\vertexresidue}_2 \right ) & = \left [ \CTmap \star \Phi \right ] \left ( \mathfrak{X}^{\vertexresidue}_2 \right ) \\
        & = \Phi \left ( \mathfrak{X}^{\vertexresidue}_2 \right ) + 3 \, C^{\propagatorresidue}_1 \cdot_\AQ Y^{\vertexresidue}_{1,1,0} + 3 \, C^{\vertexresidue}_1 \cdot_\AQ Y^{\vertexresidue}_1 + C^{\vertexresidue}_2
    \end{split}
\end{align}
\end{subequations}
Finally, the \(Z\)-factors are given via the application of the counterterm map on the combinatorial Green's function:\footnote{If a mass term would be present, we would need to demand an additional renormalization condition.}
\begin{subequations}
\begin{align}
    Z_\text{Kin} & \coloneq \frac{1}{1 - \sum_{L = 1}^\infty C^{\propagatorresidue}_L}
    \intertext{and}
    Z_\text{Int} & \coloneq 1 + \sum_{L = 1}^\infty C^{\vertexresidue}_L \, ,
\end{align}
\end{subequations}
where the fraction in the definition of \(Z_\text{Kin}\) gives rise to the geometric series, accounting for the fact that some edge-corrections are also connected, but not 1PI.\footnote{This effect starts to appear at three loops and thus does not show up in the present example.}

\subsection{Feynman rules from the renormalized Lagrange density}

Using the multiplicatively renormalized Lagrange density, given in \Cref{eqn:renormalized-phi-3-6-theory}, we obtain the following multiplicatively renormalized Feynman rules:
\begin{subequations} \label{eqn:multiplicatively-renormalized-feynman-rules}
\begin{align}
    \begin{split}
    \Phi^\texttt{R} \left ( \propagator \right ) & \coloneq \frac{1}{ Z_\text{Kin} \, |x-y|^{4}} \vphantom{\Bigg \vert} \\
    & \phantom{:} \equiv \frac{1}{|x-y|^{4}} \left ( 1 + \sum_{k = 1}^\infty \Bigl ( \sum_{L = 1}^\infty C^{\propagatorresidue}_L \Bigr )^k \right ) \\
    & \eqcolon \mathcal{Z}_\text{Kin} \vphantom{\frac{1}{|x-y|^{4}}}
    \end{split}
    \intertext{and}
    \begin{split}
    \Phi^\texttt{R} \left ( \vertex \right ) & \coloneq \lambda Z_\text{Int} \\[-15pt]
    & \phantom{:} \equiv \lambda \left ( 1 + \sum_{L = 1}^\infty C^{\vertexresidue}_L \right ) \\
    & \eqcolon \mathcal{Z}_\text{Int}
    \end{split}
\end{align}
\end{subequations}
In particular, on the level of the combinatorial Green's function, we obtain the following:
\begin{subequations}
\begin{align}
    \begin{split}
    \Phi^\texttt{R} \left ( \mathfrak{X}^{\propagatorresidue} \right ) & = \sum_{l = 0}^\infty \Phi^\texttt{R} \left ( \mathfrak{X}^{\propagatorresidue}_l \right ) \\
    & = \sum_{l = 0}^\infty \left ( \Phi \left ( \mathfrak{X}^{\propagatorresidue}_l \right ) \left ( Z^{\propagatorresidue} \right )^{- \# E \left ( r, l \right )} \left ( Z^{\vertexresidue} \right )^{\# V \left ( r, l \right )} \right ) \\
    & = \sum_{l = 0}^\infty \left ( \Phi \left ( \mathfrak{X}^{\propagatorresidue}_l \right ) \left ( Z^{\propagatorresidue} \right )^{1 - 3l} \left ( Z^{\vertexresidue} \right )^{2l} \right )
    \end{split}
    \intertext{and}
    \begin{split}
    \Phi^\texttt{R} \left ( \mathfrak{X}^{\vertexresidue} \right ) & = \sum_{l = 0}^\infty \Phi^\texttt{R} \left ( \mathfrak{X}^{\vertexresidue}_l \right ) \\
    & = \sum_{l = 0}^\infty \left ( \Phi \left ( \mathfrak{X}^{\vertexresidue}_l \right ) \left ( Z^{\propagatorresidue} \right )^{- \# E \left ( r, l \right )} \left ( Z^{\vertexresidue} \right )^{\# V \left ( r, l \right )} \right ) \\
    & = \sum_{l = 0}^\infty \left ( \Phi \left ( \mathfrak{X}^{\vertexresidue}_l \right ) \left ( Z^{\propagatorresidue} \right )^{- 3l} \left ( Z^{\vertexresidue} \right )^{1 + 2l} \right ) \, ,
    \end{split}
\end{align}
\end{subequations}
where we have used \Cref{eqns:vertex-and-edge-sets} with \(b_1 \left ( \Gamma \right ) = l\), \(m = 3\) and \(n = 2\) for \(\mathfrak{X}^{\propagatorresidue}\) and \(n = 3\) for \(\mathfrak{X}^{\vertexresidue}\), respectively. Now, considering the restriction to loop number \(L\), we obtain the following:
\begin{subequations}
\begin{align}
    \begin{split}
    \Phi^\texttt{R} \left ( \mathfrak{X}^{\propagatorresidue} \right ) \Bigg \vert_L & = \sum_{l = 0}^\infty \left ( \Phi \left ( \mathfrak{X}^{\propagatorresidue}_l \right ) \left ( 1 - \sum_{m = 1}^\infty C^{\propagatorresidue}_m \right )^{1 - 3l} \left ( 1 + \sum_{n = 1}^\infty C^{\vertexresidue}_n \right )^{2 l} \right ) \Bigg \vert_L
    \end{split}
    \intertext{and}
    \begin{split}
    \Phi^\texttt{R} \left ( \mathfrak{X}^{\vertexresidue} \right ) \Bigg \vert_L & = \sum_{l = 0}^\infty \left ( \Phi \left ( \mathfrak{X}^{\vertexresidue}_l \right ) \left ( 1 - \sum_{m = 1}^\infty C^{\propagatorresidue}_m \right )^{-3l} \left ( 1 + \sum_{n = 1}^\infty C^{\vertexresidue}_n \right )^{2 l + 1} \right ) \Bigg \vert_L \, ,
    \end{split}
\end{align}
\end{subequations}
where we have inserted the series expansion for the \(Z\)-factors, i.e.
\begin{subequations}
\begin{align}
    Z^{\propagatorresidue} & = 1 - \sum_{m = 1}^\infty C^{\propagatorresidue}_m \, ,
    \intertext{in particular interpreted via their inverted form}
    \frac{1}{Z^{\propagatorresidue}} & = 1 + \sum_{k = 1}^\infty \left ( \sum_{m = 1}^\infty C^{\propagatorresidue}_m \right )^k \, ,
    \intertext{and}
    Z^{\vertexresidue} & = 1 + \sum_{n = 1}^\infty C^{\vertexresidue}_n
\end{align}
\end{subequations}
to make their loop-dependence explicit. Specifically, for comparison with the explicit calculations in \Cref{eqn:renormalized-feynman-rules-oneloop} and \Cref{eqn:renormalized-feynman-rules-twoloop}, we calculate the cases for \(L \in \{ 1, 2 \}\). For \(L = 1\), we obtain:
\begin{subequations}
\begin{align}
    \begin{split}
    \Phi^\texttt{R} \left ( \mathfrak{X}^{\propagatorresidue} \right ) \Bigg \vert_1 & = \Phi \left ( \mathfrak{X}^{\propagatorresidue}_1 \right ) + C^{\propagatorresidue}_1
    \end{split}
    \intertext{and}
    \begin{split}
    \Phi^\texttt{R} \left ( \mathfrak{X}^{\vertexresidue} \right ) \Bigg \vert_1 & = \Phi \left ( \mathfrak{X}^{\vertexresidue}_1 \right ) + C^{\vertexresidue}_1 \, ,
    \end{split}
\end{align}
\end{subequations}
which replicates \Cref{eqn:renormalized-feynman-rules-oneloop}. Additionally, for \(L = 2\), we obtain:
\begin{subequations}
\begin{align}
    \Phi^\texttt{R} \left ( \mathfrak{X}^{\propagatorresidue} \right ) \Bigg \vert_2 & = \Phi \left ( \mathfrak{X}^{\propagatorresidue}_2 \right ) +2 \, C^{\propagatorresidue}_1 \cdot_\AQ Y^{\propagatorresidue}_{1,1,0} + 2 \, C^{\vertexresidue}_1 \cdot_\AQ Y^{\propagatorresidue}_{1,0,1} + C^{\propagatorresidue}_2
    \intertext{and}
    \Phi^\texttt{R} \left ( \mathfrak{X}^{\vertexresidue} \right ) \Bigg \vert_2 & = \Phi \left ( \mathfrak{X}^{\vertexresidue}_2 \right ) + 3 \, C^{\propagatorresidue}_1 \cdot_\AQ Y^{\vertexresidue}_{1,1,0} + 3 \, C^{\vertexresidue}_1 \cdot_\AQ Y^{\vertexresidue}_{1,0,1} + C^{\vertexresidue}_2 \, ,
\end{align}
\end{subequations}
which replicates \Cref{eqn:renormalized-feynman-rules-twoloop}.

\subsection{Renormalization of the two loop amplitude}

We first start with the one loop amplitude. Let $G_{xy}=|x-y|^{2-d}$ be the propagator from $x$ to $y$ (The Green's function for the quadratic part). For the one loop correction of the propagator,
\begin{equation*}
    \raisebox{0.33cm}{ \scriptsize \(x\) \hspace{-0.8cm}}
    \bubblegraph
    \raisebox{0.33cm}{\hspace{-0.7cm}\scriptsize \(y \) }
\end{equation*}
we obtain the Feynman rules 
\[
\Phi(\Gamma)= G_{xy}^2.
\]
The propagator has scaling degree $d-2$ with respect to the $xy$ diagonal. The degree of divergence of the one loop graph computes as $2(d-2)-d$ which equals $2$ for $d=6$. Thus there is a quadratic divergence.  For the one loop graph there are no subdivergences, and the distribution is homogeneous with respect to the $xy$ diagonal. Therefore, by \Cref{lem:extension lem of distributions}, we obtain 
\begin{equation}
    -\Phi^-(\Gamma)= \mathscr{R}\Phi(\Gamma)= c^0(\epsilon) \delta_{xy} +c^i(\epsilon)\partial_i \delta_{xy}+ c^{ij}(\epsilon) \partial_{ij} \delta_{xy}.
\end{equation}
The renormalized amplitude is given by
\begin{equation}
    \Phi^+(\Gamma)\coloneq \Phi(\Gamma)-\mathscr{R}\Phi(\Gamma).
\end{equation}
By the extension lemma of distributions \Cref{lem:extension lem of distributions},
the limit $\epsilon\to 0$ exists in $\mathscr{D}'(M^2)$ and therefore we obtain a renormalized amplitude.

For the two loop correction of the propagator, there are two contributing diagrams.
First, we renormalize the diagram 
\[\raisebox{-0.33cm}{\hspace{1.25cm} \scriptsize \(x\) \hspace{-1.25cm}} \raisebox{0.67cm}{\hspace{1cm} \scriptsize \(w\) \hspace{-1cm}} \bubbleupperedgecorrectiongraph \raisebox{-0.33cm}{\hspace{-0.75cm} \scriptsize \(y\) \hspace{0.75cm}} \raisebox{0.67cm}{\hspace{-1.25cm} \scriptsize \(z\) \hspace{1.25cm}}
\]

We obtain the following Feynman rules
\begin{equation}
    \Phi(\Gamma)=G_{xy} G_{yz} G_{xw} G_{zw}^2.
\end{equation}
The superficial degree of divergence of $\Gamma$ is $\operatorname{SDD}(\Gamma)=2$. Moreover, it contains one divergent subgraph $\gamma$, with Feynman rules $\Phi(\gamma)= G_{zw}^2$ of the same order of divergence. We obtain
\begin{equation}
    \Phi^-(\Gamma)= -\mathscr{R}[\Phi(\Gamma)-\insertionproduct{\gamma}{\Gamma}(\mathscr{R} \Phi(\gamma)) \Phi(\Gamma \backslash \gamma)]
\end{equation}
Here, $\mathscr{R}\Phi(\gamma)= \sum_{|\alpha|\leq 2} c_{xy}^\alpha(\epsilon_{xy}) \partial^\alpha \delta_{xy}$ and $\mathscr{R}\Phi(\Gamma)= \sum_{|\alpha|\leq 2} c_{I}^\alpha(\epsilon_{I}) \partial^\alpha \delta_{I}$, where $I=xyzw$. We want to show that $\lim_{\epsilon_I \to 0} \lim_{\epsilon_{xy}\to 0}\Phi^+(\Gamma)_\epsilon$ exists. First we compute (compare to \Cref{eqn:coproduct-of-two-loop-graphs})
\begin{align}
    \begin{split}
        \Phi^+(\Gamma)&=(\Phi^-\star \Phi)(\Gamma) \\
        &= m\circ(\Phi^-\otimes \Phi)(\Gamma\otimes 1+ 1\otimes \Gamma + \gamma \otimes \Gamma \backslash \gamma) \\
        &= \Phi^-(\Gamma) + \Phi(\Gamma) +\Phi^-(\gamma) \insertionproduct{\gamma}{\Gamma}\Phi(\Gamma \backslash \gamma) \\
        &= (1-\mathscr{R}) [ \Phi(\Gamma) -\insertionproduct{\gamma}{\Gamma} (\mathscr{R}\Phi(\gamma)) \Phi(\Gamma \backslash \gamma)]
    \end{split}
\end{align}
We first extend the distribution to the $zw$ diagonal. In order to do this, note that we can split up the amplitude as
\begin{equation}
    \Phi(\Gamma)=\insertionproduct{\gamma}{\Gamma}
    \Phi(\gamma)\Phi(\Gamma\backslash \gamma).
\end{equation}
We obtain
\begin{equation}
    \Phi^+(\Gamma)=(1-\mathscr{R})\insertionproduct{\gamma}{\Gamma}[\Phi(\gamma)-\mathscr{R}\Phi(\gamma) ]\Phi(\Gamma\backslash\gamma).
\end{equation}
Note that by the standard lemma, the limit 
\begin{equation}
   \Phi^+(\gamma)= \lim_{\epsilon_{xy}\to 0} \Phi(\gamma)-\mathscr{R}\Phi(\gamma)
\end{equation}
exists on $M^4\backslash D_I$. Therefore we can also take the limit 
\begin{equation}
    \lim_{\epsilon_I\to 0} (1-\mathscr{R})\insertionproduct{\gamma}{\Gamma}\Phi^+(\gamma)\Phi(\Gamma\backslash \gamma).
\end{equation}
This finishes the proof that $\Phi^+(\Gamma)$ is well defined on $M^4$.

The next graph is the vertex correction graph
\[
    \raisebox{0.33cm}{ \scriptsize \( x\)\hspace{-0.7cm}}\bubblevertexcorrectiongraph \raisebox{0.9cm}{\hspace{-1.4cm}\scriptsize \(y \) } 
    \raisebox{-0.9cm}{\hspace{-0.3cm}\scriptsize \(w \)}
    \raisebox{0.33cm}{\hspace{0.5cm}\scriptsize \(z \)}
\]
The Feynman rules are given by $\Phi(\Gamma)=G_{xy}G_{xw}G_{yz}G_{wz}G_{wy}$. Its superficial degree of divergence is given by $\SDD(\Gamma)=2$, again. (This no coincidence. Since the theory is renormalizable in $6$ dimensions, the degree of divergence only depends on the external leg structure of the graph). There are two divergent subgraphs $\gamma_1$ and $\gamma_2$ 
\begin{equation}
    \begin{split}
          \raisebox{0.33cm}{ \scriptsize \( x\)\hspace{-1.7cm}} \raisebox{0.8cm}{ \hspace{3cm}\scriptsize \( y\)\hspace{-3cm}} \raisebox{-0.8cm}{\hspace{2.6cm} \scriptsize \( w\)\hspace{-3cm}}\gamma_1 \coloneq \trianglegraph \hspace{2cm}\raisebox{0.8cm}{ \scriptsize \( y\)\hspace{-0.4cm}}\raisebox{-0.8cm}{ \scriptsize \( w\)\hspace{-1.6cm}} \gamma_2\coloneq \trianglegraphmirrored \raisebox{0.3cm}{ \hspace{-0.8cm}\scriptsize \( z\)}
    \end{split} 
\end{equation}

Their Feynman rules are
\begin{align}
\begin{split}
    \Phi(\gamma_1)=G_{xy} G_{xw}G_{yw}, \quad \Phi(\gamma_2)=G_{yz}G_{zw}G_{yw}.
\end{split}    
\end{align}
They both have $\SDD(\gamma_1)=\SDD(\gamma_2)=0$ and therefore admit a logarithmic divergence. We calculate
\begin{align}
\begin{split}
    \Phi^-(\Gamma)&=-\mathscr{R}[\Phi(\Gamma)+\insertionproduct{\gamma_1}{\Gamma}\Phi^-(\gamma_1)\Phi(\Gamma\backslash\gamma_1)+\insertionproduct{\gamma_2}{\Gamma}\Phi^-(\gamma_2)\Phi(\Gamma\backslash\gamma_2)] \\
    &=-\mathscr{R}[\Phi(\Gamma)-\insertionproduct{\gamma_1}{\Gamma}(\mathscr{R}\Phi(\gamma_1))\Phi(\Gamma\backslash\gamma_1)-\insertionproduct{\gamma_2}{\Gamma}(\mathscr{R}\Phi(\gamma_2))\Phi(\Gamma\backslash\gamma_2)]
\end{split}
\end{align}
This gives
\begin{align}
    \begin{split}
        \Phi^+(\Gamma)=(1-\mathscr{R})[\Phi(\Gamma)-\insertionproduct{\gamma_1}{\Gamma}(\mathscr{R}\Phi(\gamma_1))\Phi(\Gamma\backslash\gamma_1)-\insertionproduct{\gamma2}{\Gamma}(\mathscr{R}\Phi(\gamma_2))\Phi(\Gamma\backslash\gamma_2)].
    \end{split}
\end{align}
To proceed further, we choose a partition of unity of $M^4\backslash D_I$, where $I=xyzw$, described in \Cref{lem:stora}. The only open subsets of this covering where $\Phi(\Gamma)$ has a divergence are the subsets $C_{x}$ and $C_{z}$. On $C_{x}$, only $\gamma_2$ has divergencies. Restricting $\Phi^+(\Gamma)$ to $C_{x}$ gives the contribution
\begin{equation}
    \Phi^+(\Gamma)|_{C_{x}}= (1-\mathscr{R})[\Phi(\Gamma)-\insertionproduct{\gamma_2}{\Gamma}(\mathscr{R}\Phi(\gamma_2))\Phi(\Gamma\backslash\gamma_2)]|_{C_{x}}.
\end{equation}
Similarly to the previous example, this can be extended to $M^4\backslash D_I$. In the same manner, $\Phi^+(\Gamma)|_{C_{z}}$ can be extended to $M^4\backslash D_I$. Adding the contributions together, using the partition of unity yields a distribution that is defined on all of $M^4\backslash D_I$ and additionally can be extended to $M^4$ by $(1-\mathscr{R})$.

\section{Main results} \label{sec:main-results}

In this section, we obtain some equivalence statements concerning locality and multiplicativity. Finally, we present the construction of the multiplicatively renormalized Lagrange density.

\subsection{Multiplicative renormalization in causal perturbation theory}

Let \(\Gamma = \Gamma_1 \sqcup \Gamma_2\) be the disjoint union of two graphs \(\Gamma_1, \Gamma_2 \in \HQ\). Given the renormalized Feynman rules \(\Phi^+\) and the counterterm map \(\Phi^-\), we clarify their multiplicative behavior depending on the renormalization scheme \(\mathscr{R}\) as follows:

\enter

\begin{prop} \label{thm:rota-baxter is multiplicativity}
The renormalized Feynman rules and the counterterm maps are multiplicative, i.e.\ \(\Phi^\pm \left ( \Gamma_1 \cdot_\HQ \Gamma_2 \right ) = \Phi^\pm \left ( \Gamma_1 \right ) \cdot_\AQ \Phi^\pm \left ( \Gamma_2 \right )\) respectively, if and only if the renormalization scheme \(\mathscr{R}\) satisfies the Rota--Baxter relation 
\begin{equation}\label{eq:rota baxter}
    \mathscr{R} \circ m_\AQ + m_\AQ \circ \left ( \mathscr{R} \otimes \mathscr{R} \right ) = \mathscr{R} \circ m_\AQ \circ \left ( \mathscr{R} \otimes \Id_\AQ + \Id_\AQ \otimes \, \mathscr{R} \right ) \, .
\end{equation}
\end{prop}

\begin{proof}
This follows directly from \cite[Proposition 3.10, cf.\ Remark 3.11]{Ebrahimi-Fard:2004twg}. The proof goes through verbatim, although it should be noted that we use the insertion product for the star product of graphs and we use the Rota--Baxter property with respect to the tensor product of distributions.
\end{proof}

\henter

\begin{lem} \label{lem:hadamard is rota baxter}
Let $\mathscr{R}$ be the Hadamard singular part defined in \Cref{def:hadamard-singular-part-scheme}.
Then $\mathscr{R}$ is Rota--Baxter. 
\end{lem}
\begin{proof}
Let $(\alpha, \beta)$ and $(\gamma,\delta)$ be power indices.
Then the standard functions satisfy
$p_{\alpha,\beta} p_{\gamma,\delta} = p_{\alpha+\gamma, \beta+\delta}$.
Note that $\operatorname{deg}(\alpha+\gamma, \beta+\delta) < 0$ if both $\operatorname{deg}(\alpha, \beta) < 0$ and $\operatorname{deg}(\gamma, \delta) < 0$.

Similarly, $\operatorname{deg}(\alpha +\gamma, \beta+ \delta) \geq 0$ if both
$\operatorname{deg}(\alpha, \beta) \geq 0$ and $\operatorname{deg}(\gamma, \delta) \geq 0$.
In words: the product of singular terms is singular, and the product of regular terms is regular.

The projections $\mathscr{R}, \Pi$ therefore satisfy the following relations with the multiplication $m$:
\begin{subequations}
\begin{align}
    \mathscr{R} \circ m \circ (\mathscr{R} \otimes \mathscr{R}) & = m \circ (\mathscr{R} \otimes \mathscr{R}) \label{eq:proj-rel-1}\\
    \Pi \circ m \circ (\Pi \otimes \Pi) & = m \circ (\Pi \otimes \Pi) \label{eq:proj-rel-2}.
\end{align}
\end{subequations}
From \Cref{eq:proj-rel-1} it follows that the Rota--Baxter relation is equivalent to
\begin{equation}
    \mathscr{R} \circ m \circ (\Id \otimes \mathscr{R} + \mathscr{R} \otimes \Id) = \mathscr{R} \circ m \circ (\mathscr{R}\otimes \mathscr{R}) + \mathscr{R} \circ m \circ (\Id\otimes \Id).
\end{equation}
We expand $\Id \otimes \Id$ by using $\Id = \mathscr{R} + \Pi$, giving
\begin{equation}
    \mathscr{R} \circ m \circ (\Id \otimes \Id) = \mathscr{R} \circ m \circ (\mathscr{R}\otimes \mathscr{R} + \Pi \otimes \mathscr{R} + \mathscr{R} \otimes \Pi + \Pi \otimes \Pi).
\end{equation}
From \Cref{eq:proj-rel-2} it follows that $\mathscr{R}m(\Pi\otimes\Pi) = 0$.
Using $\Pi = \Id - \mathscr{R}$ we are left with
\begin{equation}
\begin{split}
    \mathscr{R} \circ m \circ (\Id \otimes \Id)
    & = \mathscr{R} \circ m \circ (\mathscr{R}\otimes \mathscr{R} + (\Id - \mathscr{R}) \otimes \mathscr{R} + \mathscr{R} \otimes (\Id - \mathscr{R}))\\
    & = \mathscr{R} \circ m \circ (\Id \otimes \mathscr{R} + \mathscr{R} \otimes \Id) + \mathscr{R} \circ m \circ (\mathscr{R}\otimes \mathscr{R} - \mathscr{R} \otimes \mathscr{R} - \mathscr{R} \otimes \mathscr{R})\\
    & = \mathscr{R} \circ m \circ ( \Id \otimes \mathscr{R} + \mathscr{R} \otimes \Id) - \mathscr{R} \circ m \circ (\mathscr{R}\otimes \mathscr{R})
\end{split}
\end{equation}
which verifies the Rota--Baxter relation.
\end{proof}

\enter

\begin{defn}
An extension operator $E:\mathcal{O}(M,I)\to \mathscr{D}'(M^k)$ (see \Cref{def: Feynman distributions}) is $C^\infty$-linear if for any $u\in\mathcal{O}(M,I)$ and $f\in C^\infty(M^k)$, then $E(fu)=fE(u)$. A renormalization scheme $\mathscr{R}$ is called $C^\infty$-linear if it is $C^\infty$-linear in the same sense.
\end{defn}

\enter

\begin{rem}
By \Cref{prop:hadamard-polyhmg-dists} it follows that the Hadamard singular part is $C^\infty$-linear.
\end{rem}

\enter

\begin{lem}
If an extension operator is $C^\infty$-linear then there exists an extension operator on each open subset $U\subset M^k$ which is $C^\infty(U)$-linear.
\end{lem}

\begin{proof}
Given any distribution $u\in \mathscr{D}'(U\backslash D_I)$, and any testfunction $\varphi \in \mathscr{D}(U)$, we would like to define $\langle \varphi, u \rangle$. Choose any function $g\in C_c^\infty(U)$ such that $g=1$ on the support of $\varphi$. The distribution $gu$ can be extended by $0$ to $M^k\backslash D_I$. Therefore, we can define the evaluation by $\langle \varphi, E(gu)\rangle$. We show that this construction is independent of the choice of $g$. Indeed, given a different $g'$ satisfying the same conditions, one has
\begin{align}
    \begin{split}
        \langle \varphi, E(gu)\rangle &= \langle g' \varphi, E(gu) \rangle \\
        &= \langle \varphi, g' E(gu) \rangle \\
        &= \langle \varphi, g E(g'u) \rangle \\ 
        &= \langle g \varphi, E(g'u) \rangle \\
        &= \langle \varphi, E(g'u) \rangle.
    \end{split}
\end{align}
It is easy to see that this linear functional is continuous, and therefore a distribution.
\end{proof}

\enter

\begin{prop} \label{thm:Locality+Rota-Baxter implies Epstein-Glaser}
The renormalized Feynman rules $\Phi^+$ satisfy the Epstein--Glaser extension property \Cref{eq: factorization property} if the renormalization scheme $\mathscr{R}$ is $C^\infty$-linear and satisfies the Rota--Baxter property \Cref{eq:rota baxter}.
\end{prop}

\begin{proof}
Assuming $C^\infty$-linearity and the Rota--Baxter property, and a graph that decomposes as $\Gamma=\Gamma_1 \cdot_\HQ \Gamma_2 \mathcal{P}_{1,2}$, where $\mathcal{P}_{1,2}$ are all edges in $\Gamma$ connecting $\Gamma_1$ and $\Gamma_2$. Let $C_{1,2}$ be the open set such that $\mathcal{P}_{1,2}$ is non-singular. We obtain that 
\begin{equation}
    \Phi^+(\Gamma)|_{C_{1,2}}=\insertionproduct{\mathcal{P}_{1,2}}{\Gamma}\Phi^+(\Gamma_1 \Gamma_2) \Phi(\mathcal{P}_{1,2})|_{C_{1,2}}=\Phi^+(\Gamma_1)\Phi^+(\Gamma_2)\Phi(\mathcal{P}_{1,2})|_{C_{1,2}} 
\end{equation}
The first equality follows from $C^\infty$-linearity  and the second equality follows from the Rota--Baxter property by \Cref{thm:rota-baxter is multiplicativity}. 
\end{proof}

\enter

\begin{thm} \label{thm:Multiplicative Renormalization scheme yields local extension+ Epstein-Glaser}
Let $\mathscr{R}$ be a $C^\infty$-linear renormalization scheme that is Rota--Baxter. Then for each graph $\Gamma\in \HQ$, the limit $\lim_{\epsilon \to 0} \Phi^+(\Gamma)$ exists and defines a $C^\infty$-linear extension of $\, \Phi(\Gamma)\in \mathcal{O}(M,I)$ to $\mathscr{D}'(M^k)$. Furthermore, it satisfies the Epstein--Glaser factorization property of \Cref{eq: factorization property}.
\end{thm}

\begin{proof}
We prove the statement by induction on the number of divergent subgraphs. First suppose that the graph $\Gamma$ is primitive, so that $\Delta(\Gamma)=1\otimes \Gamma+\Gamma\otimes 1$. Then $\Phi^+(\Gamma)=(1-\mathscr{R})\Phi(\Gamma)$. The only singularity is on the small diagonal, hence the limit exists.

For the induction step, suppose that the limit exists for all graphs of $n$ divergent subgraphs. Assume that $\Gamma$ has $n+1$ divergent subgraphs. Then
\begin{equation}
    \Phi^+(\Gamma)=(1-\mathscr{R})\left[\Phi(\Gamma)+\sum_{\emptyset \subsetneq \gamma \subsetneq \Gamma}\insertionproduct{\gamma}{\Gamma}\Phi^-(\gamma)\Phi(\Gamma\backslash \gamma)\right].
\end{equation}
Let  $\{C_I\}_{I\subset \{1, \dots, d\}}$ be the cover of $M^d \backslash D_d$ introduced in \Cref{lem:stora}.
Let $\chi_I$ be a partition of unity subordinate to that cover. We can write 
\begin{equation}\label{eq: renormalization with partition}
    \Phi^+(\Gamma)=(1-\mathscr{R})\sum_I \chi_I \left[\Phi(\Gamma)+\sum_{\emptyset \subsetneq \gamma \subsetneq \Gamma}\insertionproduct{\gamma}{\Gamma}\Phi^-(\gamma)\Phi(\Gamma\backslash \gamma)\right].
\end{equation}
Restricting to a single subset $C_I$, a graph has only singularities on subdiagonals of the $I$ and the $I^c$ diagonal. Let $\Gamma_I$ be the graph obtained from $\Gamma$ by deleting all edges connecting a vertex contained in $I$ with a vertex contained in the complement $I^c$. Let $\mathcal{P}_{I,J}$ be the product of all propagators connecting $I$ and $J$. For any connected component $i\in \pi_0(\Gamma_I)$ let $\Gamma_i$ be the subgraph of $\Gamma$ corresponding to that connected component. Because $\Phi^-$ satisfies $C^\infty$ linearity and the Rota--Baxter property, we can write
\begin{equation}
    \Phi(\Gamma)|_{C_I}=\mathcal{P}_{I,I^c} \prod_{i\in \pi_0(\Gamma_I)} \Phi(\Gamma_i)
\end{equation}
and
\begin{equation}
    \Phi^-(\gamma)|_{C_I}=\mathcal{P}_{I,I^c} \prod_{i\in \pi_0(\Gamma_I)} \Phi^-(\gamma \cap \Gamma_i) 
\end{equation}
since $\mathcal{P}_{I,I^c}$ is non-singular on the set $C_I$. Thus, \Cref{eq: renormalization with partition} splits as
\begin{equation}
    (1-\mathscr{R})\sum_I \chi_I \left[\mathcal{P}_{I,I^c} \prod_{i\in \pi_0(\Gamma_I)} \Phi(\Gamma_i)+\sum_{\emptyset \subsetneq \gamma \subsetneq \Gamma}\mathcal{P}_{I,I^c} \prod_{i\in \pi_0(\Gamma_I)} \Phi^-(\gamma \cap \Gamma_i)\Phi(\Gamma_i\backslash \gamma)\right].
\end{equation}
Using the fact that $\Phi^-(\gamma \cap \Gamma_i)$ vanishes if $\gamma \cap \Gamma_i$ is not singular, we can write
\begin{equation}
    \Phi^+(\Gamma)=(1-\mathscr{R})\sum_I \chi_I \mathcal{P}_{I,I^c} \prod_{i\in \pi_0(\Gamma_I)} \left[ \Phi(\Gamma_i)+\sum_{\emptyset \subsetneq\gamma'\subsetneq \Gamma_i} \Phi^-(\gamma')\Phi(\Gamma_i\backslash \gamma') \right].
\end{equation}
The term in brackets is by definition $\Phi^+(\Gamma_i)$ and by induction hypothesis the limit exists and therefore we obtain a well defined distribution on $C_I$. The sum then defines in the limit a distribution on $M^d\backslash D_d$. The $(1-\mathscr{R})$ term subtracts the overall divergence and the limit
\begin{equation}
    \lim_{\epsilon\to 0} (1-\mathscr{R}_\epsilon)\sum_I \chi_I \mathcal{P}_{I,I^c} \prod_{i\in \pi_0(\Gamma_I)} \Phi^+(\Gamma_i)
\end{equation}
exists on $M^d$. The construction is independent of the choice of partition of unity, since $C_I\cap C_J$ contains only diagonals that are contained in both $C_I$ and $C_J$. On these diagonals, the distributions agree, so a different partition of unity would give the same answer. On diagonals that are only contained in one $C_I$, $\chi_I$ is equal to $1$, so two partitions of unity also agree.
\end{proof}

\enter

\begin{cor}\label{cor:Hadamard finite part is multiplicative Renormalization scheme}
The Hadamard singular part (\Cref{def:hadamard-singular-part-scheme}) is a multiplicative renormalization scheme that defines an extension operator $\mathcal{O}(M,I) \to \mathscr{D}'(M^I)$ for all Feynman graphs given by $t_\Gamma \mapsto\lim_{\epsilon \to 0} \Phi^+(\Gamma)$.
\end{cor}

\begin{proof}
We combine \Cref{prop:hadamard-reg-basic-block} and \Cref{lem:hadamard is rota baxter} to conclude with \Cref{thm:Multiplicative Renormalization scheme yields local extension+ Epstein-Glaser}.
\end{proof}

\subsection{The renormalized Lagrange density}
In this subsection, we explicitly describe the multiplicatively renormalized Lagrange density together with its properties.

\enter

\begin{defn}
Let \(\CTmap\) be the counterterm map from \Cref{defn:renormalized-feynman-rules} and \(\combgreen\) be the combinatorial Green's function from \Cref{defn:combinatorial-greens-function}. Then, the \(Z\)-factors are defined via:
\begin{equation}
    Z^r := \CTmap \left ( \combgreen \right )
\end{equation}
\end{defn}

\henter

\begin{thm}\label{thm:Renormalized Lagrangian gives renormalized Feynman rules}
Let \(\LQ \in \AQ_0\) be a Lagrange density with multiplicatively renormalized Lagrange density
\begin{equation}
    \LQ^\emph{\texttt{R}} \coloneq \sum_{r \in \RQ} Z^r \mathcal{M}^r \in \AQ^\emph{\texttt{R}}
\end{equation}
as in \Cref{defn:mult-ren-lagr-dens}. Denote the sum over all unrenormalized Feynman distributions contributing to the amplitude \(r\) and loop number \(L\) by
\begin{subequations}
\begin{align}
    u^r_L \left ( x_1, \dots, x_n \right ) & \coloneq \Phi \left ( \combgreen_L \right ) \, ,
    \intertext{with \(\combgreen_L\) the restricted combinatorial Green's functions from \Cref{defn:combinatorial-greens-function}, and}
    v^r_L \left ( \epsilon; \, x_1, \dots, x_n \right ) & \coloneq \Phi^\emph{\texttt{R}} \left ( \combgreen \right ) \Bigg \vert_L
\end{align}
\end{subequations}
the corresponding sum using the multiplicatively renormalized Feynman rules \(\Phi^\emph{\texttt{R}}\) of \Cref{thm:multiplicative-renormalization}. Then, the limit \(\lim_{\epsilon \to 0} v^r_L \left ( \epsilon; x_1, \dots, x_n \right )\) exists and defines an extension of \(u \left ( x_1, \dots, x_n \right )\) in the sense of Epstein--Glaser.
\end{thm}

\begin{proof}
We have shown in \Cref{thm:multiplicative-renormalization} the combinatorial equivalence of the renormalized Feynman rules \(\Phi^+\), defined via an algebraic Birkhoff decomposition in the sense of \Cref{defn:algebraic-birkhoff-decomposition}, and the multiplicatively renormalized Feynman rules \(\Phi^\texttt{R}\), defined via the multiplicatively renormalized Lagrange density in the sense of \Cref{defn:mult-ren-lagr-dens}. Additionally, starting with a local and multiplicative renormalization scheme, we have shown in \Cref{thm:Multiplicative Renormalization scheme yields local extension+ Epstein-Glaser} that the limit of the regularized and renormalized Feynman rules exists and that it defines a local extension in the sense of Epstein--Glaser. Thus, combining both results yields the claimed statement.
\end{proof}

\section{Conclusion} \label{sec:conclusion}

We have started in \Cref{sec:introduction} with an informal discussion of the present analysis. Then, in \Cref{sec:epstein-glaser}, we provided an overview on the Epstein--Glaser approach to renormalization. Next, in \Cref{sec:connes-kreimer}, we gave an introduction to the Connes--Kreimer framework of renormalization. In the following \Cref{sec:examples}, we exemplified the developed theory with the example of massless \(\phi^3_6\)-theory. Finally, in \Cref{sec:main-results}, we stated and proved our main results: Specifically, in \Cref{thm:Locality+Rota-Baxter implies Epstein-Glaser}, we have shown that the Rota--Baxter property is essentially equivalent to the Epstein--Glaser factorization property. Furthermore, in \Cref{thm:Multiplicative Renormalization scheme yields local extension+ Epstein-Glaser}, we have shown that the combinatorics of Connes--Kreimer actually define an extension for all distributions coming from Feynman graphs. Moreover, we have also proved that this extension operator satisfies the Epstein--Glaser factorization property. Additionally, in \Cref{thm:Renormalized Lagrangian gives renormalized Feynman rules}, we conclude by showing that the Feynman rules of the renormalized Lagrange density are equivalent to the renormalized Feynman rules. Finally, in Appendix \ref{sec:apx-hadamard-finite-part}, we provided the necessary analytic background for the extension of distributions, in particular the \emph{Hadamard finite part}.

\section*{Acknowledgements}

DP thanks Dorothea Bahns and the \emph{University of Göttingen} for hospitality during his visit in June 2023, where the present project was initiated. AH thanks the \emph{Max Planck Institute for Mathematics} for hospitality during his visits in August 2023 and March 2025, where the main results of the present work were developed.

\begin{appendix}
\section{Appendix: The Hadamard finite part} \label{sec:apx-hadamard-finite-part}

In the Epstein--Glaser approach, the evaluation of singular Feynman amplitudes is expressed in terms of \emph{extension maps}.
In this formulation, the notion of \emph{counterterms} does not exist in the same form as in other schemes.
Since we do not explicitly deal with any singular limits, there is no need to subtract any counterterms.

In practice, the construction of extension maps does involve limits and counterterms.
One can see the notion of extension map as a convenient black box which hides the details of such procedures and exposes only the salient features to the user.

Because the Connes--Kreimer combinatorics deal primarily with the correct subtraction of counterterms, any comparison of the Epstein--Glaser and Connes--Kreimer approaches has to open up the black box to a certain extent.

In this appendix we will discuss the relation between extension maps and regularization procedures in the model case.
That is, we consider the problem of extending distributions from $\dot{U} = U \backslash 0$ to $U$, where $U\subset\mathbb{R}^d$ is a neighborhood of $0$.

\subsection{Partially finite functions in a regularization parameter}
We will consider regularizations which introduce a parameter $\epsilon \in [0,1]$ and state two lemmas concerning standard integrals.

\enter

\begin{lem}
    \label{lem:standard-hadamard-integral}
    Define the integral
    \begin{align*}
        I_{a,m}(\epsilon) = \int_{\epsilon}^1 t^{a-1} \log^m t~\mathrm{d} t.
        \end{align*}
    We have
    \begin{align*}
    I_{a,m}(\epsilon) = \operatorname{pf}(I_{a,m}) + \epsilon^a \sum_{j=0}^m \theta_{a,m,j} \log^j \epsilon
    \end{align*}
    with finite part
    \begin{align*}
        \operatorname{pf}\left(I_{b, m}\right)= \begin{cases}(-1)^m \frac{m!}{a^{m+1}}, & a \neq 0 \\ 0, & b= 0\end{cases}
    \end{align*}
    and coefficients
    \begin{align*}
    \theta_{a,m,j} = \begin{cases}(-1)^{m+1+j} \frac{m!}{j!a^{m+1-j}}, & a \neq 0,0 \leq j \leq m, \\ 0, & a \neq 0, j=m+1, \\ -\frac{1}{m+1}, & a=0, j=m+1, \\ 0, & a=0,0 \leq j \leq m .\end{cases}
    \end{align*}
\end{lem}

\enter

\begin{lem}
    \label{prop:log-hg-hadamard-integrals}
    Let $\phi$ be a Schwartz function on $(0,\infty)$.
    Define
    \begin{align*}
    H_{a,k,\epsilon}(\phi) = \int_{\epsilon}^\infty t^{a-1} \log^k t \phi(t) \mathrm{d} t.
\end{align*}
    Then $H_{a,m,\epsilon}(\phi) \in \mathsf{PF}_{a, m+1}$.
    Explicitly, we have an expansion
\begin{align*}
    H_{a, m, \epsilon}(\phi)=C_{a, m}[\phi]+\sum_{n=0}^N \sum_{j=0}^{m+1} c_{n, j}^{(a, m)}[\phi] \epsilon^{a+n}(\log \epsilon)^j+R_N(\epsilon ; \phi),
\end{align*}
with coefficients
\begin{align*}
    c_{n, j}^{(a, m)}[\phi]= \frac{\phi^{(n)}(0)}{n!} \theta_{a+n, m, j}
\end{align*}
where $\theta_{a,m,j}$ are as in \Cref{lem:standard-hadamard-integral}.
The finite part can be represented as
\begin{align*}
    C_{a, m}[\phi]=
         \int_1^{\infty} t^{a-1}(\log t)^m \varphi(t) \D t 
         +\int_0^1 t^{a-1}(\log t)^m\left(\varphi(t)-\sum_{n=0}^N \frac{\varphi^{(n)}(0)}{n!} t^n\right) \D t\\
         +\sum_{n=0}^N \frac{\varphi^{(n)}(0)}{n!} \operatorname{pf}\left(I_{a+n, m}\right).
\end{align*}
We also have
\begin{align*}
C_{a,m}[\phi] &= \operatorname{fp}_{z=0} \int_0^\infty t^{a+ z - 1} \log^k \phi(t) \D t\\
&= \int_0^{\infty} t^{a+N} P_{a, m, N}(\log t) \varphi^{(N)}(t) d t+\sum_{n=0}^{N-1} \frac{\varphi^{(n)}(0)}{n!} \operatorname{pf}\left(I_{a+n, m}\right)
\end{align*}
where $P_{a,m,N}$ is a polynomial.
\end{lem}

\subsection{Regularization schemes}

Let $U\subset \mathbb{R}^d$ be a neighborhood of $0$ and define $\dot{U} := U\setminus \{0\}$.
We extend our notion of partially finite functions to functions with values in $\mathscr{D}'(U)$.

\enter

\begin{defn}
A function $f: (0,1] \to \mathscr{D}'(U)$ is in $\mathsf{PF}_{a,m}((0,1], \mathscr{D}'(U))$ if and only if for every $\phi \in \mathscr{D}(U)$, the function
$\epsilon \mapsto \langle f(\epsilon), \phi\rangle$ is in $\mathsf{PF}_{a,m}$.
\end{defn}

\enter

As in the scalar case, we denote
    \begin{align}
    \label{eq:-pf-space-all}
        \mathsf{PF}_{\mathbb{Z}, \mathbb{N}}((0,1], \mathscr{D}'(U)) := \bigoplus_{a \in \mathbb{Z}, M \in \mathbb{N}} \mathsf{PF}_{a,M}((0,1], \mathscr{D}'(U)).
    \end{align}
\begin{defn}
\label{def:regulator}
Let \(\mathscr{X}(\dot{U}) \subset \mathscr{D}'(\dot{U})\) be a space of distributions.
A \emph{regulator} is a map
\begin{align*}
\operatorname{Reg}: \mathscr{X}(\dot{U}) \to
\mathsf{PF}_{a,M}((0,1], \mathscr{D}'(U))
\end{align*}
such that
\begin{enumerate}
    \item $\operatorname{Reg}(u)(\epsilon) \in \mathscr{C}^\infty(U)$ for all $\epsilon \in (0,1]$,
    \item \(\langle \operatorname{Reg}(u)(\epsilon), \phi\rangle \xrightarrow{\epsilon \to 0} \langle u, \phi\rangle\) for all \(\phi \in \mathscr{D}(\dot{U})\).
\end{enumerate}
We will write $\operatorname{Reg}_\epsilon(u) := \operatorname{Reg}(u)(\epsilon)$.
\end{defn}

\enter

\begin{prop}
\label{prop:regulator-coefficient-distributions}
Let \(T\) be a regulator on \(\mathscr{X}(\dot{U})\),
i.e. for all \(u\in \mathscr{X}\) and all
\(\phi\ \in \mathscr{D}(U)\) we have an expansion
\begin{align*}
\langle T(u)(\epsilon), \phi\rangle
= C[u,\phi] + \sum_{k=0}^N \sum_{j=0}^M a_{k,j}[u,\phi]
\epsilon^{\alpha + k} \log^j\epsilon + R_{N,M}(\epsilon)
\end{align*}
with $R_{N,M}(\epsilon) = O(\epsilon^{N+1}\log^M\epsilon)$.

Then there are distributions $C[u], t_{k,j}[u] \in \mathscr{D}'(U)$ such that $C[u,\phi] = \langle C[u], \phi\rangle$, $a_{k,j}[u,\phi] = \langle t_{k,j}[u], \phi\rangle$.
\end{prop}
\begin{proof}
Since $\mathscr{D}'(U)$ is a quasi-complete space,
we can define the Gelfand--Pettis (weak) integral
\begin{align*}
    f(z) = \int_{0}^1 T_\epsilon(u) \epsilon^{z-1} d\epsilon.
\end{align*}
This is the Mellin transform of $T_\epsilon(u)$.
Since $T_\epsilon(u)$ is partially finite, the integral is well defined for large $\Re z$ and it analytically continues to a meromorphic function on a strip.

The residues of $f$
are precisely the coefficients \(a_{k,j}[u, \phi]\).
Since weakly holomorphic (meromorphic) functions with values in quasi-complete spaces
are strongly holomorphic (meromorphic), the coefficients
of the Laurent series are distributions.
\end{proof}

\enter

Since \(\langle T_{\epsilon} u, \phi\rangle\) converges as \(\epsilon \to 0\)
for \(\phi \in \mathscr{D}'(\dot{U})\),
the coefficient distributions \(t_{k,j}[u]\) of singular terms are supported at the origin,
i.e. they are of the form \(\sum c_{\alpha}[u] \partial^{\alpha}\delta\).

There are several notions of locality available to us.

\enter

\begin{defn}
    Let \(T\) be a regulator on \(\mathscr{X}(\dot{U})\).
    We will call \(T\)
    \begin{enumerate}
    \item \emph{local} if the coefficients \(a_{k,j}[u,\phi]\) of \(\langle T_{\epsilon}(u),\phi\rangle\) satisfy $a_{k,j}[u,\phi] = 0$ whenever $\supp (u) \cap \supp\phi = \emptyset$.
    \item \emph{$C^\infty$-linear} if $a_{k,j}[f u, \phi] = a_{k,j}[u, f\phi]$ for all $f \in \mathscr{C}^\infty(U)$.
    \end{enumerate} 
\end{defn}

\enter

\begin{rem}
We can also rephrase these conditions in terms of the distributions $t_{k,j}$:
\begin{enumerate}
    \item $T$ is local if $\operatorname{supp} t_{k,j}[u] \subset \operatorname{supp} u$,
    \item $T$ is $C^\infty$-linear if $t_{k,j}[fu] = f t_{k,j}[u]$ for all $f \in \mathscr{C}^\infty(U)$.
\end{enumerate}
\end{rem}

\enter

The basic building block for the regulators which we use in this article is the following.

\enter

\begin{prop}
\label{prop:hadamard-reg-basic-block}
Let \(\chi, \varphi \in \mathscr{C}_c^{\infty}(\mathbb{R}^d)\)
and suppose that \(\chi \equiv 1\) in a neighborhood of \(0\).
We write $\rho(x) = |x|$.
Let $\mathscr{X}(\dot{U})$ be spanned by $\rho^{a} \log^k \rho \in \mathscr{D}'(\dot{U})$.
Then
\begin{align*}
    \langle T_\epsilon(u), \varphi\rangle
    := \langle (1- \chi_\epsilon) u, \varphi\rangle
\end{align*}
defines a $C^\infty$-linear regulator on $\mathscr{X}(\dot{U})$
with values in $\mathsf{PF}_{a+d, k+1}$.
\end{prop}
\begin{proof}[Sketch of proof.]
    First let us show that $T_\epsilon u$ is partially finite.
    We use a spherical layer decomposition.
    Note that
    \begin{align*}
    1 - \chi(x/\epsilon) = \int_\epsilon^\infty \psi(x/\lambda) \frac{\mathrm{d} \lambda}{\lambda}
    \end{align*}
    where $\psi(x) = -x \cdot \nabla \chi(x)$.
    Therefore
    \begin{align*}
    &\langle (1-\chi_\epsilon) u, \varphi\rangle = \int_\epsilon^\infty \langle \psi(x/\lambda) u, \varphi\rangle \frac{\mathrm{d} \lambda}{\lambda}\\
    =\ &\int_\epsilon^\infty \langle u(\lambda x), \psi(x) \varphi(x\lambda) \rangle \lambda^{d - 1} \mathrm{d} \lambda.
    \end{align*}
    Now $u(\lambda x) = \lambda^{a} |x|^a \sum_{j=0}^k b_{jk} \log^j\lambda \log^{k-j} |x|$.
    We therefore obtain a sum of terms
    \begin{align*}
        \int_\epsilon^\infty \langle \rho^a \log^j\rho, \psi(x) \varphi(x\lambda) \rangle \lambda^{d+ a - 1} \log^{k-j} \lambda \mathrm{d} \lambda.
    \end{align*}
    Now we note that
    \begin{align*}
    \Phi(\lambda) = \langle \rho^a \log^j \rho, \psi(x) \varphi(x\lambda) \rangle
    \end{align*}
    is in $\mathscr{S}([0,\infty))$ and therefore we can apply \Cref{prop:log-hg-hadamard-integrals}.
    
    To see that $T$ is $C^\infty$-linear, it suffices to observe that
    \begin{align*}
        \langle T_\epsilon(fu), \varphi\rangle
        = \langle (1-\chi_\epsilon) f u, \varphi\rangle
        = \langle (1-\chi_\epsilon) u, f \varphi\rangle
        = \langle T_\epsilon(u), f\varphi\rangle.
    \end{align*}
\end{proof}

\enter

The regulator $T_\epsilon$ defined above has other nice properties:

\enter

\begin{lem}
Let $T_\epsilon$ be the regulator of \Cref{prop:hadamard-reg-basic-block}.
We use the notation of \Cref{prop:regulator-coefficient-distributions} for the coefficients in the asymptotic expansion of $T_\epsilon[u]$.
\begin{itemize}
    \item[i)] Except for the constant term $C[u]$, all coefficients $t_{k,j}[u]$ are distributions supported at $0$.
    \item[ii)] If $V \subset U$ is another open neighborhood of 0, then $T_\epsilon[u|_{\dot{V}}] = T_\epsilon[u]|_{V}$.
\end{itemize}
\end{lem}

\subsection{Polyhomogeneous distributions}
We now move from the basic building block distribution in \Cref{prop:hadamard-reg-basic-block} to a wider class of distributions.
First we must begin by analyzing homogeneous distributions.

\enter

\begin{defn}
Let $(\Gamma, \mathfrak{s})$ be a  conic manifold.
A distribution $u\in \mathscr{D}'(\Gamma)$ is called \emph{homogeneous}
of degree $\alpha \in \mathbb{C}$ if
\begin{align}
    \mathfrak{s}(\lambda)^* u = \lambda^\alpha u.
\end{align}
The space of $\alpha$-homogeneous distributions on $\Gamma$ is denoted
$\Hg_\alpha(\Gamma)$.
\end{defn}

\enter

\begin{defn}
The \emph{Euler operator} on $\mathbb{R}^d$ is the vector field
$\mathcal{E} = x\cdot \nabla = \sum_{i=1}^d x_i \partial_i$.
\end{defn}

\enter

\begin{prop}
Let $\Gamma \subset \mathbb{R}^d$ be open and conic, then $u \in \mathscr{D}'(\Gamma)$
is homogeneous of degree $\alpha$ if and only if
$(\mathcal{E}-\alpha) u = 0$.
\end{prop}

\begin{proof}
Straightforward computation, see the discussion on pp. 74 -- 75 in \cite{hormanderAnalysisLinearPartial1998}.
\end{proof}

\enter

\begin{prop}[Theorem 3.2.3, \cite{hormanderAnalysisLinearPartial1998}]
Let $u \in \mathscr{D}'(\dot{\mathbb{R}}^d)$ be homogeneous of degree $\alpha$.
If $\alpha$ is not an integer $\leq -d$, then $u$ has a unique extension $\overline{u} \in \mathscr{D}'(\mathbb{R}^d)$ which is homogeneous of degree $\alpha$.
If $P$ is a homogeneous polynomial then $\overline{Pu} = P \overline{u}$, and if $\alpha \neq 1-d$ then $\overline{\partial_j u} = \partial_j \overline{u}$.
The map
\begin{align}
    \mathscr{D}'(\dot{\mathbb{R}}^d) \ni u \mapsto \overline{u} \in \mathscr{D}'(\mathbb{R}^d)
\end{align}
is continuous.
\end{prop}

\begin{proof}
In addition to the proof given by Hörmander (loc.\,cit.\@), an elegant proof can be found in Proposition 3.14 of \cite{bahnsOnshellExtensionDistributions2014}.
\end{proof}

\enter

Let $F(\omega) \in \mathscr{D}'(\mathbb{S}^{d-1})$ and define $\rho(x) = |x|$.
Then $f = F \rho^\alpha \in \mathscr{D}'(\dot{\mathbb{R}}^d)$
defines an $\alpha$-homogeneous distribution on $\dot{\mathbb{R}}^d$.
We provide an explicit extension to $\mathbb{R}^d$.

\enter

\begin{prop}[Proposition 2.2.1 \cite{ortnernorbertDistributionValuedAnalytic2013}]
For $\alpha \in \mathbb{C}$ with $\Re\alpha > -d$, define
$F\cdot \rho^\alpha \in \mathscr{S}'(\mathbb{R}^d)$ by
\begin{align}
    \langle F\cdot \rho^\alpha, \varphi\rangle
    = \int_0^\infty {}_{\mathscr{D}'(\mathbb{S}^{d-1})}\langle F(\omega),
    \varphi(\lambda\omega)\rangle_{\mathscr{D}(\mathbb{S}^{d-1})}
    \lambda^{\alpha} d\lambda
\end{align}
where $\varphi \in \mathscr{S}(\mathbb{R}^d)$.
Then the following holds:
\begin{enumerate}
    \item For $F \in L^1(\mathbb{S}^{d-1})$ and $\Re \alpha > -1$, the
    distribution $F\cdot \rho^\alpha$ coincides with the $L^1$ function
    $F(x/|x|) |x|^\alpha$.
    \item For $F \in \mathscr{D}'(\mathbb{S}^{d-1})$, the distribution valued
    function
    \begin{align}
        \{\alpha \in \mathbb{C}: \Re \alpha > -1\}
        \to \mathscr{S}'(\mathbb{R}^d),\quad \alpha \mapsto F\cdot \rho^\alpha
    \end{align}
    is analytic and can be analytically continued to a meromorphic function
    on $\mathbb{C}$ with at most simple poles at
    \begin{align}
        \Lambda = (-\infty, -d] \cap \mathbb{Z}.
    \end{align}
    For $\alpha \in \mathbb{C}\setminus \Lambda$ we denote the analytically
    continued distribution again by $F\cdot \rho^\alpha$.
    \item For $\alpha \in \mathbb{C}\setminus \Lambda$, the distribution
    $F\cdot \rho^\alpha$ is $\alpha$-homogeneous.
\end{enumerate}
\end{prop}

\enter

\begin{prop}
The analytic continuation has the following explicit representation:
\begin{align}
    \begin{aligned}
    \left\langle\varphi, F \cdot \varrho^\lambda\right\rangle & =\int_1^{\infty}\left\langle\varphi\left(t\omega\right), F(\omega)\right\rangle t^{\lambda+d-1} \mathrm{~d} t \\
    & +\int_0^1\Big(\left\langle\varphi\left(t\omega\right), F(\omega)\right\rangle-\sum_{\substack{\alpha \in \mathbb{N}_0^n \\
    c+d + |\alpha| <0}} \frac{M_\alpha}{\alpha!} \partial^\alpha \varphi(0) t^{a \alpha}\Big) t^{\lambda+d-1} \mathrm{~d} t \\
    & +\sum_{\substack{\alpha \in \mathbb{N}_0^n \\
    c+d + |\alpha| <0}} \frac{M_\alpha}{\alpha!(\lambda+d+|\alpha|)} \partial^\alpha \varphi(0), \quad \varphi \in \mathcal{S}\left(\mathbb{R}^n\right)
    \end{aligned}
\end{align}
\end{prop}

\enter

Observe that
\begin{align*}
\left\langle\varphi\left(t \omega\right), F(\omega)\right\rangle-\sum_{\substack{\alpha \in \mathbb{N}_0^n \\ c+d + |\alpha| <0}} \frac{M_\alpha}{\alpha!} \partial^\alpha \varphi(0) t^{|\alpha|}=O\left(t^{-c-d}\right)
\end{align*}
so that the middle integral term in the analytic continuation is well defined.

\enter

\begin{defn}
Let $\Gamma \subset \mathbb{R}^d$ be open and conic. A distribution $u \in \mathscr{D}'(\Gamma)$ is called \emph{almost homogeneous} of degree $\alpha$
and order $k$ if
\begin{align}
    (\mathcal{E} - \alpha)^{k+1} u = 0.
\end{align}
The space of almost homogeneous distributions of degree $\alpha$ and order $k$
on $\Gamma$ is denoted $\Loghg_{\alpha, k}(\Gamma)$.
\end{defn}

\enter

Almost homogeneous distributions are products of logarithmic and
homogeneous terms.
Hence they are also called \emph{log-homogeneous} distributions.

\enter

\begin{prop}
Let $\alpha \in \mathbb{C}\setminus \Lambda$ and $k\in \mathbb{N}_0$.
Then $u \in \mathscr{D}'(\mathbb{R}^d)$ is almost homogeneous
if and only if
\begin{align}
    u = \sum_{j=0}^k F_j \cdot |x|^\alpha \log^j |x|
\end{align}
for some distributions $F_j \in \mathscr{D}'(\mathbb{S}^{d-1})$.
The summands are defined by
\begin{align}
    F\cdot |x|^\alpha \log^j|x| := \frac{d^j}{d\alpha^j}\left(F\cdot |x|^\alpha\right).
\end{align}
\end{prop}

\enter

The following structure theorem allows us to explicitly
characterize the almost homogeneous distributions.

\enter

\begin{prop}[Cf.\ Proposition 2.5.3, 
\cite{ortnernorbertDistributionValuedAnalytic2013}]
\label{prop:struct-hg-dists}
    Let $\rho(x) = |x|$ and let $\rho^\alpha \cdot \log^j\rho$ be
    the almost homogeneous distributions defined above.
    Define the distribution spaces
    \begin{align}
        &M_\lambda = \{F \in \mathscr{D}'(\mathbb{S}^{d-1})|\ \forall\
        \alpha \in \mathbb{N}_0^d\ \text{with}\ \lambda = -\sum_{i=1}^d \alpha_i, \langle F, \omega^\alpha\rangle = 0\}\\
        & N_\lambda = \{S \in \Hg_{\lambda}(\mathbb{R}^d)| \supp S \subseteq \{0\}\}.
    \end{align} 
    Equip $\Loghg(U) \subset \mathscr{D}'(U)$ and the above spaces
    with the subspace topology.
    Then the following are isomorphisms of topological vector spaces:
    \begin{align}
        &\text{for $\alpha \in \mathbb{C}$},
        &&\begin{aligned}
            &\mathscr{D}'(\mathbb{S}^{d-1})^{\oplus (N+1)} \to \Loghg_{\alpha, N}(\dot{\mathbb{R}}^d)\\
            &(F_0,\ldots, F_N) \mapsto \sum_{j=0}^N F_j \cdot \rho^\alpha \log^j \rho,
        \end{aligned}\\
        &\text{for $\alpha \in \mathbb{C}\setminus \Lambda$},
        &&\begin{aligned}
            &\mathscr{D}'(\mathbb{S}^{d-1})^{\oplus (N+1)} \to \Loghg_{\alpha, N}(\mathbb{R}^d)\\
            &(F_0,\ldots, F_N) \mapsto \sum_{j=0}^N F_j \cdot \rho^\alpha \log^j \rho,
        \end{aligned}\\
        &\text{for $\lambda \in \Lambda$},
        &&\begin{aligned}
             &\mathscr{D}'(\mathbb{S}^{d-1})^{\oplus N} \oplus M_{\lambda}
             \oplus N_{\lambda}\to \Loghg_{\lambda, N}(\dot{\mathbb{R}}^d)\\
            &(F_0,\ldots, F_{N-1}, F_N, S) \mapsto S + \sum_{j=0}^N F_j \cdot \operatorname{fp}_{\alpha=\lambda}\rho^\alpha \log^j \rho.
        \end{aligned}
    \end{align}
\end{prop}

\enter

\begin{defn}Let $U \subset \mathbb{R}^d$ be a neighborhood of 0.
\begin{enumerate}
\item Denote by $\mathcal{I}_N(U, 0) \subset \mathscr{C}^\infty(U)$ the ideal of smooth functions vanishing at $0$ to order $N$. 
\item Denote by $\mathsf{Phg}^{m, k}(U, 0)$ the set of distributions $u \in \mathscr{D}^{\prime}(U)$ such that for any $N \in \mathbb{N}$ there exists an $M$ such that
\[
u=\sum_{j=0}^M u_j+R_N
\]
where $u_j \in \Loghg_{j-m, k}(\mathbb{R}^d)$ (i.e. $u_j$ is almost homogeneous of degree $j-m$ and order $\leq k$ and $R_N \in \mathcal{I}_N(U, 0)$.
\end{enumerate}
\end{defn}

\enter

Finally we can define the Hadamard regularization procedure on polyhomogeneous distributions.
This is a straightforward extension of \Cref{prop:hadamard-reg-basic-block}.

\enter

\begin{prop}
\label{prop:hadamard-polyhmg-dists}
Let $U\subset \mathbb{R}^d$ be a neighborhood of zero, \(\chi \in \mathscr{D}(U)\)
and suppose that \(\chi \equiv 1\) in a neighborhood of \(0\).
Then
\begin{align*}
    \langle T_\epsilon(u), \varphi\rangle
    := \langle (1- \chi_\epsilon) u, \varphi\rangle
\end{align*}
for every $\varphi \in \mathscr{D}(\mathbb{R}^d)$
defines a $C^\infty$-linear regulator
\begin{align*}
    \Phg^{m,k}(\dot{U}, 0) \to \mathsf{PF}_{m+d-1,k+1}((0,1], \mathscr{D}'(\mathbb{R}^d)).
\end{align*}
\end{prop}

\end{appendix}

\medskip

\printbibliography

\end{document}